\newif\iftwocolumn
\twocolumntrue 

\newif\ifdraft
\draftfalse

\ifdraft
    \documentclass[11pt]{article}
    \usepackage{fullpage}
\else
    \iftwocolumn
        \documentclass[aps,pra,nofootinbib,reprint,twocolumn,longbibliography]{revtex4-1}
    \else
        \documentclass[aps,notitlepage,nofootinbib,11pt,tightenlines,longbibliography]{revtex4-1}
    \fi
\fi

\ifdraft
    \usepackage[title]{appendix}
\fi
\newcommand{\eqnref}[1]{Eq.\ (\ref{#1})}
\newcommand{\figref}[1]{Fig.\ \ref{#1}}
\newcommand{\secref}[1]{Sec.\ \ref{#1}}
\usepackage{amsmath,amsfonts}
\usepackage{hyperref}
\usepackage{comment}
\usepackage{framed}
\usepackage{color}
\ifdraft
    \usepackage{subfig}
\else
    \usepackage[caption=false,font=normalsize]{subfig}
\fi

\newcommand{\bra}[1]{\langle #1|}
\newcommand{\ket}[1]{|#1\rangle}
\newcommand{\braket}[2]{\langle #1|#2\rangle}

\usepackage{amsthm}
\newtheorem{lemma}{Lemma}
\newtheorem{thm}{Theorem}

\newtheorem{conjecture}{Conjecture}
\theoremstyle{definition}
\newtheorem{condn}{Condition}
\newtheorem{rmthm}[thm]{Theorem}

\newtheorem{defn}{Definition}

\newcommand{\dist}{\operatorname{dist}}

\makeatletter
\let\tensor\@undefined
\makeatother

\newcommand{\Tr}{\operatorname{Tr}}

\usepackage{tikz}
\newenvironment{tikztensor}{\begin{tikzpicture}[baseline=-0.3em]\tikzset{yscale=-1}}{\end{tikzpicture}}
\input{tikzsettings.pgs}

\newcommand{\cA}{\mathcal{A}}
\newcommand{\cB}{\mathcal{B}}
\newcommand{\cC}{\mathcal{C}}
\newcommand{\cD}{\mathcal{D}}

\newcommand{\cG}{\mathcal{G}}
\newcommand{\cH}{\mathcal{H}}

\newcommand{\cP}{\mathcal{P}}

\newcommand{\cU}{\mathcal{U}}
\newcommand{\cV}{\mathcal{V}}

\newcommand{\wtPsi}{\widetilde{\Psi}}

\newcommand{\tcG}{\widetilde{\mathcal{G}}}
\newcommand{\wtrho}{\widetilde{\rho}}

\usepackage{enumerate}

\ifdraft
   \usepackage{doi}
   \usepackage[numbers,sort&compress,merge]{natbib}
\fi

\begin{document}
\title{Symmetry protection of measurement-based quantum computation in ground states}
\ifdraft
    \author{Dominic V. Else \and Stephen D. Bartlett \and Andrew C. Doherty}
\else
    \author{Dominic V. Else}
    \affiliation{Centre for Engineered Quantum Systems, School of Physics,
    The University of Sydney, Sydney, NSW 2006, Australia}
    \author{Stephen D. Bartlett}
    \affiliation{Centre for Engineered Quantum Systems, School of Physics,
    The University of Sydney, Sydney, NSW 2006, Australia}
    \author{Andrew C. Doherty}
    \affiliation{Centre for Engineered Quantum Systems, School of Physics,
    The University of Sydney, Sydney, NSW 2006, Australia}
\fi
\ifdraft
\maketitle
\fi
\begin{abstract}
The two-dimensional cluster state, a universal resource for
measurement-based quantum computation, is also the gapped ground state of a
short-ranged Hamiltonian. Here, we examine the effect of perturbations to
this Hamiltonian. We prove that, provided the perturbation is sufficiently small
and respects a certain symmetry, the perturbed ground state remains a universal
resource. We do this by characterizing
the operation of an adaptive measurement protocol throughout a suitable 
symmetry-protected quantum phase, relying on generic properties of the phase
rather than any analytic control
over the ground state.
\end{abstract}
\ifdraft
\else
\maketitle
\fi
\section{Introduction}
A quantum computer relies on quantum entanglement to achieve
computational speedups. In the traditional, circuit-based model for quantum
computation, the required entanglement is built up throughout the course of the
computation through application of entangling gates coupling two or more
qubits at a time. 
Alternatively, in the model of
\emph{measurement-based quantum computation} (MBQC)
\cite{raussendorf-prl-2001,briegel_mbqc_natphys},
universal quantum computation is achieved solely through
single-particle operations (specifically, single-particle measurements) on a
fixed entangled resource state, independent of the quantum algorithm being
performed.

Since the initial discovery that the 2-D cluster state
is a universal resource
for MBQC \cite{raussendorf-prl-2001}, much effort has been devoted to
characterizing other universal resource states.
Many of the
universal resource states so far identified
\cite{raussendorf-prl-2001,correlation_space_prl,*correlation_space_pra,chen_tricluster,raussendorf_aklt1,*miyake_aklt,darmawan_phase} have been \emph{projected entangled
pair states} (PEPS) \cite{peps} of small bond dimension.
The tensor network structure of these states facilitates the analysis of
measurements, which might otherwise be an intractable problem.
Another advantage of such states is that under appropriate
conditions \cite{peps_injectivity}, they 
are unique (possibly gapped) ground states of local frustration-free Hamiltonians on
spin lattices. This suggests a method of constructing the resource state by cooling an
appropriate interacting spin system
\cite{verstraete_cirac_2004,jennings_etal_2009}.

However, if we wish to adopt this viewpoint of the resource state for MBQC as
the ground state of a quantum spin system, it would be too restrictive to
confine ourselves to states in which the effect of measurements can be
determined analytically from the tensor-network structure. A generic local Hamiltonian, or even an
arbitrarily small generic local perturbation to a PEPS parent Hamiltonian, will
not have such a property.
Therefore, it is
desirable to develop an understanding of MBQC in ground states of spin systems that does not
rely on analytic control of the ground state. For this reason, there has been an
interest in relating MBQC to forms of \emph{quantum
order} which, as parameters of the Hamiltonian are varied, can disappear only at a
quantum phase transition
\cite{doherty-bartlett-prl-2008,bartlett_renormalization,else_schwarz_bartlett_doherty_symmetry}.

In this paper, we will use such a connection between MBQC and quantum
order to give a precise characterization of
the operation of MBQC in the ground states of a large class of perturbations to
the 2-D cluster model. This will allow us to give a rigorous proof that such
perturbed ground states remain universal resources for MBQC provided that the
perturbation is sufficiently small. Our proof relies in part on an extension
of the 
the relationship introduced in \cite{else_schwarz_bartlett_doherty_symmetry} 
between MBQC and \emph{symmetry-protected topological (SPT)
order}
\cite{gu_wen_2009,chen_gu_wen,schuch},
a form of quantum order characterizing quantum systems which cannot be
smoothly deformed into a product state while a certain symmetry is enforced.
If the perturbation to the 2-D cluster model respects an appropriate 
symmetry, then the perturbed ground state will still possess non-trivial SPT order, and we
will show that this gives us sufficient information about the ground state to characterize the
implications of the perturbation for MBQC. 
Our result therefore
holds independently of any analytic solution for the perturbed ground state.

Our proof of universality is in the same
spirit as \cite{nielsen_cluster_faulttolerance}. There, it was shown that,
whereas measurements on the cluster state simulate quantum circuits,
measurements on a noisy cluster state simulate the same
circuits, but with added noise. Here, our task is 
complicated by the highly correlated nature of the ``errors'' in the resource
state that from result
a change in the Hamiltonian. Nevertheless, we will show how to exploit the additional
structure resulting from SPT order to establish an effective noise model for ground
states of appropriate perturbed cluster models. Therefore, universal quantum computation can be
achieved (for sufficiently small perturbations, corresponding to sufficiently
weak noise in the effective circuit model) by choosing a measurement protocol
which simulates a \emph{fault-tolerant}
quantum circuit. The universality is then a consequence of the \emph{threshold
theorem} \cite{aharonov_ben_or} for fault-tolerant
quantum computation with noisy quantum circuits.

\subsection{Summary of results}
\label{summary_of_results}
Our ultimate goal in this paper is to prove the universality for a MBQC of a
class of perturbations of the 2-D cluster state. However, in order to reach this
goal, most of this paper will be devoted to
a further elucidation of the relationship
between SPT order and MBQC. For simplicity of presentation, we will first
explore this relationship in one-dimensional systems.
It has already
been shown that in a class of quantum phases characterized by SPT order, the structure implied by SPT order
leads to the perfect operation of the identity gate in MBQC
\cite{else_schwarz_bartlett_doherty_symmetry}. Here, we
consider the 1-D cluster model, which lies in the simplest of the SPT phases
considered in \cite{else_schwarz_bartlett_doherty_symmetry},
and characterize the
operation of non-trivial (i.e., not the identity) gates in the presence
of a perturbation which respects the symmetry protecting this SPT phase. We obtain the following:
\begin{thm}[Effective noise model in one dimension]
\label{effective_noise_model}
Consider a measurement protocol which in the exact 1-D cluster model
would simulate a sequence of gates. In the perturbed resource state, the same
measurement protocol simulates the same gate sequence, but with additional
noise associated with each non-trivial gate. So long as the non-trivial gates are
sufficiently separated from each other by identity gates, this effective noise
has no correlations between different time steps, i.e.\ it is
Markovian.
\end{thm}
The proof of Theorem \ref{effective_noise_model} will be divided into two stages. First, in Section
\ref{section_mps} we will establish Theorem \ref{effective_noise_model} for ground states which are \emph{pure
finitely-correlated states} (pFCS), a special case of \emph{matrix-product
states} (MPS).
For such states,
both the manifestations of SPT order \cite{chen_gu_wen,schuch}, and the effect of measurements
\cite{correlation_space_prl,correlation_space_pra} can be understood straightforwardly in terms of the tensor-network
structure. The ideas leading to Theorem \ref{effective_noise_model}
can thus be understood most
directly in this context. Second, in Section \ref{sec_general} we will prove Theorem
\ref{effective_noise_model} for arbitrary ground states within the SPT phase.

The extension of these ideas to the 2-D cluster model will be considered in
Section \ref{section_2d}. We will construct an appropriate symmetry group, such that the
following result is satisfied for symmetry-respecting perturbations.
\begin{thm}[Effective noise model in two dimensions]
\label{effective_noise_model_2d}
Consider a measurement protocol which in the exact 2-D cluster model
would simulate a sequence of gates. In the perturbed resource state, the same
measurement protocol simulates the same gate sequence, but with additional
noise associated with each gate. So long as the non-trivial gates are
sufficiently separated from each other by identity gates, this effective noise has no correlations
between different time steps, or between different gates taking place at the
same
time step, i.e.,\ it is local and
Markovian.
\end{thm}
Combined with the existing results on fault tolerance in the circuit model
\cite{aharonov_ben_or}, Theorem \ref{effective_noise_model_2d} will imply the
main result of this paper:
\begin{thm}
\label{mainresult}
For sufficiently small symmetry-respecting perturbations, the perturbed ground
state remains a universal resource for measurement-based quantum computation.
\end{thm}

\section{The effective noise model construction: finitely-correlated states}
\label{section_mps}
In this section, we will prove our effective noise model result, Theorem
\ref{effective_noise_model}, for a restricted class of ground
states of infinite one-dimensional chains. Specifically, we consider
\emph{pure finitely-correlated states} (pFCS) \cite{fcs,pure_fcs}. 
A pFCS can be considered as
the thermodynamic limit of the translationally-invariant  matrix-product
states (MPS) $\ket{\Psi_N}$ generated by the MPS tensor $A$, on finite chains of $N$ sites with periodic boundary
conditions, e.g.\
\begin{equation}
\label{psi_mps}
\ket{\Psi_5} = 
\begin{tikztensor}
\mps[boundary=periodic]{(0,0)}{5}{$A$},
\end{tikztensor}
\end{equation}
(here and throughout this paper, we use a graphical notation to represent the
contraction of tensors, e.g.\ see \cite{vidal_tensor_network,peps_degeneracy_and_topology}).
The MPS tensor $A$ must satisfy an
additional condition known as \emph{injectivity}, which is related to the
exponential decay of the correlation functions; each of the finite-chain states
$\ket{\Psi_N}$ (for sufficiently large $N$) is then the unique gapped ground
state of a local frustration-free Hamiltonian \cite{perezgarcia_mps}.

We have several motivations for considering this class of ground states. First,
it is widely believed that pFCS capture the essential physics of gapped ground
states of infinite
one-dimensional translationally-invariant spin chains in general. (Note that, although the theorem
regarding the efficient approximation of ground states of finite spin chains by
MPS \cite{mps_faithfully,hastings_area_law} could be regarded as supporting this
belief, we cannot use this theorem to draw any rigorous conclusions for our
purposes here, since it does not hold that the MPS tensor $A$ can be kept
fixed independently of the system size for a constant accuracy.) Second, the ideas leading to our effective
noise model result find their simplest and most physically meaningful
expression in this context.
Finally, the proof presented here will play a dual role in
our paper, as it can also be applied to arbitrary quantum states, provided that they
satisfy a few extra criteria in common with pFCS. Thus, in order to establish
the effective noise model result for
general ground states, which we do in Section \ref{sec_general}, it will suffice to provide a
separate proof of these criteria.

The outline of this section is as follows. We begin in Secs.\
\ref{cluster_1d_no_perturbations} and
\ref{spto_fcs} by
reviewing the properties of the 1-D cluster model and the nature of the SPT
phase in which it is contained. In \secref{sec_dual_state}, we review the results of
\cite{else_schwarz_bartlett_doherty_symmetry} regarding the structure shared by pFCS ground states throughout the
whole SPT phase.
In Sec.~{\ref{sec_dual}},
we prove a key result: the standard adaptive measurement protocol acting on a
ground state in the phase is equivalent to a non-adaptive \textit{dual process}
acting on a `topologically disentangled' version of the ground state, which we
refer to as the \textit{dual state}.
In \secref{sec_cluster_dual}, we give a characterization of the dual process in the case that the
original resource state is the exact cluster state. Finally, in 
\secref{subsec_effective_noise_model}
we exploit the short range of the correlations in pFCS to construct the effective
noise model for any pFCS ground states within the SPT phase, establishing
Theorem \ref{effective_noise_model} for the case of pFCS ground states.

\subsection{The 1-D cluster model in the absence of perturbations}
\label{cluster_1d_no_perturbations}
\begin{figure}
\input{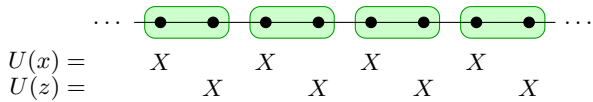}
\caption{\label{fig_1d_cluster} The generators of the on-site $Z_2 \times Z_2$ symmetry
of the 1-D cluster model. Each dot is a qubit, and the shaded areas
constitute two-qubit sites.}
\end{figure}

Here we recall the properties of the 1-D cluster model in the absence of
perturbations. The Hamiltonian is
\begin{equation}
\label{cluster_hamiltonian}
H = -\sum_{i} Z_{i-1} X_i Z_{i+1},
\end{equation}
where $X_i$ denotes the Pauli $X$ operator acting on the $i$-th site, and
similarly for $Z_i$. With appropriate boundary conditions, the system has a
unique ground state (the \emph{cluster state}), and an energy gap of 2,
independent of the system size.

This model has a global $Z_2
\times Z_2$ symmetry generated by the symmetry operations 
$\prod_{i\,\mathrm{even}} X_i$ and $\prod_{i
\, \mathrm{odd}} X_i$.  We
consider this symmetry to be \emph{on-site}, which is to say it acts on
states as a unitary representation $U(g)$ of the symmetry group $G = Z_2 \times
Z_2 = \{ 1, x, y, z \}$ (with $y = xz$), such that $U(g)$ acts 
as $U(g) = [u(g)]^{\otimes N}$, where $N$ is the number of sites
(we group qubits into two-qubit sites in order to ensure this condition is
satisfied; see \figref{fig_1d_cluster}).
As we will see in \secref{spto_fcs}, the cluster model lies in a nontrivial SPT phase with
respect to this symmetry, so that the cluster state cannot be smoothly deformed
into a product state without breaking the symmetry \cite{son-cluster-2011}.

The 1-D cluster state can be represented as a pFCS 
\cite{correlation_space_pra}. For our purposes we will take the MPS tensor
$A_\cC$ to have the form
\iftwocolumn
  \begin{align}
  \label{cluster_mps}
  A_{\mathcal{C}}[++] &= \mathbb{I}, & A_{\mathcal{C}}[+-] &= X, \nonumber \\
  A_{\mathcal{C}}[-+] &= Z, & A_{\mathcal{C}}[--] &= XZ = -iY.
  \end{align}
\else
  \begin{equation}
  \label{cluster_mps}
  A_{\mathcal{C}}[++] = \mathbb{I}, \quad A_{\mathcal{C}}[+-] = X, \quad
  A_{\mathcal{C}}[-+] = Z, \quad A_{\mathcal{C}}[--] = XZ = -iY.
  \end{equation}
\fi
This is expressed with respect to a particular basis for a two-qubit site, where
$\ket{\pm} = (\ket{0} \pm \ket{1})/\sqrt{2}$. Here,
and throughout this paper, we use the notation $A[\psi]$ to refer to
the linear operator obtained from the MPS tensor $A$ by interpreting
\begin{equation*}
\begin{tikztensor}
\mps{(0,0)}{1}{$A$}
\measurementop{(0,1)}{$\psi^{*}$}
\end{tikztensor}
\end{equation*}
as a linear operator (acting on states from the right), where
\begin{equation*}
\begin{tikztensor}
\draw(0,0) -- (0.5,0); \measurementop{(0,0)}{$\psi^{*}$}
\end{tikztensor}
\end{equation*}
denotes the tensor obtained by complex conjugation from the rank-1 tensor
corresponding to the state $\ket{\psi}$.

The MPS representation \label{cluster_state_mps} for the cluster state plays a
crucial role in the correlation space picture
\cite{correlation_space_prl,correlation_space_pra} 
for the operation of the cluster state as a quantum
computational wire \cite{gross-eisert-2010-webs}.
When a projective measurement is performed on a site, giving
the outcome $\ket{\psi}$, this is interpreted as inducing an evolution $A[\psi]$
on a ``correlation system''. In the case of the cluster state, for any
qubit rotation $U$ about the $x$- or $z$-axis, one can find a product basis $\{
\ket{\alpha} \}$ for a two-qubit site, such that
\begin{equation}
A_\cC[\alpha] = B_\alpha U,
\end{equation}
where $B_\alpha$ is an outcome-dependent unitary byproduct operator. This
byproduct can be accounted for by adjusting the basis for
future measurements depending on the outcome of the current one.

\subsection{Symmetry-protected topological order in finitely-correlated states}
\label{spto_fcs}
Here, we will review the results of \cite{chen_gu_wen,schuch} on the
manifestation of SPT order in pFCS, and demonstrate that the 1-D cluster model
indeed lies in a nontrivial SPT phase with respect to the $Z_2 \times Z_2$
symmetry.

Consider some ground state which is invariant under the on-site representation
$U(g) = [u(g)]^{\otimes N}$ of some symmetry group $G$, and which can be
represented as a pFCS, as in \eqnref{psi_mps}. 
The
tensor $A$ can be taken to obey a symmetry condition
\cite{string_order_symmetries,schuch}
\begin{equation}
\label{A_symmetry_condn_nographical}
A[u(g)^{\dagger} \ket{\psi}] = \beta(g) W(g)^{\dagger} A[\psi] W(g),
\end{equation}
where $\beta(g)$ is a one-dimensional linear representation of the symmetry
group $G$, 
and $W(g)$ is a \emph{projective} unitary representation of the
symmetry group $G$. This means that 
\begin{equation}
W(g_1) W(g_2) = \omega(g_1, g_2) W(g_1 g_2)
\end{equation}
for some function $\omega$, called the \emph{factor system} of the projective
representation, which maps pairs of group elements to phase factors. By blocking
sites, we can ensure that $\beta(g) = 1$ (however, for simplicity we will assume
that $\beta(g) = 1$ without blocking).
\eqnref{A_symmetry_condn_nographical} can then be represented graphically as
\begin{equation}
\label{symmetry_fractionalization}
\begin{tikztensor}
\tensor{(0,0)}{$A$}
\tensorconnect{(-1,0)}{(0,0)}
\tensorconnect{(0,0)}{(1,0)}
\tensorconnectvertical{(0,0)}{(0,2)}
\tensor{(0,1)}{\footnotesize $u(g)$};
\end{tikztensor}
\; = \;
\begin{tikztensor}
\tensor{(0,0)}{$A$}
\tensorconnect{(-2,0)}{(0,0)}
\tensorconnect{(0,0)}{(2,0)}
\tensorconnectvertical{(0,0)}{(0,1)}
\tensor[fattening=0.2]{(-1,0)}{\footnotesize $W(g)$}
\tensor[fattening=0.2]{(1,0)}{\footnotesize $W(g)^{\dagger}$}
\end{tikztensor}.
\end{equation}
Observe that
$W(g)$ can be multiplied by a $g$-dependent phase factor without affecting
\eqnref{symmetry_fractionalization}; a set of factor systems related by
such a transformation is referred to as a \emph{cohomology class}. The arguments of
\cite{chen_gu_wen,schuch} show that two such pFCS ground states
correspond to the same cohomology class if and only if they are in the same
symmetry-protected phase. Nontrivial cohomology classes [those not containing
the trivial factor system $\omega(g_1,g_2) = 1$] correspond to phases with
nontrivial SPT order.

As an example, consider the cluster model, and its $Z_2
\times Z_2$ symmetry. The on-site
representation $u(g)$ of the symmetry is generated by
\begin{align}
\label{ug}
u(x) &= X \otimes \mathbb{I}, \\
u(z) &= \mathbb{I} \otimes X.
\label{ug2}
\end{align}
and the MPS tensor is given by \eqnref{cluster_mps}. It can be shown that
the symmetry condition
\eqnref{symmetry_fractionalization} is satisfied with the projective representation $W =
V_\cP$, where $V_\cP$ is the \emph{Pauli representation}
\begin{equation}
\label{pauli_projective}
V_\cP(1) = \mathbb{I}, \quad V_\cP(x) = X, \quad V_\cP(z) = Z, \quad V_\cP(y) = Y.
\end{equation}
This projective representation has nontrivial cohomology class, so that the
cluster model lies in a non-trivial symmetry-protected phase.

\subsection{Symmetry-respecting perturbations to the cluster state}
\label{sec_dual_state}
Suppose we now consider a perturbation to the cluster Hamiltonian
\eqnref{cluster_hamiltonian}, such that the perturbed model still
respects the $Z_2 \times Z_2$ symmetry and admits a pFCS ground state.
Unless the perturbation is large enough
to induce a phase transition, the MPS tensor $A$ corresponding to the perturbed
pFCS ground state should still satisfy the symmetry condition
\eqnref{symmetry_fractionalization}, for some projective representation
$W(g)$
with the same factor system as the Pauli projective representation $V$
[\eqnref{pauli_projective}].

The general form of the MPS satisfying these symmetry conditions was established
in \cite{else_schwarz_bartlett_doherty_symmetry}. Here we will briefly review
the relevant results from \cite{else_schwarz_bartlett_doherty_symmetry}. We observe that the Pauli representation
satisfies a property which we refer to as \emph{maximal non-commutativity}:
\begin{defn}
A projective representation $W(g)$ of an abelian group $G$ is called
\emph{maximally non-commutative} if the subgroup 
$Z_{W}(G) \equiv \{ g \in G : W(g)$ commutes with $W(h)$ for all
$h \in G \}$
(which we can think of as the ``projective centre'' of $G$) is the trivial
subgroup.
\end{defn}
\noindent
Notice that the subgroup $Z_{W}(G)$ is actually determined by the factor system
$\omega$, since $W(g) W(h) = \omega(g,h) W(gh) = \omega(g,h) \omega(h,g)^{-1}
W(h) W(g)$. Furthermore, it is the same for all factor systems within a given
cohomology class. Much of the discussion in this paper can be applied to any SPT
phase characterized by a \emph{finite} abelian symmetry group and a
maximally non-commutative cohomology class.

An important consequence of maximal non-commutativity of a factor system is
\cite{else_schwarz_bartlett_doherty_symmetry}
\begin{lemma}
\label{uniqueness_lemma}
For each maximally non-commutative factor system $\omega$ of a finite abelian
group $G$, there exists a
unique (up to unitary equivalence) irreducible projective representation
with factor system $\omega$. The dimension of this irreducible representation is
$\sqrt{|G|}$.
\end{lemma}
\noindent Specifically, the Pauli representation $V_\mathcal{P}$ of $Z_2 \times Z_2$ is the unique irreducible
projective representation corresponding to its factor system. In general,
throughout this paper, we
will use $V(g)$ to denote the unique irreducible projective representation for
the current factor system. A consequence of Lemma
\ref{uniqueness_lemma} is that,
for a tensor satisfying the symmetry condition \eqnref{symmetry_fractionalization}, 
the bond space decomposes as a tensor product of a $\sqrt{|G|}$-dimensional \emph{protected
subsystem} in which $W(g)$ acts irreducibly as $V(g)$ 
and a \emph{junk
subsystem} in which $W(g)$ acts trivially, i.e.
\begin{equation}
\label{symmetry_decomposition}
W(g) = V(g) \otimes \mathbb{I}
\end{equation}
Thus the tensor $A$ appearing the MPS representation of ground states in the
symmetry-protected phase satisfies the symmetry condition
\begin{equation}
\label{A_symmetry_condition}
\begin{tikzpicture}[yscale=-1,baseline=-1.1cm]
\tensor[height=2]{(0,0)}{$A$}
\tensorconnect[irrep]{(-1,0)}{(0,0)}
\tensorconnect[null]{(-1,1)}{(0,1)}
\tensorconnect[irrep]{(0,0)}{(1,0)}
\tensorconnect[null]{(0,1)}{(1,1)}
\tensorconnectvertical{(0,1)}{(0,3)}
\tensor{(0,2)}{\footnotesize $u(g)$};
\end{tikzpicture}
\; = \;
\begin{tikzpicture}[yscale=-1,baseline=-1.1cm]
\tensor[height=2]{(0,0)}{$A$}
\tensorconnect[irrep]{(-2,0)}{(0,0)}
\tensorconnect[null]{(-2,1)}{(0,1)}
\tensorconnect[irrep]{(0,0)}{(2,0)}
\tensorconnect[null]{(0,1)}{(2,1)}
\tensorconnectvertical{(0,1)}{(0,2)}
\tensor[fattening=0.15]{(-1,0)}{\footnotesize $V(g)$};
\tensor[fattening=0.15]{(1,0)}{\footnotesize $V(g)^{\dagger}$};
\end{tikzpicture}.
\end{equation}
Here we use a thick line
(\begin{tikzpicture}[baseline=-0.5ex]\draw[irrep] (0,0) --
(0.5,0);\end{tikzpicture})
to represent the protected subsystem, and a dashed line
(\begin{tikzpicture}[baseline=-0.5ex]\draw[null] (0,0) --
(0.5,0);\end{tikzpicture})
to represent the junk subsystem.
The protected subsystem enjoys several nice properties for storing and manipulating logical information in a quantum computation, as we now show.  

Suppose we perform a projective measurement on one site in a simultaneous
eigenbasis $\{ \ket{i} \}$ (which is $\ket{{+}{+}}, \ket{{+}{-}}, \ket{{-}{+}},
\ket{{-}{-}}$ for the
$Z_2 \times Z_2$ cluster state symmetry),
and obtain the outcome $\ket{i}$.
Then the
resulting state on the remaining sites is found by replacing the original MPS
tensor at the measured site by
\begin{equation}
\begin{tikzpicture}[yscale=-1,baseline=-1.1cm]
\tensor[height=2]{(0,0)}{$A$}
\tensorconnect[irrep]{(-1,0)}{(0,0)}
\tensorconnect[null]{(-1,1)}{(0,1)}
\tensorconnect[irrep]{(0,0)}{(1,0)}
\tensorconnect[null]{(0,1)}{(1,1)}
\tensorconnectvertical{(0,1)}{(0,2)}
\measurementop{(0,2)}{$i^{*}$}
\end{tikzpicture}
\end{equation}
Now we make use of another consequence of maximal non-commutativity
\cite{else_schwarz_bartlett_doherty_symmetry}:
\begin{lemma}
\label{lemma_that_must_not_be_named}
Let $u(g)$ be a linear on-site representation of a finite abelian symmetry group
$G$, and let $\omega$ be a maximally non-commutative factor system of $G$. Then
for each basis element $\ket{i}$ in a simultaneous eigenbasis $\{ \ket{i} \}$ of
$u(g)$, there exists a group element $g_i$ such that
\begin{equation}
\label{commutation_condition}
\chi_i(g) V(g) = V(g_i) V(g) V(g_i)^{\dagger}, \quad \forall g \in G
\end{equation}
for any projective representation $V(g)$ with factor system $\omega$, where
$\chi_i(g)$ is the scalar representation of $G$ such that $u(g) \ket{i} =
\chi_i(g) \ket{i}$.
\end{lemma}
\noindent For the example of the cluster state symmetry, we have
$g_{++} = 1$, $g_{+-} = x$, $g_{-+} = z$, $g_{--} = y$, as can
readily be verified directly.

As was shown in
\cite{else_schwarz_bartlett_doherty_symmetry}, Lemma
\ref{lemma_that_must_not_be_named} in conjunction with the symmetry condition
\eqnref{A_symmetry_condition}
implies the decomposition $A[i] = V(g_i) \otimes \widetilde{A}[i]$, represented
graphically as
\begin{equation}
\label{mps_tensor_decomposition}
\begin{tikzpicture}[yscale=-1,baseline=-1.1cm]
\tensor[height=2]{(0,0)}{$A$}
\tensorconnect[irrep]{(-1,0)}{(0,0)}
\tensorconnect[null]{(-1,1)}{(0,1)}
\tensorconnect[irrep]{(0,0)}{(1,0)}
\tensorconnect[null]{(0,1)}{(1,1)}
\tensorconnectvertical{(0,1)}{(0,2)}
\measurementop{(0,2)}{$i^{*}$}
\end{tikzpicture}
\; = \;
\begin{tikzpicture}[yscale=-1,baseline=-1.1cm]
\tensorconnect[irrep]{(-1,0)}{(0,0)}
\tensorconnect[null]{(-1,1)}{(0,1)}
\tensorconnect[irrep]{(0,0)}{(1,0)}
\tensorconnect[null]{(0,1)}{(1,1)}
\tensor{(0,1)}{$\widetilde{A}$}
\tensor[fattening=0.2]{(0,0)}{\footnotesize $V(g_i)$}
\tensorconnectvertical{(0,1)}{(0,2)}
\measurementop{(0,2)}{$i^{*}$}
\end{tikzpicture}
\end{equation}
for some tensor $\widetilde{A}$.
Another way of writing this result is that
\begin{equation}
\label{formula_for_A}
\begin{tikzpicture}[yscale=-1,baseline=-1.1cm]
\tensor[height=2]{(0,0)}{$A$}
\tensorconnect[irrep]{(-1,0)}{(0,0)}
\tensorconnect[null]{(-1,1)}{(0,1)}
\tensorconnect[irrep]{(0,0)}{(1,0)}
\tensorconnect[null]{(0,1)}{(1,1)}
\tensorconnectvertical{(0,1)}{(0,3)}
\end{tikzpicture}
\; = \;
\begin{tikzpicture}[yscale=-1,baseline=-1.1cm]
\tensorconnect[irrep]{(-1,0)}{(2,0)}
\draw (0,2) -| (1,0);
\bigcircle{(1,0)}{$\mathcal{V}$}
\tensorconnect[null]{(-1,1)}{(0,1)}
\tensorconnect[null]{(0,1)}{(2,1)}
\tensor{(0,1)}{$\widetilde{A}$}
\tensorconnectvertical{(0,1)}{(0,3)}
\smallfilledcircle{(0,2)}
\end{tikzpicture},
\end{equation}
where we have defined the tensor 
\begin{equation}
\label{Kdefn}
\begin{tikzpicture}[yscale=-1,baseline=-1.1cm]
\tensorconnect[irrep]{(0,0)}{(2,0)}
\draw (0,1) -| (1,0);
\bigcircle{(1,0)}{$\mathcal{V}$}
\tensorconnectvertical{(0,0)}{(0,2)}
\smallfilledcircle{(0,1)}
\end{tikzpicture}
\; = \quad
\sum_i \left(
\begin{tikzpicture}[yscale=-1,baseline=-1.1cm]
\tensorconnect[irrep]{(0,0)}{(2,0)}

\draw ($(-0.2,-\tensorpadding) + (0,1)$) rectangle (0.2,1cm-\tensorpadding+0.4cm);
\draw ($(0,-\tensorpadding) + (0,1) + (0,0.2)$) node {\small $i^{*}$};
\begin{scope}
\tikzset{yshift=0.5cm}
\draw ($(-0.2,-\tensorpadding) + (0,1)$) rectangle (0.2,1cm-\tensorpadding+0.4cm);
\draw ($(0,-\tensorpadding) + (0,1) + (0,0.2)$) node {\small $i$};
\end{scope}
\draw[midarrow] (0,1.6) -- (0,2.1);
\tensorconnectvertical{(0,0)}{(0,1)}
\tensor[fattening=0.2]{(1,0)}{\footnotesize $V(g_i)$}

\end{tikzpicture}
\right).
\end{equation}
Note that, from a quantum circuit perspective, this tensor can also be interpreted as a
unitary controlled operation $\sum_i \ket{i} \bra{i} \otimes V(g_i)$ coupling a
site to an ancilla particle; hence the
choice of notation.
Conversely, any MPS tensor of the form \eqnref{formula_for_A} for some tensor
$\widetilde{A}$ will satisfy the
symmetry condition \eqnref{A_symmetry_condition}. Following
\cite{tensor_network_global_symmetry}, we refer to the tensor $\widetilde{A}$ as
as the \emph{degeneracy
tensor}; and to the tensor of \eqnref{Kdefn}, which is determined entirely by the
symmetry, as the \emph{structural tensor}.

From \eqnref{mps_tensor_decomposition}, we see that, in the correlation space
picture, measuring in a simultaneous eigenbasis $\{ \ket{i}
\}$ leads to an evolution on the protected
subsystem of correlation space given by an outcome-dependent unitary $V(g_i)$; this evolution is determined by the symmetry (hence the same throughout the
SPT phase), and decoupled from the junk subsystem.  Viewing the unitaries $V(g_i)$ as outcome-dependent byproducts of the measurements, which can straightforwardly be accommodated in a deterministic evolution using the standard techniques of measurement-based quantum computation, we say that the
identity gate operates perfectly throughout the SPT phase.
However, the result of measurement in any other basis is not fixed by the
symmetry, and in general leads to the protected subsystem being coupled to the
junk subsystem, so that the operation of other measurement-based gates will not be a robust
property of the symmetry-protected phase.

\subsection{The dual picture for MBQC on a 1D resource state}
\label{sec_dual}
In order to deal with the randomness of measurement outcomes, the measurement
protocol for MBQC with the cluster state needs to be \emph{adaptive}: the outcome of
the measurement on one site will affect the measurement basis on other sites
arbitrarily far away. In analysing the effect of this protocol when acting on a
\emph{perturbed} resource state, we would like to make an
argument based on the locality of the perturbed Hamiltonian, but the non-local
adaptivity of the measurement protocol poses a difficulty.
Therefore, in this section, we develop an alternate characterization of the
effect of the cluster state adaptive measurement protocol acting on a ground
state in the symmetry-protected phase. We will show that this protocol is
equivalent to a \emph{dual process} acting on a
related state, which we call the \emph{dual state}. We will show that this dual process simply consists of a sequence of
unitary interactions between selected sites (those corresponding to the
locations of non-trivial gates) and an ancilla particle, with no adaptivity.

In our discussion of the dual process, we will represent a pFCS ground state on an infinite
chain as a formal tensor network
\begin{equation}
\ket{\Psi} = 
\begin{tikztensor}
\mps[boundary=dots]{(0,0)}{5}{$A$}
\end{tikztensor}.
\end{equation}
This is not, of course, the
mathematically rigorous way to describe pFCS, but we find it useful for
facilitating understanding. In Appendix \ref{dual_mps}, we will discuss how to formulate
similar arguments in the rigorous pFCS framework. Later on (in Section
\ref{sec_general}), we will also be
interested in finite chains; in that case, the arguments of this subsection can
be applied more directly, given appropriate boundary conditions [specifically,
the boundary conditions at the right edge should be as depicted in
\eqnref{mps_with_ends_19}].

The usefulness of the exact 1-D cluster state (with MPS tensor $A_\cC$) as a
quantum computational wire results from the fact that, for each gate $U$ in a certain set,
there exists a basis $\{ \ket{\alpha} \}$ such that
\begin{equation}
\label{clustergate}
\begin{tikztensor}
\mps[style=irrep]{(0,0)}{1}{$A_\cC$}
\measurementop{(0,1)}{$\alpha^{*}$}
\end{tikztensor}
=
\begin{tikztensor}
\tensorconnect[irrep]{(-1,0)}{(2,0)}
\tensor{(0,0)}{$U$}
\tensor{(0,0)}{$U$}
\tensor{(1,0)}{$B_\alpha$}
\end{tikztensor}
\end{equation}
where $B_\alpha$ is the outcome-dependent unitary byproduct operator. When we
measure one site projectively and obtain the outcome
$\ket{\alpha}$, the original MPS tensor $A$ is replaced at the measured site by
\eqnref{clustergate} in the tensor-network description of
the resultant state.

In the case of the exact cluster state, the effect of the byproduct operator can be accounted for by
adjusting the measurement basis for future measurements. This fact turns out to
be closely related to the nontrivial SPT order, as we now demonstrate.
Our discussion relies on the observation that, in the cluster state, the
byproduct operators are Pauli operators. That is to say, it is always the case
that $B_\alpha$ is a scalar multiple of $V(g_\alpha)$ for some $g_\alpha \in Z_2 \times Z_2$.
Hence, we can make use of the symmetry condition [which can be derived from
\eqnref{A_symmetry_condition}]
\begin{equation}
\label{tensor_symmetry_cluster}
\begin{tikzpicture}[yscale=-1,baseline=-0.6cm]
\tensor[height=1]{(0,0)}{$A_\cC$}
\tensorconnect[irrep]{(-2,0)}{(0,0)}
\tensorconnect[irrep]{(0,0)}{(1,0)}
\tensorconnectvertical{(0,0)}{(0,2)}
\tensor{(-1,0)}{$B_\alpha$};
\end{tikzpicture}
\; = \;
\begin{tikzpicture}[yscale=-1,baseline=-0.6cm]
\tensor[height=1]{(0,0)}{$A_\cC$}
\tensorconnect[irrep]{(-1,0)}{(0,0)}
\tensorconnect[irrep]{(0,0)}{(2,0)}
\tensorconnectvertical{(0,0)}{(0,2)}
\tensor{(1,0)}{$B_\alpha$}
\tensor{(0,1)}{$b_\alpha$}
\end{tikzpicture},
\end{equation}
where $b_\alpha = u(g_\alpha)$. Applying this condition repeatedly shows that
the byproduct operator can be displaced arbitrarily far to the right.
In our formal tensor-network picture for an infinite chain, we consider that this process is continued indefinitely, so that
the byproduct operator ``disappears out to infinity'', and is replaced with $b_\alpha$ acting on all
sites to the right of the one on which the measurement took place, i.e.
\iftwocolumn
    \begin{equation}
    \label{adaptive_cluster}
    \input{composite_0.fig}.
    \end{equation}
\else
    \begin{equation}
    \label{adaptive_cluster}
    \input{mps_with_ends_11.fig}
    =
    \input{mps_with_ends_12.fig}.
    \end{equation}
\fi
Hence, whenever we obtain the ``wrong'' outcome for a measurement (i.e.\ the
corresponding byproduct operator $B_\alpha$ is not the identity), we can recover
the ``correct'' resultant state 
by applying the correction $b_\alpha^{\dagger}$ to all the remaining sites on the right
(equivalently, we can simply adjust the measurement basis for measurements on
those sites).

Let us now examine what happens when we perform the same adaptive measurement
protocol on a resource state that is not the exact cluster state. 
Consider a pFCS ground state contained with the same SPT phase as the
cluster state, characterised by the Pauli representation of the group $Z2
\times Z2$. We will keep using the same measurement protocol as for the
exact cluster state. (Our argument could be generalised to any pFCS
ground state contained within any SPT phase characterized by a finite
abelian symmetry group $G$ and a maximally non-commutative cohomology
class, so long as the the adaptive correction appearing in the
measurement protocol takes the same form
as for the cluster state, i.e. application of $u(g_\alpha)$ to the
sites on the right for some group elements $g_\alpha \in G$.) The resource
state is then of the form
\begin{equation}
\label{mps_with_ends_0}
\ket{\Psi} = \, \input{mps_with_ends_0.fig} \; ,
\end{equation}
with the MPS tensor $A$ of the form \eqnref{formula_for_A}. We now repeat the
above argument, in reverse. We make use of the symmetry condition
\eqnref{A_symmetry_condition} in the form
\begin{equation}
\label{tensor_symmetry_cluster_2}
\begin{tikzpicture}[yscale=-1,baseline=-0.6cm]
\tensor[height=2]{(0,0)}{$A$}
\tensorconnect[irrep]{(-1,0)}{(0,0)}
\tensorconnect[irrep]{(0,0)}{(1,0)}
\tensorconnect[null]{(-1,1)}{(0,1)}
\tensorconnect[null]{(0,1)}{(1,1)}
\tensorconnectvertical{(0,1)}{(0,3)}
\tensor{(0,2)}{$b_\alpha^{\dagger}$}
\end{tikzpicture}
\; = \;
\begin{tikzpicture}[yscale=-1,baseline=-0.6cm]
\tensor[height=2]{(0,0)}{$A$}
\tensorconnect[irrep]{(-2,0)}{(0,0)}
\tensorconnect[irrep]{(0,0)}{(2,0)}
\tensorconnect[null]{(-2,1)}{(0,1)}
\tensorconnect[null]{(0,1)}{(2,1)}
\tensorconnectvertical{(0,1)}{(0,3)}
\tensor{(-1,0)}{$B_\alpha^{\dagger}$};
\tensor{(1,0)}{$B_\alpha$}
\end{tikzpicture},
\end{equation}
from which we obtain
\iftwocolumn
    \begin{equation}
    \label{mps_with_ends_4_5}
    \input{composite_1.fig}.
    \end{equation}
\else
    \begin{equation}
    \label{mps_with_ends_4_5}
    \input{mps_with_ends_5.fig}
    \; = \;
    \input{mps_with_ends_4.fig}.
    \end{equation}
\fi
Therefore, we have shown that the process we actually perform, i.e.\ applying the measurement-dependent correction
to the sites on the right of the one measured, is equivalent to a different
process, in which the measurement-dependent correction is applied in the
internal bond space of the MPS, as depicted in the right-hand side of
\eqnref{mps_with_ends_4_5}. In a physical system, of course, we do not have direct
access to the internal bonds of a tensor network state, so we could never
perform the latter process directly; nevertheless, the two are equivalent.

Following the measurement and the adaptive correction, which we think of as
being performed internally, as in the right-hand side of
\eqnref{mps_with_ends_4_5}, the outcome
of the measurement can be ``forgotten'', i.e.\ we describe the resultant state
of the system as the mixture of the
right-hand side of \eqnref{mps_with_ends_4_5} for all possible measurement
outcomes. Without affecting
the reduced state on the remaining unmeasured sites, for notational
convenience we replace this mixture with a coherent superposition, i.e.\
\begin{equation}
\label{mps_replacement}
\input{mps_with_ends_18.fig},
\end{equation}
where we have defined the 
tensor 
\begin{equation}
\label{diamond_defn}
\begin{tikzpicture}[yscale=-1,baseline=-1.1cm]
\tensorconnect[irrep]{(0,0)}{(2,0)}
\draw (0,1) -| (1,0);
\bigcircle{(1,0)}{$\mathcal{V}^{\dagger}$}
\draw (0,1) node[anchor=north west]{$k$};
\tensorconnectvertical{(0,0)}{(0,2)}
\smallfilleddiamond{(0,1)}
\end{tikzpicture}
\; = \quad
\sum_\alpha \left(
\begin{tikzpicture}[yscale=-1,baseline=-1.1cm]
\tensorconnect[irrep]{(0,0)}{(2,0)}

\draw ($(-0.2,-\tensorpadding) + (0,1)$) rectangle (0.2,1cm-\tensorpadding+0.4cm);
\draw ($(0,-\tensorpadding) + (0,1) + (0,0.2)$) node {\small $\alpha^{*}$};
\begin{scope}
\tikzset{yshift=0.5cm}
\draw ($(-0.2,-\tensorpadding) + (0,1)$) rectangle (0.2,1cm-\tensorpadding+0.4cm);
\draw ($(0,-\tensorpadding) + (0,1) + (0,0.2)$) node {\small $\alpha$};
\end{scope}
\draw[midarrow] (0,1.6) -- (0,2.1);
\tensorconnectvertical{(0,0)}{(0,1)}
\tensor{(1,0)}{$B_\alpha^{\dagger}$}

\end{tikzpicture}
\right),
\end{equation}
which we can also interpret as a unitary coupling $\sum_\alpha \ket{\alpha}
\bra{\alpha} \otimes B_\alpha^\dagger = \sum_\alpha \ket{\alpha} \bra{\alpha} \otimes
V(g_\alpha)^{\dagger}$ coupling a site to an ancilla particle. We now
use the $k$ index throughout the paper to distinguish the ``$\mathcal{G}$'''s
resulting from different measurement operations.
(The label $k$ refers to the site at which the measurement is being
performed; we include this label to reflect the dependence on the measurement
basis $\{ \ket{\alpha} \}$ and byproduct operators $B_\alpha$, which will in
general be different for each site at which a measurement is performed.)

Now, using the expression \eqnref{formula_for_A} for the MPS tensor $A$, we can
write
\begin{equation}
\label{A_correct_dual}
\begin{tikzpicture}[yscale=-1,baseline=-1.1cm]
\tensorconnect[irrep]{(-1,0)}{(2,0)}
\tensorconnectvertical{(0,1)}{(0,3)}
\draw (0,2) -| (1,0);
\smallfilleddiamond{(0,2)}
\bigcircle{(1,0)}{$\mathcal{V}^{\dagger}$}
\tensorconnect[null]{(-1,1)}{(0,1)}
\tensorconnect[null]{(0,1)}{(2,1)}
\tensor[height=2]{(0,0)}{$A$}
\draw (0,2) node[anchor=north west]{$k$};
\end{tikzpicture}
\; = \;
\begin{tikzpicture}[yscale=-1,baseline=-1.1cm]
\tensorconnect[irrep]{(-1,0)}{(2,0)}
\draw (0,2) -| (1,0);
\smallfilledcircle{(0,2)}
\tensorconnect[null]{(-1,1)}{(0,1)}
\tensorconnect[null]{(0,1)}{(2,1)}
\tensor{(0,1)}{$\widetilde{A}$}
\tensorconnectvertical{(0,1)}{(0,3)}
\tearDown{(1,0)}{$\cG_k$}
\tearRight{(0,2)}{}
\end{tikzpicture},
\end{equation}
where we have defined the tensor $\cG_k$ (which can also be interpreted as a
unitary coupling between a site and the ancilla particle)
according to
\begin{equation}
\cG_k = 
\label{gkdual}
\begin{tikzpicture}[yscale=-1,baseline=-0.5cm]
\draw[irrep] (-0.5,0) -- (0.5,0);
\draw (-0.5,1) -- (0.5,1);
\draw (0,0) -- (0,1);
\tearUp{(0,1)}{}
\tearDown{(0,0)}{$\cG_k$}
\end{tikzpicture}
\; = \;
\begin{tikzpicture}[yscale=-1,baseline=-0.5cm]
\draw[irrep] (-1.5,0) -- (0.5,0);
\draw (-1.5,1) -- (0.5,1);
\correctionop{(-1,1)}{(-1,0)}{$\mathcal{V}$}
\correctionopdiamond{(0,1)}{(0,0)}{$\mathcal{V}^{\dagger}$}
\draw (0,1) node[anchor=north west]{$k$};
\end{tikzpicture}.
\end{equation}

\begin{figure}
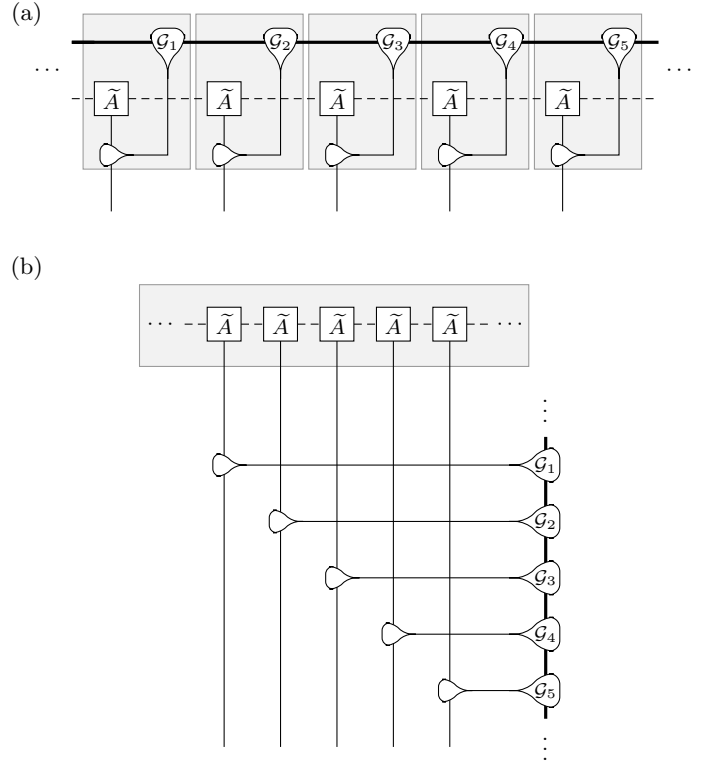

\iftwocolumn
    \input{composite_2.fig}
\else
    \subfloat[]{\input{mps_with_ends_16.fig}}
    \subfloat[]{\input{mps_with_ends_15.fig}}
\fi
\caption{\label{dual_process}(a) After the adaptive
measurement sequence, we can treat the resultant state as having the form shown.
This is equivalent (b) to building $\ket{\Psi}$ from the dual state
$\ket{\widetilde{\Psi}}$ [the shaded box; see \eqnref{dualstate_defn}]
by unitary couplings to an ancilla particle.}
\end{figure}
We are now in a position to define our dual process.
Suppose we perform a sequence of such adaptive measurements at successive
sites, which at each site is described by the insertion of the tensor
\eqnref{diamond_defn},
as in \eqnref{mps_replacement}. There will be a different
coupling $\mathcal{G}_k$ associated with each site $k$.
As shown in Figure \ref{dual_process}, we find that the
original adaptive measurement process, applied to the resource state
$\ket{\Psi}$, is equivalent to a \emph{dual process} applied to the \emph{dual
state} $\ket{\widetilde{\Psi}}$. The dual state $\ket{\widetilde{\Psi}}$ is defined to
be the state built from the degeneracy tensor $\widetilde{A}$, with the
structural tensor discarded:
\begin{equation}
\label{dualstate_defn}
\ket{\widetilde{\Psi}} = 
\begin{tikztensor}
\mps[boundary=dots,style=null]{(0,0)}{4}{$\widetilde{A}$}
\end{tikztensor}.
\end{equation}
The dual process comprises a series of consecutive unitary interactions
$\mathcal{G}_k$ between individual sites $k$ and an ancilla particle.

There are several reasons why this ``dual picture'' is a useful way to
understand the operation of MBQC in one-dimensional ground states. First, the
dual process lacks the long-range measurement adaptivity which is a
characteristic of the original adaptive measurement protocol. Second, the
perfect operation of the identity gate is automatically built in, because,
for sites $k$ at which the
adaptive measurement process at the given site is the one that corresponds in
the exact cluster state to the identity gate,
[i.e. the measurement basis is the simultaneous eigenbasis $\{ \ket{i} \}$ of the symmetry, and the byproduct operators
are $B_i = V(g_i)$, where the $g_i$ are the group elements appearing in
\eqnref{mps_tensor_decomposition}], the corresponding coupling is
trivial, 
$\mathcal{G}_k = \mathbb{I}$.

The final motivation for the dual picture is that the dual state on which it is
based has some physical significance in its
own right, and retains some key
properties of the the original resource state. For example, if the original MPS tensor $A$
generates a pFCS, then so does $\widetilde{A}$, and the respective correlation
lengths obey the inequality $\widetilde{\xi} \leq \xi$ (see Appendix
\ref{dual_mps}). Additionally, in Appendix \ref{kt_connection} we will show how our dual state can
be obtained from the original ground state through a generalization of the unitary that was introduced by
Kennedy and Tasaki \cite{kt,*kt2} to transform the SPT Haldane phase
\cite{haldane1,*haldane2,pollmann-arxiv-2009} into a local symmetry-breaking
phase; this unitary has recently been described
as a ``topological disentangler'' \cite{topological_disentangler}, and in some sense we can think of the dual
state as being a topologically disentangled version of the original resource
state.

\subsubsection{The dual process for initialization and readout in the 1-D
cluster state}
\label{dual_initialization_readout}
Above we only discussed measurement sequences corresponding
to unitary gates in correlation space. A complete scheme for using a 1-D
resource state as a quantum computational wire also includes measurement sequences 
corresponding to \emph{initialization} (i.e.\ discarding the current
state in correlation space and replacing it with a fixed state), and
\emph{readout} (i.e.\ making the state in correlation space available as the
physical state of one qubit). We
now describe briefly how the measurement protocols used on the 1-D cluster state
for these purposes can be accommodated in our framework.

\emph{Initialization}.---The initialization procedure involves measuring a site
in the computational basis $\{ \ket{00}, \ket{01}, \ket{10}, \ket{11} \}$. In
this basis, the MPS tensor $A_{\mathcal{C}}$ for the exact cluster state takes
the form
\begin{align}
\label{Ac_comp_first}
A_{\mathcal{C}}[00] &= \ket{0} \bra{0}, \\
A_{\mathcal{C}}[01] &= \ket{1} \bra{0}, \\
A_{\mathcal{C}}[10] &= \ket{0} \bra{1}, \\
A_{\mathcal{C}}[11] &= -\ket{1} \bra{1}.
\label{Ac_comp_last}
\end{align}
The randomness of measurement outcomes can therefore be accounted for by
applying the appropriate outcome-dependent correction operator in correlation
space following the measurement:
 $B_{00}^{\dagger} = B_{10}^{\dagger} = \mathbb{I}$,
$B_{01}^{\dagger} = B_{11}^{\dagger} = X$.
Since the correction operators are Pauli operators, the above discussion applies
without change.

\emph{Readout}.---The 
standard readout procedure for the cluster state involves measuring the
\emph{second} qubit of a two-qubit site in the computational basis, then
applying an outcome-dependent correction operator to the first qubit, which acts
as the output. Provided
that we are only interested in the final state of the output qubit, this
procedure is equivalent to a coherent correction operator
coupling the two
qubits in the site 
(specifically, it is a controlled-Z gate $C_Z = \ket{0}
\bra{0} \otimes \mathbb{I} + \ket{1} \bra{1} \otimes Z$). Carrying through a
similar argument to that given above for unitary gates,
\secref{sec_dual}, we obtain the same result, but with
the interaction $\mathcal{G}_k$ in the dual process
between the site $k$ in question and
the ancilla particle given by
\begin{equation}
\label{G_readout}
\cG_k = 
\begin{tikzpicture}[yscale=-1,baseline=-0.5cm]
\draw[irrep] (-0.5,0) -- (0.5,0);
\draw (-0.5,1) -- (0.5,1);
\draw (0,0) -- (0,1);
\tearUp{(0,1)}{}
\tearDown{(0,0)}{$\cG_k$}
\end{tikzpicture}
\; = \;
\begin{tikzpicture}[yscale=-1,baseline=-0.5cm]
\draw[irrep] (-1.5,0) -- (0.5,0);
\draw (-1.5,1) -- (0.5,1);
\correctionop{(-1,1)}{(-1,0)}{$\mathcal{V}$}
\tensor{(0,1)}{$C_Z$}
\end{tikzpicture}.
\end{equation}

\subsection{MPS of minimal bond dimension and the dual picture}
\label{sec_cluster_dual}
As an example of the general formalism introduced in \secref{sec_dual}, here we
will examine the form of the couplings $\mathcal{G}_k$
appearing in the dual process [\eqnref{gkdual}],
in the particular case that the resource state is an MPS with bond dimension $D
= \sqrt{|G|}$, where $G$ is the symmetry group characterizing the
symmetry-protected phase.
Given that the
dimension of the protected subsystem is $\sqrt{|G|}$ (by Lemma
\ref{uniqueness_lemma}), this is the
smallest possible value of $D$, and corresponds to the absence of a junk
subsystem (or, more precisely, a junk subsystem of dimension 1). In particular,
the 1D cluster
state is of this type. 
In general, the MPS tensor $A$ for such an MPS must be of
the form
\begin{equation}
A[i] = \widetilde{A}[i] V(g_i),
\end{equation}
where the $\widetilde{A}[i]$ here are \emph{scalars}.
It follows that the dual of such a state is a product state, with each site in
the state $\ket{\phi} = \sum_i
\widetilde{A}[i] \ket{i}$. (We choose the normalization for the MPS tensor $A$ so
that $\braket{\phi}{\phi} = 1$.)
Therefore, the effect of the dual
process acting on the dual state results from a series of
independent interactions of the form
\begin{equation}
\begin{tikzpicture}[yscale=-1,baseline=-0.5cm]
\draw[irrep] (-1,0) -- (1,0);
\draw (-1,1) -- (1,1);
\draw (0,0) -- (0,1);
\tearUp{(0,1)}{}
\tearDown{(0,0)}{$\cG_k$}
\draw (-1,1) node[left] {$\ket{\phi}$};
\end{tikzpicture}.
\end{equation}

We recall that, in the correlation space picture of quantum computational
wires, a quantum state can serve as a resource for executing a unitary gate $U$
if there exists some basis $\{ \ket{\alpha} \}$ such that 
\begin{equation}
\label{resource_U}
A[\alpha] = \beta_\alpha B_\alpha U,
\end{equation}
for some set of unitary byproduct operators $B_\alpha$ and scalars
$\beta_\alpha$. We will now show how this
property manifests itself in the dual picture, for the class of states considered
here. We make use of the representation for the MPS tensor $A$ as
\begin{equation}
\label{A_product_repn}
\begin{tikzpicture}[yscale=-1,baseline=-1.1cm]
\tensor{(0,0)}{$A$}
\tensorconnect[irrep]{(-1,0)}{(0,0)}
\tensorconnect[irrep]{(0,0)}{(1,0)}
\tensorconnectvertical{(0,0)}{(0,3)}
\end{tikzpicture}
\; = \;
\begin{tikzpicture}[yscale=-1,baseline=-1.1cm]
\tensorconnectvertical{(0,1)}{(0,2)}
\tensorconnect[irrep]{(-1,0)}{(2,0)}
\draw (0,2) -| (1,0);
\smallfilledcircle{(0,2)}
\bigcircle{(1,0)}{$\cV$}
\tensor{(0,1)}{$\phi$}
\tensorconnectvertical{(0,1)}{(0,3)}
\end{tikzpicture}.
\end{equation}
It follows that, at a site $k$ measured in the basis $\{ \ket{\alpha} \}$, with
the byproduct operators $B_\alpha$, we have
\begin{align}
\label{gate_dual_representation}
\begin{tikzpicture}[yscale=-1,baseline=-0.6cm]
\tensorconnectleft[irrep]{(0,0)}{(3,0)}
\tensorconnectleft[site]{(1,1)}{(3,1)};
\smalltensor{(1,1)}{$\phi$}
\draw (2,1) -- (2,0);
\tearUp{(2,1)}{$\cG_k$}
\tearDown{(2,0)}{$\cG_k$}
\end{tikzpicture}
\quad &= \quad
\begin{tikzpicture}[yscale=-1,baseline=-0.6cm]
\tensorconnect[irrep]{(0,0)}{(1,0)}
\tensorconnect[irrep]{(1,0)}{(2,0)}
\tensorconnectleft[irrep]{(2,0)}{(3,0)}
\tensorconnectupper[site]{(1,0)}{(1,1)};
\tensor[height=1]{(1,0)}{$A$}
\draw[site] (1,1) -- (2,1);
\draw[site] (2,1) -- (3,1);
\correctionopdiamond{(2,1)}{(2,0)}{$\cV^\dagger$}
\draw (2,1) node[anchor=north west] {$k$};
\end{tikzpicture} \\
\nonumber \\
&= \quad
\begin{tikzpicture}[baseline=-0.6cm,yscale=-1]
\tensorconnectleft[irrep]{(0,0)}{(3,0)}
\tensorconnectleft[site]{(1,1)}{(3,1)};
\smalltensor{(1,1)}{$\varphi$}
\tensor{(2,0)}{$U$}
\end{tikzpicture},
\end{align}
where $\ket{\varphi} = \sum_\alpha \beta_{\alpha} \ket{\alpha}$. (It can be shown
that our choice of normalization ensures that $\braket{\varphi}{\varphi} = 1$.) Here the first
equality follows from \eqnref{A_product_repn} and the definition of $\cG$;
and the second inequality follows by \eqnref{resource_U}.
Thus, we have shown that in the dual picture the gate $U$ simply acts on the ancilla
particle.

Next we will do a similar analysis for the initialization and readout procedures
specific to the 1-D cluster state.

\emph{Initialization}.---
Recall the discussion of initialization in \secref{sec_dual}. We make use of the
form of the MPS tensor $A_{\mathcal{C}}$ in the computational basis, Eqs.\
(\ref{Ac_comp_first}--\ref{Ac_comp_last}), multiplied by the appropriate normalization
factor as discussed above. Thus, for the site $k$ at which initialization takes
place, we find that
\begin{align}
\label{initialization_dual_representation}
\begin{tikzpicture}[yscale=-1,baseline=-0.6cm]
\tensorconnectleft[irrep]{(0,0)}{(3,0)}
\tensorconnectleft[site]{(1,1)}{(3,1)};
\smalltensor{(1,1)}{$\phi$}
\draw (2,1) -- (2,0);
\tearUp{(2,1)}{$\cG_k$}
\tearDown{(2,0)}{$\cG_k$}
\end{tikzpicture}
\quad &= \quad
\begin{tikzpicture}[yscale=-1,baseline=-0.6cm]
\tensorconnect[irrep]{(0,0)}{(1,0)}
\tensorconnect[irrep]{(1,0)}{(2,0)}
\tensorconnectleft[irrep]{(2,0)}{(3,0)}
\tensorconnectupper[site]{(1,0)}{(1,1)};
\tensor[height=1]{(1,0)}{$A_{\mathcal{C}}$}
\draw[site] (1,1) -- (2,1);
\draw[site] (2,1) -- (3,1);
\correctionopdiamond{(2,1)}{(2,0)}{$\cV^\dagger$}
\draw{(2,1)} node[anchor=north west]{$k$};
\end{tikzpicture} \\
\nonumber \\
&=
\begin{tikzpicture}[baseline=-0.6cm,yscale=-1]
\draw[irrep] (0,0) -- (1,0) -- (1,1) -- (2,1);
\tensorconnectleft[site]{(2,1)}{(3,1)}
\smalltensor{(2,0)}{$0$}
\tensorconnectleft[irrep]{(2,0)}{(3,0)}
\tensor{(2,1)}{$\Gamma$}
\end{tikzpicture},
\end{align}
where
$\Gamma = \frac{1}{\sqrt{2}}\left(\ket{00} \bra{0} + \ket{01}\bra{0} + \ket{10}
\bra{1} - \ket{11}\bra{1}\right)$ (thanks to our choice of normalization, we
find that $\Gamma$ is an isometry, i.e.\ $\Gamma^{\dagger} \Gamma =
\mathbb{I}$).
Therefore, applying the measurement sequence for initialization leads to the
ancilla system getting initialized in the state $\ket{0}$, as we would expect.

\emph{Readout}.---
From the definition of the operator $\cG_k$ in the case of sites $k$ at which
readout takes place
[\eqnref{G_readout}], we find that
(here we separate a site into its two constituent
qubits, each denoted by a thick line)
\begin{align}
\label{readout_dual_representation}
\begin{tikzpicture}[yscale=-1,baseline=-0.6cm]
\tensorconnectleft[irrep]{(0,0)}{(3,0)}
\tensorconnectleft[irrep_double]{(1,1)}{(3,1)};
\smalltensor{(1,1)}{$\phi$}
\draw (2,1) -- (2,0);
\tearUp{(2,1)}{$\cG_k$}
\tearDown{(2,0)}{$\cG_k$}
\end{tikzpicture}
\quad &= \quad
\begin{tikzpicture}[yscale=-1,baseline=-0.6cm]
\tensorconnect[irrep]{(0,0)}{(1,0)}
\draw[irrep_double] ($(1,0) + (0,\tensorpadding)$) -- (1,1) -- (3,1);
\tensorconnectleft[irrep]{(1,0)}{(3,0)}
\tensor[height=1]{(1,0)}{$A_{\mathcal{C}}$}
\tensor{(2,1)}{$C_Z$}
\end{tikzpicture}
\\ \nonumber
\\ &= \quad
\begin{tikzpicture}[baseline=-0.6cm,yscale=-1]
\draw[irrep] (-1,0) -- (0,0) -- (0,1.5) -- ($(2,1.5) - (\tensorpadding,0)$);
\tensor[height=2]{(1,0)}{$\ket{I}$}
\tensorconnect[irrep]{(1,0)}{(2,0)}
\tensorconnect[irrep]{(1,1)}{(2,1)}
\end{tikzpicture},
\end{align}
where $\ket{I} = (1/\sqrt{2}) \left(\ket{00} + \ket{11}\right)$ is the canonical
maximally-entangled state. 
Thus the state of the ancilla qubit indeed gets transferred onto the output
qubit.

\subsection{MBQC on a perturbed resource state simulates a noisy quantum circuit}
\label{subsec_effective_noise_model}
In the previous subsection, we saw how measurements on an MPS
of minimal bond dimension correspond to quantum gates.
Now we will consider what happens when we perform the same measurement
sequences on a \emph{perturbed} resource state, assuming that the perturbed state
remains within the same SPT phase.
We will find that measurements on such a perturbed cluster state simulate the
same quantum circuit, but with \emph{noisy} gates. The noise is described by application
of a completely positive, trace preserving (CPTP) noise superoperator following each gate.

\begin{figure}
\subfloat[]{\input{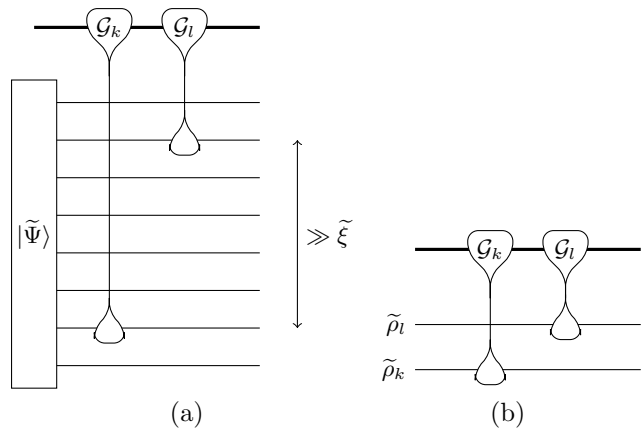}}
\subfloat[]{
\begin{tikzpicture}[yscale=-1]
\draw[irrep] (-1,0) -- (2,0);
\draw (-1,1) -- (2,1);
\draw (-1,1.6) -- (2,1.6);
\draw (-1,1) node[left] {$\wtrho_l$};
\draw (-1,1.6) node[left] {$\wtrho_k$};

\draw (0,1.6)--(0,0);
\tearUp{(0,1.6)}{$\cG_k$}
\tearDown{(0,0)}{$\cG_k$}

\draw (1,1)--(1,0);
\tearUp{(1,1)}{$\cG_l$}
\tearDown{(1,0)}{$\cG_l$}
\end{tikzpicture}
}
\caption{\label{dualsep}(a) As long as the locations of non-trivial gates (sites $k$
and $l$ in this diagram) are separated by a
distance much greater than the correlation length $\widetilde{\xi}$, the reduced
state on those locations will be approximately a product state, and the dual
process then reduces (b) to a sequence of independent interactions.}
\end{figure}

In Section \ref{sec_cluster_dual}, we were able to treat each gate independently in the
case of the unperturbed cluster state because the dual state was a product
state, $\ket{\widetilde{\Psi}} = \ket{\phi}^{\otimes N}$. This will no longer be true once we introduce
perturbations, but we still want to treat gates independently. 
Towards this end, we recall that for a site for which the corresponding sequence is that for the
identity gate, the associated coupling $\mathcal{G}_k$ in the dual
picture between that site and the ancilla particle is trivial. Therefore, such a
site can be
traced out from the beginning without affecting the final state of the output
qubit. That is to say, we only need to consider the reduced state $\wtrho =
\Tr_{\mathrm{trivial\ sites}}
\ket{\widetilde{\Psi}} \bra{\widetilde{\Psi}}$ on the remaining
sites, which are those corresponding to non-identity gates (we refer to these as
the \emph{non-trivial} sites). We are free to
choose our measurement protocol to
ensure that the distance between any two non-trivial sites is much greater than the
correlation length $\widetilde{\xi}$.  For pFCS, it is then straightforward to show that $\widetilde{\rho}$ is approximately a product state $\widetilde{\rho}_{\mathrm{prod}} = \bigotimes_k \widetilde{\rho}_k$ over the non-trivial sites, or more precisely
\begin{equation}
\label{pfcs_inequality}
\| \wtrho - \wtrho_{\mathrm{prod}} \|_1 \leq m f(R),
\end{equation}
where $\|.\|_1$ denotes the trace norm,
$m$ is the number of non-trivial sites, $R$ is the minimum distance between any two
non-trivial sites, and $f(R)$ is a function related to the transfer channel of the
pFCS, behaving asymptotically as $f(R) = \mathcal{O}[\exp(-R/\widetilde{\xi})]$ for large $R$,
where $\widetilde{\xi}$ is the correlation length associated with the pFCS.  

We first consider the case where $\wtrho = \wtrho_{\mathrm{prod}}$ exactly. 
Because $\wtrho$ is then a product state, we find, as in the previous
subsection, that
the dual process acting on the dual state is again effectively a sequence of
independent interactions, this time of the form
\begin{equation}
\label{rhok_evolution}
\begin{tikzpicture}[yscale=-1,baseline=-0.5cm]
\draw[irrep] (-1,0) -- (1,0);
\draw (-1,1) -- (1,1);
\couplingop{(0,0)}{(0,1)}{$\cG_k$}
\draw (-1,1) node[left] {$\wtrho_k$};
\end{tikzpicture}
\end{equation}
(see \figref{dualsep}).
Let us suppose that $\mathcal{G}_k$ results from the
measurement sequence corresponding to a unitary gate
$U_k$. Then, after tracing out the physical site, \eqnref{rhok_evolution}
corresponds to an evolution on the ancilla qubit described by the CPTP map
\begin{equation}
\mathcal{A}_k(\sigma) = \Tr_{\mathrm{physical\;site}}\left(
\mathcal{G}_k (\sigma \otimes \wtrho_k)
\mathcal{G}_k^{\dagger}\right)
\end{equation}
As we saw in \secref{sec_cluster_dual}, in the absence of perturbations to the
cluster state, $\wtrho_k = \ket{\phi}\bra{\phi}$ and $\mathcal{A}_k = \mathcal{U}_k$, where
$\mathcal{U}_k(\sigma) = U_k \sigma U_k^{\dagger}$. In general we can write
$\mathcal{A}_k = \mathcal{E}_k \circ \mathcal{U}_k$, where $\mathcal{E}_k$ is a
noise superoperator for which it is straightforward to show that
\begin{equation}
\label{Ek_inequality}
\|\mathcal{E}_k - \mathcal{I}\|_{\Diamond} \leq \| \wtrho_k -
\ket{\phi}\bra{\phi}\|_1,
\end{equation}
where $\|.\|_{\Diamond}$ is the diamond norm on superoperators
\cite{qc_with_mixed_states}.
The cases when $\mathcal{G}_k$ corresponds to initialization or readout are
analogous. Therefore we have shown (in the case $\wtrho =
\wtrho_{\mathrm{prod}}$) that the measurement protocol on the
perturbed cluster state reproduces the desired quantum circuit, except that each
gate (as well as the initialization and readout steps) is accompanied by some associated
noise. Furthermore, if the perturbation is sufficiently small,
then the reduced states $\widetilde{\rho}_k$ will be close to
$\ket{\phi}\bra{\phi}$ (see Appendix D for the proof), so that the noise
will be weak, in the sense that $\mathcal{E}_k$ is close to the identity
superoperator in the diamond norm.

In the general case, in which $\wtrho$ and $\wtrho_{\mathrm{prod}}$ are not equal,
but are $\epsilon$-close in the trace distance, we just need to observe that the
reduced state of the output qubit following the dual
process can be obtained from $\wtrho$ by application of some CPTP superoperator,
which we call $\mathcal{B}$. From the contractivity property of the trace distance, it follows
that $\| \cB(\wtrho) - \cB(\wtrho_{\mathrm{prod}}) \|_1 \leq \| \wtrho -
\wtrho_{\mathrm{prod}}\|_1 \leq \epsilon$. Therefore, the
effective noisy quantum circuit description correctly describes the final state
of the output qubit up to an accuracy $\epsilon$.
Note that, because the bound \eqnref{pfcs_inequality} depends on the number of non-trivial
gates $m$, it will be necessary to have the separation $R$ scale with $m$ in
order to obtain a fixed accuracy $\epsilon$, but only logarithmically; indeed,
the minimum separation
required to achieve an accuracy $\epsilon$ scales like
$R_{\mathrm{min}}/\widetilde{\xi} = \mathcal{O} [\log(m/\epsilon)]$.

\subsection{Summary of Section \ref{section_mps}}
In Section \ref{section_mps}, we have presented, within the context of pure
finitely-correlated states, the main ideas leading to our effective noise model
construction. Our discussion has hinged around the ``dual state'' which we
associated with each ground state carrying the appropriate SPT order. Loosely
speaking, we can think of the entanglement in SPT-ordered ground states as
comprising ``topological'' and ``non-topological'' components intertwined. The
topological component is fixed throughout the phase and is responsible for the
distinctive characteristics of the SPT phase, such as the degeneracy in the
entanglement spectrum \cite{pollmann-prb-2010}, the diverging localizable entanglement length
\cite{popp-le,venuti-2005-le-string}, and the perfect
operation of the identity gate. One can think of the dual state as being
obtained from the original ground state by a topological disentangler, ``separating out'' the topological
component of the entanglement and leaving only the non-topological component
\cite{topological_disentangler}.

In this paper, the importance of the dual state is due to the following fact,
which we established in \secref{sec_dual}: the cluster state adaptive
measurement protocol,
when applied to an SPT-ordered ground state, couples in a natural way
to the topological component of the entanglement, and the effect
is thus equivalent to a ``dual
process'' (with a simpler structure) acting on the dual state.
This result gives rise to an effective quantum circuit
description describing the outcome of the measurement protocol applied to any
SPT-ordered ground state [e.g.\ see \figref{dualsep}(a)]. The action of non-trivial
gates is determined by an interaction with a single site in the dual state, and
perturbations to the dual state give rise to noisy gates. As long
as the locations of non-trivial gates are sufficiently separated, the reduced
state on the sites relevant for the gate operation will be a product state, and
this corresponds to independent noise acting on each gate in the effective
circuit description.

\section{The effective noise model construction: general ground states}
\label{sec_general}
In this section, we will extend our characterization of the
effective noise model to any ground state within the symmetry-protected phase,
without reference to finitely correlated states.
Instead
of starting from scratch, we will build on the results of Section
\ref{section_mps}, as
follows.
We 
formulate a condition which we believe (on physical grounds)  to be satisfied for any system within the
symmetry-protected phase. We will show that this condition leads to a
construction
for the dual state of any ground state in the SPT phase, independently of the
pFCS formalism. Furthermore, given an exact MPS representation for the dual state
(which always exists, albeit possibly with a bond
dimension exponentially large in the system size), we show that one can construct a
corresponding MPS representation for the original ground state, such that the arguments of 
Section \ref{section_mps} can be applied without significant change.
In order to establish the approximate factorization condition
\eqnref{pfcs_inequality} in the case of
general ground states,
we will show that the dual state is (like the original resource state) the
gapped ground state of a local Hamiltonian, which can be
constructed in a straightforward way from the original Hamiltonian. This will allow us to
establish the approximate
factorization condition \eqnref{pfcs_inequality} without assuming that the
dual state has a pFCS structure.

\subsection{Symmetry-protected topological order and boundary conditions}
\label{spto_boundary_conditions}
Because we are considering general ground states, we can no longer make direct use of the characterization of SPT order
in finitely-correlated states of \cite{chen_gu_wen,schuch}. Instead, we adopt
the perspective in which SPT order is related to the fractionalized edge
modes associated with open boundary conditions \cite{pollmann-arxiv-2009}.
Our discussion will, out of necessity, be physically motivated rather than
mathematically rigorous,
but will suggest the formulation of the precise assumptions under which the rigorous
results of this paper can be proven.

Consider a 1-D
chain with open boundary conditions, with symmetry-respecting interactions 
such that, in the bulk, there is no symmetry-breaking and a finite energy gap
for excitations. It is still possible that the energy gap for \emph{edge}
excitations may be much smaller than the bulk gap (or even zero).
We denote by $\mathcal{P}$ the subspace comprising the
low-lying edge states. If the chain is sufficiently long, we expect that the
gap in the bulk should ensure that the edges are non-interacting, so that $\mathcal{P}$ decomposes
a tensor product of degrees of freedom
associated with the left and right
edges respectively, $\mathcal{P} = \mathcal{P}_l \otimes \mathcal{P}_r$, 
and the restriction $H_{\cP}$ of the Hamiltonian $H$ onto this subspace 
is a sum of terms acting on each edge: $H_\cP = h_l \otimes
\mathbb{I} + \mathbb{I} \otimes h_r$. 

For our purposes, we are interested in how the edge states transform
under the symmetry. Because
$\mathcal{P}$ is spanned by a set of energy eigenspaces, and the Hamiltonian
commutes with the representation of the symmetry $U(g)$, it follows that
$\mathcal{P}$ must be an invariant subspace for $U(g)$.
We write the operation of $U(g)$
restricted to $\mathcal{P}$ as $U_{\mathcal{P}}(g)$. We expect that, for
sufficiently long chains, the symmetry should act independently on the
respective edge states,
i.e.
\begin{equation}
\label{projective_fractionalization}
U_{\mathcal{P}}(g) = V_l(g) \otimes V_r(g).
\end{equation}
By assumption,
$U(g)$, and hence $U_{\mathcal{P}}(g)$, is a linear representation of the symmetry
group $G$. It
follows that $V_l(g)$ and $V_r(g)$ are in general \emph{projective} representations of $G$,
and if $V_l(g)$ has factor system $\omega$, then $V_r(g)$ must have factor
system $\omega^{-1}$. We are free to transform $V_r(g) \to \beta(g) V_r(g)$,
$V_l(g) \to \beta^{-1}(g) V_l(g)$ for any $g$-dependent phase factors $\beta(g)$ without affecting
\eqnref{projective_fractionalization}, but the cohomology class $[\omega]$ is
uniquely determined. Furthermore, it is intuitively clear that
any continuous symmetry-respecting variation in the Hamiltonian cannot change
the cohomology class $[\omega]$, except at a phase transition (where the gap
closes in the bulk, and the left- and right- edge modes need no longer be
non-interacting).
Therefore, we have an alternative characterization of the SPT phase
corresponding to a cohomology class $[\omega]$: it comprises the systems
where a left edge is associated with emergent edge states transforming 
projectively under the symmetry with cohomology class $[\omega]$.
This is a generalization of the well-known observation that
systems in the $SO(3)$-invariant Haldane phase have emergent spin-1/2 degrees of
freedom at the edges \cite{haldane_emergent_edge2,haldane_emergent_edge1}.

In non-trivial SPT phases, the edge interactions $h_l$ and $h_r$ (and therefore
the overall Hamiltonian $H$) will always
have degenerate ground states, due to the fact that non-trivial projective
representations cannot be one-dimensional. On the other hand, we expect, at
least in the case that the
symmetry group $G$ is abelian, that
a non-degenerate ground state can be recovered by introducing terminating 
particles at the left and right edges, transforming projectively under the
symmetry with factor systems $\omega^{-1}$ and $\omega$ respectively (see Figure
\ref{fig_terminating_particles}). This is because,
loosely speaking, these terminating particles can couple to the edge modes, with
the composite system at each edge transforming under a linear representation
(and therefore, in the case of an abelian symmetry group, generically
having a non-degenerate ground state). For example, the ground state of a spin
chain in the Haldane phase can be made non-degenerate through coupling to
spin-1/2 particles at the edges. Conversely, if the
terminating particles do not transform with the cohomology classes
$[\omega^{-1}]$ and $[\omega]$ respectively, then the degeneracy cannot be removed
completely because there is still a non-trivial projective symmetry
transformation at each edge.

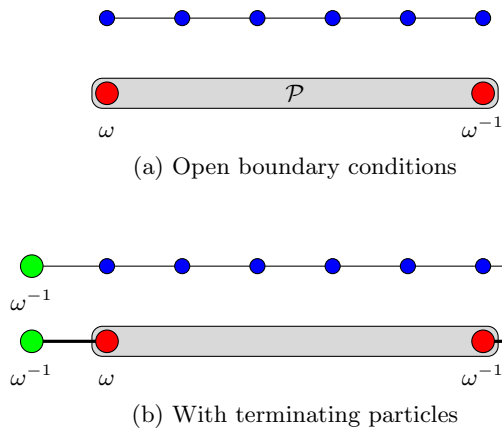
\begin{figure}
\iftwocolumn
\begin{tikzpicture}
\draw (2.5,-2) node{(a) Open boundary conditions};
\draw (0,0) -- (5,0);
\foreach \x in {0,...,5}
{
   \draw[fill=blue] (\x,0) circle(0.1cm);
}


\begin{scope}
\tikzset{yshift=-1cm}
\draw[black,fill=gray!30,rounded corners] (-0.2,-0.2) rectangle (5.2,0.2);
\draw[fill=red] (0,0) circle(0.15cm);
\draw[fill=red] (5,0) circle(0.15cm);
\draw (0,-0.7) node[above] {$\omega$};
\draw (5,-0.7) node[above] {$\omega^{-1}$};
\draw (2.5,0) node {$\mathcal{P}$};
\end{scope}

\begin{scope}
\tikzset{yshift=-3.3cm}
\draw (2.5,-2) node{(b) With terminating particles};
\draw (-1,0) -- (6,0);
\foreach \x in {0,...,5}
{
   \draw[fill=blue] (\x,0) circle(0.1cm);
}
\draw[fill=green] (-1,0) circle(0.15cm);
\draw[fill=green] (6,0) circle(0.15cm);
\draw (-1,-0.7) node[above] {$\omega^{-1}$};
\draw (6,-0.7) node[above] {$\omega$};

\begin{scope}
\tikzset{yshift=-1cm}

\draw[black,fill=gray!30,rounded corners] (-0.2,-0.2) rectangle (5.2,0.2);

\draw[irrep] (-1,0) -- (0,0);
\draw[irrep] (5,0) -- (6,0);

\draw[fill=red] (0,0) circle(0.15cm);
\draw[fill=red] (5,0) circle(0.15cm);
\draw (0,-0.7) node[above] {$\omega$};
\draw (5,-0.7) node[above] {$\omega^{-1}$};
\draw[fill=green] (-1,0) circle(0.15cm);
\draw[fill=green] (6,0) circle(0.15cm);
\draw (-1,-0.7) node[above] {$\omega^{-1}$};
\draw (6,-0.7) node[above] {$\omega$};

\end{scope}
\end{scope}
\end{tikzpicture}
\else
\subfloat[Open boundary conditions]{
\begin{tikzpicture}
\draw (0,0) -- (5,0);
\foreach \x in {0,...,5}
{
   \draw[fill=blue] (\x,0) circle(0.1cm);
}


\begin{scope}
\tikzset{yshift=-1cm}
\draw[black,fill=gray!30,rounded corners] (-0.2,-0.2) rectangle (5.2,0.2);
\draw[fill=red] (0,0) circle(0.15cm);
\draw[fill=red] (5,0) circle(0.15cm);
\draw (0,-0.7) node[above] {$\omega$};
\draw (5,-0.7) node[above] {$\omega^{-1}$};
\draw (2.5,0) node {$\mathcal{P}$};
\end{scope}
\end{tikzpicture}
}
\subfloat[With terminating particles]
{
\begin{tikzpicture}
\draw (-1,0) -- (6,0);
\foreach \x in {0,...,5}
{
   \draw[fill=blue] (\x,0) circle(0.1cm);
}
\draw[fill=green] (-1,0) circle(0.15cm);
\draw[fill=green] (6,0) circle(0.15cm);
\draw (-1,-0.7) node[above] {$\omega^{-1}$};
\draw (6,-0.7) node[above] {$\omega$};

\begin{scope}
\tikzset{yshift=-1cm}

\draw[black,fill=gray!30,rounded corners] (-0.2,-0.2) rectangle (5.2,0.2);

\draw[irrep] (-1,0) -- (0,0);
\draw[irrep] (5,0) -- (6,0);

\draw[fill=red] (0,0) circle(0.15cm);
\draw[fill=red] (5,0) circle(0.15cm);
\draw (0,-0.7) node[above] {$\omega$};
\draw (5,-0.7) node[above] {$\omega^{-1}$};
\draw[fill=green] (-1,0) circle(0.15cm);
\draw[fill=green] (6,0) circle(0.15cm);
\draw (-1,-0.7) node[above] {$\omega^{-1}$};
\draw (6,-0.7) node[above] {$\omega$};

\end{scope}
\end{tikzpicture}
}
\fi
\caption{\label{fig_terminating_particles}(a) The low-lying energy subspace $\mathcal{P}$ of a 1-D chain with
open boundary conditions, in an SPT
phase characterized by the cohomology class $[\omega]$, decomposes as a tensor
product of ``emergent edge modes'' (red) associated with each end, transforming
projectively under the symmetry. (b) The degeneracy can be removed by
coupling terminating particles (green) at each end, leading to an effective
coupling to the edge
modes.}
\end{figure}

Thus, we have arrived at yet another characterization of SPT
order, which we state as a conjecture in the absence of a rigorous proof:
\begin{conjecture}
\label{a_conjecture}
A 1-D chain respecting an on-site representation of an abelian symmetry group
$G$ is in the SPT phase characterized by cohomology class $[\omega]$ if
and only if the following condition is satisfied:
\end{conjecture}
\begin{condn}
\label{terminated_condn}
The finite-chain ground state can be made
non-degenerate and gapped 
by the inclusion of symmetry-respecting interactions
coupling the left and right edges of the chain to terminating particles transforming
projectively under the symmetry, with factor systems $\omega$ and $\omega^{-1}$
respectively.
\end{condn}
\noindent In any case, in the remainder of this section, we will consider
systems satisfying Condition \ref{terminated_condn}. Specifically, all the
results will apply to finite chains with the appropriate edge couplings imposed
to ensure a non-degenerate gapped ground state. This will prove convenient for our
analysis, but the properties of the system in the bulk should not, of course, depend on the boundary
conditions.

Note also that, in the case of a system with the interactions governed by the parent Hamiltonian
of a pFCS [generated by an MPS tensor satisfying the symmetry condition
\eqnref{symmetry_fractionalization}
corresponding to the symmetry-protected phase], Condition
\ref{terminated_condn} can easily be established directly. Furthermore, the
stability theorem of \cite{frustration_free_stability} ensures that Condition
\ref{terminated_condn} remains true for sufficiently small symmetry-respecting perturbations of such
models, regardless of the validity of Conjecture \ref{terminated_condn}.

\subsection{The general construction for the dual state; exact MPS representation of SPT-ordered ground states}
\label{exact_mps}
Recall that in Section {\ref{sec_dual}}, we defined the dual state in the context of pFCS.
Here, we will give an analogous construction for the dual state corresponding to  a general
ground state within a 
 symmetry-protected phase, provided that the phase is characterized by a finite
abelian group $G$ and a maximally
non-commutative cohomology class $[\omega]$. The construction applies to a
finite chain, with the appropriate boundary conditions as
discussed in Section
{\ref{spto_boundary_conditions}}. This construction will then allow us to
express the original ground state as an MPS, with the MPS tensors satisfying an
appropriate symmetry condition.

We
consider a finite chain coupled to terminating particles, such that the overall system
is invariant under the symmetry
$U(g) = V^{*}(g) \otimes [u(g)]^{\otimes N} \otimes V(g)$. 
Here we have taken the right
terminating particle to transform under $V(g)$, the unique irreducible projective
representation with factor system $\omega$; and the left terminating particle under
$V^{*}(g)$ [$V^{*}(g)$ is the operator obtained from $V(g)$ by complex
conjugation of the matrix elements in some
basis; observe that $V^{*}(g)$ is a projective representation of $G$ with factor system $\omega^{-1}$].

\begin{figure}
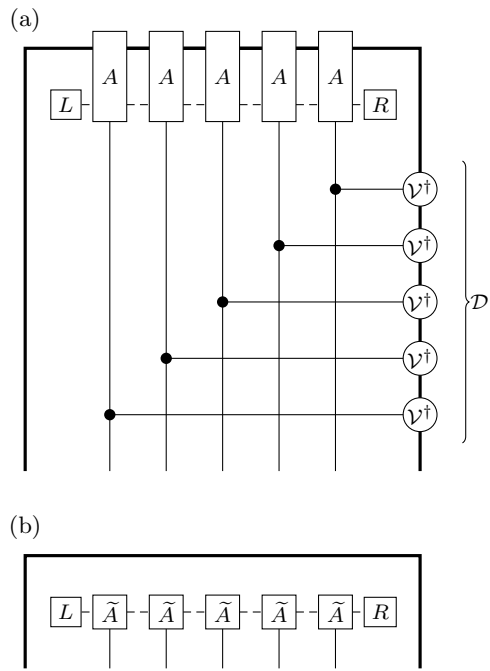

\iftwocolumn
    \input{composite_3.fig}
\else
    \subfloat[]{\input{mps_with_ends_20.fig}}
    \subfloat[]{\input{mps_with_ends_21.fig}}
\fi
\caption{\label{disentangling} The ``topological disentangler'' $\mathcal{D}$ applied
to the original ground state (a) turns it into the dual state (b), leaving the terminating particles
maximally entangled.}
\end{figure}

The natural analogues in the current setting (finite chains, with the specific
choice of boundary conditions) of the pFCS ground states which we considered in Section
\ref{section_mps} are states of the form
\begin{equation}
\label{mps_with_ends_19}
\ket{\Psi} = \input{mps_with_ends_19.fig} \; ,
\end{equation}
for some end vectors $\ket{L}$ and $\bra{R}$, and
where the MPS tensor $A$ satisfies the symmetry condition
\eqnref{A_symmetry_condition} [which can
be shown to ensure the invariance of the state under $U(g)$].
Given the
decomposition \eqnref{formula_for_A} for the MPS tensor $A$, it follows that the dual state
can be obtained from the original ground state by a sequence of unitary
interactions between individual sites and the terminating particle on the right
(see \figref{disentangling}); we can think of the overall unitary transformation
$\mathcal{D}$ as a ``topological disentangler''. 
Specifically, we have
$\mathcal{D} \ket{\Psi} = \ket{\widetilde{\Psi}} \otimes \ket{I}$, where
$\ket{I} = \sum_{k=1}^D \ket{k}\ket{k}$ is the canonical maximally-entangled
state between the terminating particles.

We will now show that, for a general
gapped symmetry-respecting ground state $\ket{\Psi}$ [not necessarily in the MPS
form \eqnref{mps_with_ends_19}], it
remains the case that $\mathcal{D} \ket{\Psi} = \ket{\widetilde{\Psi}} \otimes
\ket{I}$ for some state $\ket{\widetilde{\Psi}}$ on the non-terminating sites; this will serve as the definition of the dual state
$\ket{\widetilde{\Psi}}$
for general ground states.

We observe that the original ground state $\ket{\Psi}$ must be invariant under the global symmetry
operation $U(g)$, i.e.\
\begin{equation}
U(g) \ket{\Psi} = \alpha(g) \ket{\Psi}.
\end{equation}
(Without loss of generality, we can set $\alpha(g) = 1$ by absorbing it into
into the definition of the symmetry\footnote{Specifically, we replace the action
of the symmetry $V(g)$ on the right terminating particle with its
equivalent under a unitary transformation (by Lemma 1 of {Ref.\ \cite{else_schwarz_bartlett_doherty_symmetry}}),
$\alpha^{-1}(g) V(g)$.}.) This implies that the state
$\mathcal{D} \ket{\Psi}$ is invariant under $\mathcal{D} U(g)
\mathcal{D}^{\dagger}$. Let us examine
what this `dual' symmetry looks like.
We observe that
\begin{align}
\cD_1 [u(g) \otimes V(g)] \cD_1^{\dagger}
&= \sum_i \chi_i(g) \ket{i} \bra{i} \otimes V(g_i)^{\dagger} V(g) V(g_i) \\
&= \sum_i \ket{i}\bra{i} \otimes V(g) \label{foobar1} \\
&= \mathbb{I} \otimes V(g)
\label{symrep},
\end{align}
where
\begin{equation}
\cD_1 = 
\begin{tikzpicture}[yscale=-1,baseline=-0.5cm]
\draw[irrep] (-0.5,0) -- (0.5,0);
\draw (-0.5,1) -- (0.5,1);
\correctionop{(0,1)}{(0,0)}{$\mathcal{V}^{\dagger}$}
\end{tikzpicture}
= \sum_i \ket{i} \bra{i} \otimes V(g_i)^{\dagger}
\end{equation}
is the interaction from which $\mathcal{D}$ is built;
to get to \eqnref{foobar1}, we made use of \eqnref{commutation_condition}. From
this, one can show that
\begin{equation}
\label{dual_symmetry}
\mathcal{D} U(g) \mathcal{D}^{\dagger} = V^{*}(g)
\otimes \mathbb{I}^{\otimes N} \otimes V(g).
\end{equation}
It it straighforward to show (using the irreducibility of $V$) that invariance of a state under the right-hand side of
\eqnref{dual_symmetry} implies that it must be of the form $\mathcal{D}
\ket{\Psi} = \ket{\widetilde{\Psi}} \otimes \ket{I}$ for some state
$\ket{\wtPsi}$, as required.

It is now straightforward to construct an appropriate MPS representation for a
general ground state. Indeed, let us consider an MPS representation for the dual
state $\ket{\widetilde{\Psi}}$ of
the form
\begin{equation}
\ket{\widetilde{\Psi}} = \input{mps_with_ends_17.fig}.
\end{equation}
We choose this representation to be \emph{exact}; this
may require the bond dimension to be very large (scaling exponentially in the
system size), but that is of no importance to
us. Then we have
\begin{align}
\mathcal{D} \ket{\Psi} &= \ket{I} \otimes \ket{\widetilde{\Psi}} \\ &=
\input{mps_with_ends_3.fig}.
\end{align}
Now we can apply the inverse transformation $\mathcal{D}^{\dagger}$ to obtain
\begin{multline}
\label{mps_with_ends_9}
\ket{\Psi} = \\ \input{mps_with_ends_9.fig}.
\end{multline}
This is a representation of $\ket{\Psi}$ as an MPS, with each of the shaded
regions corresponding to an MPS tensor $A$ of the
form \eqnref{formula_for_A}, and hence satisfying the symmetry condition
corresponding to the symmetry-protected phase. In addition, we should take note
of the boundary conditions at the right edge. These boundary conditions ensure
that the arguments of Section {\ref{sec_dual}} apply without any need to invoke an
infinite-system limit.

\subsection{The dual state as the ground state of a local Hamiltonian}
\label{sec_dual_hamiltonian}
In the previous subsection, we have constructed the dual state for any ground state
in the symmetry-protected phase. The original ground state is, by assumption, the
gapped ground state of a local Hamiltonian.
In this subsection we will show that this is also true of the dual
state. That is, starting from the original Hamiltonian $H$, we construct another
local Hamiltonian $\widetilde{H}$ for which the dual state is the gapped ground state.

We start by proving a useful fact about the unitary transformation
$\mathcal{D}$ introduced in the previous section:
although it is in general non-local, it maps symmetry-respecting local
observables (i.e.\ those supported on a small set of sites of finite size) 
to local observables.
Indeed, let us consider some local
observable $h$; we will show that $\mathcal{D} h \mathcal{D}^{\dagger}$ is
also local.
For concreteness, we suppose that $h$ acts on two
adjacent sites somewhere in the bulk. Now, observe that
\begin{equation}
\label{hdual}
\begin{tikzpicture}[yscale=-1,baseline=-1cm]
\def \n {2};
\draw[irrep] (-0.6,0) -- (\n+0.6,0);
\foreach \y in {1,...,\n}
{
   \draw (-0.6,\y) -- (\n+0.6,\y);
   \correctionop{(\y,\y)}{(\y,0)}{$\cV^\dagger$}
}

\tensor[height=2]{(0,1)}{$h$}
\end{tikzpicture}
\; = \;
\begin{tikzpicture}[yscale=-1,baseline=-1cm]
\def \n {2};
\draw[irrep] (0.4,0) -- (\n+1.6,0);
\foreach \y in {1,...,\n}
{
   \draw (0.4,\y) -- (\n+1.6,\y);
   \correctionop{(\y,\y)}{(\y,0)}{$\cV^\dagger$}
}

\tensor[height=3]{(3,0)}{$\widetilde{h}$}
\end{tikzpicture},
\end{equation}
where
\begin{equation}
\label{hdual2}
\begin{tikzpicture}[yscale=-1,baseline=-1cm]
\def \n {2};
\draw[irrep] (-0.6,0) -- (0.6,0);
\foreach \y in {1,...,\n}
{
   \draw (-0.6,\y) -- (0.6,\y);
}

\tensor[height=3]{(0,0)}{$\widetilde{h}$}
\end{tikzpicture}
\; = \;
\begin{tikzpicture}[yscale=-1,,baseline=-1cm]
\def \n {2};
\draw[irrep] (-\n-0.6,0) -- (\n+0.6,0);
\foreach \y in {1,...,\n}
{
   \draw (-\n-0.6,\y) -- (\n+0.6,\y);
   \correctionop{(\y,\y)}{(\y,0)}{$\cV^\dagger$}
   \correctionop{(-\y,\y)}{(-\y,0)}{$\cV$}
}

\tensor[height=2]{(0,1)}{$h$}
\end{tikzpicture}.
\end{equation}
By means of Eqs. (\ref{hdual2}) and (\ref{symrep}), it can be verified that if $h$ commutes with the
symmetry, i.e.\
\begin{equation}
\begin{tikzpicture}[yscale=-1,baseline=-1.5cm]
\def \n {2};
\foreach \y in {1,...,\n}
{
   \draw (-2,\y) -- (1,\y);
}
\tensor{(-1,1)}{\footnotesize $u(g)$}
\tensor{(-1,2)}{\footnotesize $u(g)$}

\tensor[height=2]{(0,1)}{$h$}
\end{tikzpicture}
\; = \;
\begin{tikzpicture}[yscale=-1,baseline=-1.5cm]
\def \n {2};
\foreach \y in {1,...,\n}
{
   \draw (-1,\y) -- (2,\y);
}
\tensor{(1,1)}{\footnotesize $u(g)$}
\tensor{(1,2)}{\footnotesize $u(g)$}

\tensor[height=2]{(0,1)}{$h$}
\end{tikzpicture},
\end{equation}
then 

\begin{equation}
\label{hdual_commutation}
\begin{tikzpicture}[yscale=-1,,baseline=-1cm]
\def \n {2};
\draw[irrep] (-2,0) -- (1,0);
\foreach \y in {1,...,\n}
{
   \draw (-2,\y) -- (1,\y);
}

\tensor[height=3]{(0,0)}{$\widetilde{h}$}
\tensor[fattening=0.1]{(-1,0)}{\footnotesize $V(g)$}
\end{tikzpicture}
\; = \;
\begin{tikzpicture}[yscale=-1,baseline=-1cm]
\def \n {2};
\draw[irrep] (-1,0) -- (2,0);
\foreach \y in {1,...,\n}
{
   \draw (-1,\y) -- (2,\y);
}

\tensor[height=3]{(0,0)}{$\widetilde{h}$}
\tensor[fattening=0.1]{(1,0)}{\footnotesize $V(g)$}.
\end{tikzpicture}.
\end{equation}
Since $V(g)$ is an irreducible projective representation,
\eqnref{hdual_commutation}
implies (by Schur's Lemma) that $\widetilde{h}$ acts
trivially on the terminating particle, i.e.
\begin{equation}
\begin{tikzpicture}[yscale=-1,baseline=-1cm]
\label{hdual_nocouple}
\def \n {2};
\draw[irrep] (-1,0) -- (1,0);
\foreach \y in {1,...,\n}
{
   \draw (-1,\y) -- (1,\y);
}

\tensor[height=3]{(0,0)}{$\widetilde{h}$}
\end{tikzpicture}
\; = \;
\begin{tikzpicture}[yscale=-1,baseline=-1cm]
\def \n {2};
\draw[irrep] (-1,0) -- (1,0);
\foreach \y in {1,...,\n}
{
   \draw (-1,\y) -- (1,\y);
}

\tensor[height=2]{(0,1)}{$\widetilde{h}$}
\end{tikzpicture}.
\end{equation}
Now, using Eqs.\ (\ref{hdual}) and (\ref{hdual_nocouple}), we find that
$\mathcal{D} h \mathcal{D}^{\dagger} = \widetilde{h}$, where $\widetilde{h}$
acts on the same two sites as $h$ (see \figref{hcommute}). Thus, although the duality transformation
$\mathcal{D}$ is non-local, we have shown that it maps local symmetry-respecting operators to
local operators, as promised. The exception is operators $h$ at
the left edge, which act non-trivially on the left terminating particle; in that case,
the above argument breaks down, but we can observe directly from the structure
of $\mathcal{D}$ that $\mathcal{D} h \mathcal{D}^{\dagger}$ is supported on the
union of the support of $h$ and the right terminating particle. For operators $h$
acting non-trivially on the \emph{right} terminating particle, the argument must be
adjusted, but the conclusion that $\mathcal{D} h \mathcal{D}^{\dagger}$ is
supported on the support of $h$ still holds.

\begin{figure}
\iftwocolumn
\begin{tikzpicture}[yscale=-0.75, xscale=0.75,baseline=-1.5cm]
\draw[irrep] (-2,-1) -- (7,-1);
\foreach \x in {0,...,5}
{
   \draw (-2,\x) -- (7,\x);
   \correctionop{(\x,\x)}{(\x,-1)}{\footnotesize $\cV^{\dagger}$}
}
\tensor[height=2]{(-1,2)}{$h$}
\begin{scope}
\tikzset{yshift=8cm}
\draw (2.5,-2) node {\large $=$};
\draw[irrep] (-2,-1) -- (7,-1);
\foreach \x in {0,...,5}
{
   \draw (-2,\x) -- (7,\x);
   \correctionop{(\x,\x)}{(\x,-1)}{\footnotesize $\cV^{\dagger}$}
}
\tensor[height=2]{(6,2)}{$\widetilde{h}$}
\end{scope}
\end{tikzpicture}
\else
\begin{tikzpicture}[yscale=-0.75, xscale=0.75,baseline=-1.5cm]
\draw[irrep] (-2,-1) -- (7,-1);
\foreach \x in {0,...,5}
{
   \draw (-2,\x) -- (7,\x);
   \correctionop{(\x,\x)}{(\x,-1)}{\footnotesize $\cV^{\dagger}$}
}
\tensor[height=2]{(-1,2)}{$h$}
\end{tikzpicture}
=
\begin{tikzpicture}[yscale=-0.75, xscale=0.75, baseline=-1.5cm]
\draw[irrep] (-1,-1) -- (7,-1);
\foreach \x in {0,...,5}
{
   \draw (-1,\x) -- (7,\x);
   \correctionop{(\x,\x)}{(\x,-1)}{\footnotesize $\cV^{\dagger}$}
}
\tensor[height=2]{(6,2)}{$\widetilde{h}$}
\end{tikzpicture}
\fi
\caption{\label{hcommute}
From Eqs.\ (\ref{hdual}) and (\ref{hdual_nocouple}), we get the
pictured equality. This shows that $\mathcal{D}h = \widetilde{h} \cD$, or
equivalently $\cD h \cD^{\dagger} = \widetilde{h}$.}
\end{figure}
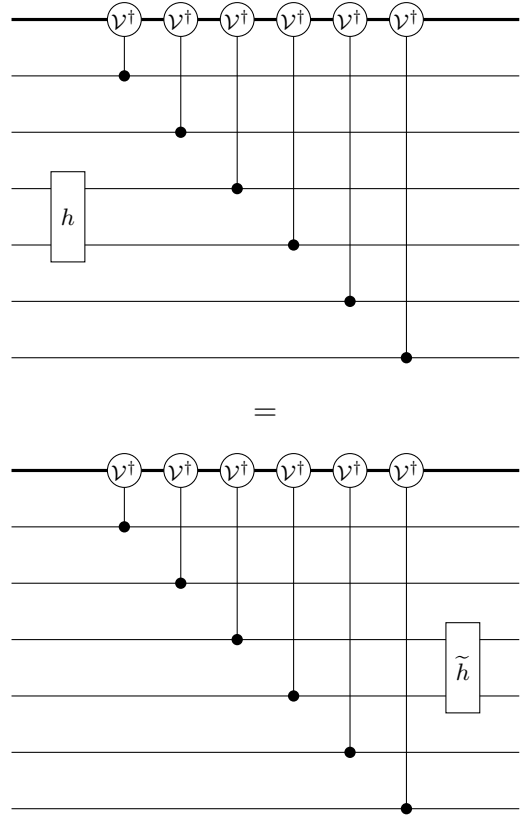

We are now in a position to construct the Hamiltonian for which the dual state
$\ket{\widetilde{\Psi}}$ is
the gapped ground state. We observe that $\mathcal{D} H \mathcal{D}^{\dagger}$
has $\ket{I} \otimes \ket{\widetilde{\Psi}}$ as its gapped ground state;
however, it includes terms acting non-trivially on the terminating particles. We define
a Hamiltonian acting only on the intermediate sites according to $\widetilde{H} =
\bra{I} \mathcal{D} H \mathcal{D}^{\dagger} \ket{I} \equiv \mathcal{F}(H)$; by the
locality result proven above, each local interaction term in $H$ corresponds to a local term in
$\widetilde{H}$ supported on the same set of sites. It can be
shown that $\ket{\widetilde{\Psi}}$ is the unique ground state of $\widetilde{H}$, and
that the gap is at least as large as that of $\mathcal{D} H
\mathcal{D}^{\dagger}$, or equivalently $H$. 

\subsection{The factorization condition for general ground states}
\label{sec_factorization_general}
Recall that the other condition that needed to be satisfied in order to apply
the arguments of Section \ref{section_mps} for general ground states was that the factorization
condition for the reduced density operator $\wtrho$ on the non-trivial sites in
the dual state,
\begin{equation}
\wtrho \approx \bigotimes_k \wtrho_k \equiv \wtrho_{\mathrm{prod}}
\end{equation}
should be satisfied when the non-trivial sites are sufficiently separated from each
other.
Recall that, for the case of pFCS, one can prove the bound
\begin{equation}
\label{pfcs_bound}
\| \wtrho - \wtrho_{\mathrm{prod}} \|_1 \leq m f(R),
\end{equation}
with $m$ the number of non-trivial sites, and $f(R)$ a function scaling asymptotically as $f(R) =
\mathcal{O}[\exp(-R/\widetilde{\xi})]$, where $\widetilde{\xi}$ is the
correlation length in the dual state. We conjecture that \eqnref{pfcs_bound}
should be a general property of all gapped ground states of local Hamiltonians. However, we have
only been able to rigorously prove the weaker bound
\begin{equation}
\label{weaker_general_bound}
\| \wtrho - \wtrho_{\mathrm{prod}} \|_1 \leq m d^{2m} f(R),
\end{equation}
where $f(R)$ is as before, and $d$ is the dimension of the Hilbert space at each
site; see Appendix \ref{appendix_factorization} for the proof.

Note that if we assume only the weaker bound \eqnref{weaker_general_bound}, then the separation
between non-trivial sites will need to scale more rapidly with the number of gates $m$; we
find that the minimum separation $R_{\mathrm{min}}$ required for an accuracy
$\epsilon$ scales like 
\begin{equation}
R_{\mathrm{min}}/\widetilde{\xi} =
\mathcal{O}(m) + \mathcal{O}[\log(1/\epsilon)].
\end{equation}
This still implies that the number of measurements need scale only as a
polynomial in
the number of non-trivial gates.

\subsection{Nonzero temperature}
\label{sec_finite_temperature}
The formulation of the dual state as the ground state of a dual Hamiltonian
extends naturally to nonzero temperature: under the topological disentangler
$\mathcal{D}$, the thermal state of the original
Hamiltonian $H$ maps to the thermal state of a dual Hamiltonian $\widetilde{H}$.
Furthermore, it can be shown
that an appropriate adaptive measurement protocol acting on
the thermal state of the original Hamiltonian is equivalent to a non-adaptive
dual process (of the same form as in the zero-temperature case),
acting on the thermal state of the dual Hamiltonian.

However, it does not appear possible to construct a Markovian effective noise model
for nonzero temperature using the same techniques as for zero temperature.
The reason is that our arguments were based on the assumption that the
reduced state $\widetilde{\rho}_k$ on each of the non-trivial sites in the dual
state does not differ greatly from its value in the dual of the unperturbed
resource state.
This is indeed the case for small local perturbations to the Hamiltonian (as we prove
in Appendix \ref{appendix_smallperturbations}), but it need not be true for
nonzero temperature. For example,
consider the one-dimensional Ising model, with Hamiltonian
\begin{equation}
-\sum_i Z_i Z_{i+1} + Z_1
\end{equation}
(we include the $Z_1$ term to select out a unique ground state). In this model,
it can be shown (e.g.\ using the transfer matrix method) that the reduced state on a single spin changes
discontinuously as soon as the temperature is switched on (this is closely related to the
disappearance of the magnetic order in the 1-D Ising model at nonzero
temperature).
Given the structure of the dual Hamiltonian as discussed in
Appendix \ref{kt_connection},
there is good reason to believe that it will exhibit a
similar phenomenon.

The difficulty of treating thermal states in our framework should
not be surprising, as the dual process has the perfect operation of the identity
gate built in, whereas the cluster model is not expected to have a long-range
identity gate at nonzero
temperature. On
the other hand, there exists a measurement protocol for a
\emph{three}-dimensional cluster model
which retains the perfect operation of the identity gate at sufficiently small
nonzero temperatures \cite{raussendorf_long_range}.
Therefore, if the dual process description could be extended to
measurement protocols such as this one, then it might be expected that the dual
Hamiltonian would possess an ordered phase that persists at nonzero temperature,
such that the local reduced state varies continuously with
temperature up to the phase transition.

\section{Two-dimensional systems and fault tolerance}
\label{section_2d}
The equivalence we demonstrated in Sections
\ref{section_mps} and \ref{sec_general}, between MBQC on perturbed resource states
and noisy quantum circuits, opens up the possibility of exploiting the results in
the literature on fault-tolerant quantum computation with noisy quantum circuits.
Here, we will extend the results of the previous sections to
the 2D cluster model, which, unlike the 1D models considered previously, is a
universal resource for quantum computation. We will again find that, provided
the perturbation to the Hamiltonian respects a certain symmetry, MBQC using the
perturbed ground state as a resource is equivalent to a noisy quantum circuit. We will show
that the noise in this effective circuit description has no correlations in time
(as in the previous section), nor any correlations in space. This reduction to
local, Markovian noise will allow us
to invoke the threshold theorem to deduce that, provided the perturbation
respects the symmetry and is sufficiently small, the perturbed ground state
remains a universal resource for MBQC.

It should be emphasized that, although we make use of the theory of
fault-tolerant quantum computation, our final result cannot be described as a
fault tolerance result for MBQC, since it applies only to symmetry-respecting
perturbations, and we assume noiseless operation of the
measurement protocol.

\subsection{The `quasi-1D' model}
\label{sec_quasi_1d}
Here, we make a first attempt at generalizing the 1D results to a 2D
model which is universal for quantum computation. The ground state of the model
we introduce here
is not strictly a universal resource for MBQC unless we allow non-single-qubit measurements; 
however the discussion here will serve as a
stepping stone for consideration of the 2D cluster model in
\secref{sec_2d_cluster}.

\begin{figure}
\input{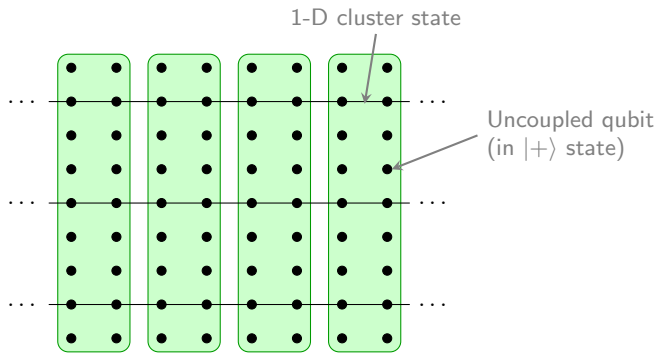}
\caption{\label{cluster_state_2d}  The first step in generalizing the 1D results
to 2D models involves consideration of a
`quasi-1D' model, which consists of a 1D cluster Hamiltonian acting on each of
$N$ qubit chains arranged in the vertical dimension, as well as a term
favouring the $\ket{+}$ states on the uncoupled qubits.
The model has a $(Z_2 \times
Z_2)^{\times N}$ symmetry, arising from the $Z_2 \times Z_2$ symmetry associated with
each of the $N$ chains.
We can treat this model as `quasi-1D' by defining our sites (shown as green
shaded areas) so that they span the vertical dimension.}
\end{figure}

In the absence of perturbations, the 2D model we consider involves $N$ uncoupled 1D cluster states arranged in
the second dimension, as shown in \figref{cluster_state_2d}.
The Hamiltonian acting on each chain is simply the 1D cluster
Hamiltonian. For generality we also assume the existence
of some uncoupled qubits, each with an associated term $-X$ in
the Hamiltonian (i.e.\ the ground state is $\ket{+}$). In order to treat this 2D
model within the same framework which we have developed for 1D systems, we will
consider an entire $N_v \times 2$ block (where $N_v$ is the extent in the vertical
direction) to be a single `site', as shown in \figref{cluster_state_2d}(a); hence we can consider the lattice to comprise a
1D chain of such `sites'. The unperturbed ground
state, which we denote $\ket{\Psi_{\mathcal{Q}}}$, then has an MPS representation which is essentially a
tensor product of several copies of the 1D cluster state MPS representation,
with a correlation system comprising $N$ qubits.
Each chain contributes a separate $Z_2 \times Z_2$
symmetry, so that the model is invariant under a symmetry group $G = (Z_2 \times
Z_2)^{\times N} = \{ (g_1, \cdots, g_N) | g_1, \cdots, g_N \in Z_2 \times Z_2 \}$.
The projective representation of this symmetry in correlation space is the
$N$-qubit generalization of the Pauli representation, namely
\begin{equation}
\label{N_qubit_pauli}
V\bigl( (g_1, \cdots, g_N) \bigr) = V_\cP(g_1) \otimes \cdots \otimes
V_\cP(g_N),
\end{equation}
where $V_\cP$ is the single-qubit Pauli representation of $Z_2 \times Z_2$,
given by
\eqnref{pauli_projective}. It can be checked that this projective representation
is maximally noncommutative, and because it has dimension $2^N = \sqrt{|G|}$, it
must be the unique irreducible projective
representation corresponding to its factor system (by Lemma
\ref{uniqueness_lemma} from \secref{sec_dual_state}).

\begin{figure}
\input{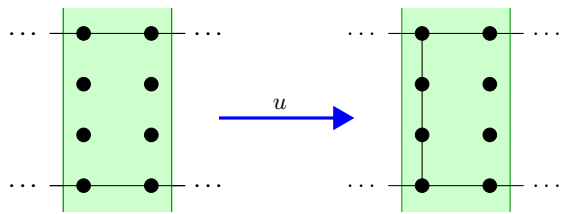}
\caption{\label{cluster_state_2d_5} In the `quasi-1D' model, two-qubit gates in
correlation space cannot be done with single-qubit measurements. However, after
applying controlled-Z gates between neighbouring qubits in order to couple two chains, an
entangling gate can be performed in
correlation space by means of
single-qubit measurements.}
\end{figure}

Now, it is easy to see that, where $\mathcal{S}$ is the set of gates which can be
executed in correlation space by a single-site measurement (up to Pauli
byproducts) in the 1D cluster state, we can execute any
tensor product
\begin{equation}
s_1 \otimes s_2 \otimes \cdots \otimes s_N, \quad s_1, \cdots, s_N \in \mathcal{S}
\end{equation}
(up to Pauli byproducts) in correlation space by a single-site measurement in our
2D model; we just do the corresponding measurements on each chain separately. We
can also find a measurement basis for a columnar site which induces entangling
gates between two qubits in correlation space; however, this measurement basis
clearly cannot correspond to single-qubit measurements, since the two chains
would then remain uncoupled.
For reasons that will become clear when we consider the relation of the
present model to the 2-D cluster state in \secref{sec_2d_cluster}, we will only consider
entangling gates between nearest-neighbour qubits in correlation space, for which we construct the
measurement basis in a particular way, as follows.

We define the on-site unitary $u$, which involves applying controlled-Z
gates between neighbouring qubits to turn our original resource state
$\ket{\Psi_{\mathcal{Q}}}$ into another graph state
$\ket{\Psi_{\mathcal{Q}}^{\prime}}$ in which the two chains of
interest are coupled, as shown in Figure \ref{cluster_state_2d_5}. If 
$A[\cdot]$ is the MPS tensor
for $\ket{\Psi_{\mathcal{Q}}}$ at the given site, then $A^{\prime}[\cdot] = A[u^{\dagger}(\cdot)]$ is the MPS tensor
for $\ket{\Psi_{\mathcal{Q}}^{\prime}}$. Using the measurement sequences described in
\cite{raussendorf_et_al_2003}, it can be shown that there exists a measurement
basis $\{ \ket{\alpha} \}$ for a columnar site, corresponding to
\emph{single-qubit} measurements, such that $A^{\prime}[\alpha] = B_\alpha U$,
where $U$ is an entangling two-qubit gate, and the $B_\alpha$ are
outcome-dependent Pauli byproducts. It follows that this two-qubit gate can be
performed in correlation space (up to the same Pauli byproducts) by measuring in
the basis $\{ u^{\dagger} \ket{\alpha} \}$.

From the above considerations, we see that the model we are discussing can be considered as a generalization of the
1-D cluster state in which $N$ qubits can be propagated in correlation space, acted on
by entangling gates between nearest neighbour qubits as well as single-qubit gates.
In the presence of symmetry-respecting perturbations to the Hamiltonian,
the arguments of Sections \ref{section_mps} and
\ref{sec_general} can still be applied for any finite $N$. However, if we
want to exploit the locality of the perturbation in the vertical direction as
well as the horizontal, we need to make some additional arguments.
First, we observe that (by Lemma \ref{uniqueness_lemma}) the protected subsystem of
correlation space (which corresponds to the ancilla system appearing in the dual
picture of MBQC)
will have dimension $2^N$, and
by identifying the action of the symmetry within the protected subsystem with
\eqnref{N_qubit_pauli}, we can decompose the protected subsystem into $N$ qubits, one
associated with each chain.

Our argument now hinges on two observations.
First, the dual Hamiltonian of which the dual state is the gapped
ground state, as constructed in \secref{sec_dual_hamiltonian}, is in fact a sum of interactions that are local on the original
two-dimensional lattice. Second, the unitary couplings $\cG_k$ appearing in the
dual process, which \emph{a priori} could couple an entire columnar site to the entire
$N$-qubit ancilla system, in fact acts trivially outside an appropriately
localized area (see \figref{tcG}). 
These observations both follow from the form of the interaction
\begin{equation}
\mathcal{D}_1 = \sum_i \ket{i}
\bra{i} \otimes V(g_i)
\end{equation}
between a columnar site and the ancilla system.
(Recall that $\mathcal{D}_1$ and its inverse appeared in the development of the
dual picture in \secref{sec_dual}, as well as in the construction of the duality
transformation $\cD$ from which the dual Hamiltonian $\cH$ was obtained in
\secref{sec_dual_hamiltonian}.)
It is easily seen that in the present quasi-1D setup, $\mathcal{D}_1$
simply comprises a number of applications of
the corresponding operator $\mathcal{D}^{(1)}_1$ for the one-dimensional cluster
chain (see \figref{cluster_state_2d_7}). 

\begin{figure}
\includegraphics{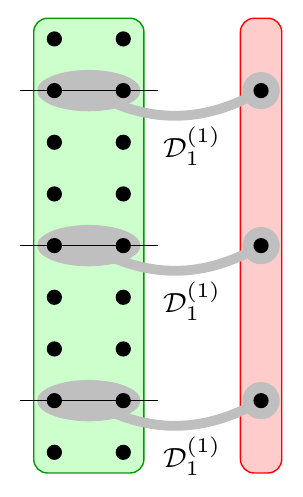}
\caption{\label{cluster_state_2d_7}The unitary operator $\mathcal{D}_1$, which
couples a columnar site and the $N$-qubit ancilla system.}
\end{figure}

\begin{figure}
\subfloat[Single-qubit gate]{\includegraphics{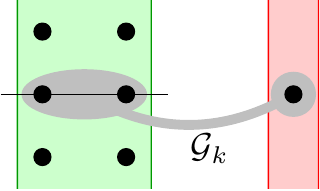}}
\hspace{0.7cm}
\subfloat[Two-qubit gate]{\includegraphics{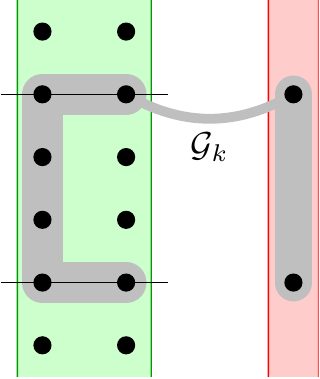}}
{\caption{\label{tcG} The unitary couplings $\cG_k$
appearing in the dual process corresponding to single- and two-qubit gates.}}
\end{figure}

Now, consider a quantum circuit comprising a sequence of gates, and let
$Q_k$ be the sets of physical (not ancilla) qubits acted on by the corresponding
couplings $\tcG_k$ in the dual process. Just as in the one-dimensional case, we
expect that if $R \equiv \min_{k_1, k_2}
\mathrm{dist}(Q_{k_1}, Q_{k_2})$ is much larger than the correlation length
$\widetilde{\xi}$ for the dual state,
then the reduced state $\Tr_{(\bigcup_k Q_k)^c}
\ket{\widetilde{\Psi}}\bra{\widetilde{\Psi}}$ on $\bigcup_k Q_k$ should be approximately a product state over the
$Q_k$'s.
Thus, arguing as in the one-dimensional case (see
\secref{subsec_effective_noise_model}), we find that performing
the measurement sequence on a perturbed resource state 
    corresponds to a noisy quantum circuit, with the noise described
by a noise superoperator $\mathcal{E}_k$ following each gate. Furthermore,
$\mathcal{E}_k$ acts non-trivially only on the same qubits that were acted on by
the
corresponding gate in the
original noiseless quantum circuit. 
The strength of the noise, as given by $\|
\mathcal{E}_k - \mathcal{I}\|_{\Diamond}$, is determined by the deviation (in
the trace norm) of the reduced
density operator on $Q_k$ from its unperturbed value [see \eqnref{Ek_inequality} in
\secref{subsec_effective_noise_model}], which should be small for
small perturbations.

Let us now estimate the required scale-up in the size of the resource state. We
only consider in detail the case of local quantum circuits (i.e.\ containing only gates
acting between nearest-neighbour qubits). As in the one-dimensional
case (\secref{sec_factorization_general}), according to the rigorous factorization result proved in
Appendix \ref{appendix_factorization}, 
the minimum separation $R_{\mathrm{min}}$ between any of the $Q_k$'s required for an accuracy
$\epsilon$ scales like
\begin{equation}
R_{\mathrm{min}}/\widetilde{\xi} = \mathcal{O}(m) + \mathcal{O}[\log(1/\epsilon)].
\end{equation}

The required scale-up can be expressed in terms of $R_{\mathrm{min}}$, as
follows. First, we must ensure that, at each time step,
all non-trivial gates are separated
by a distance of at least $R_{\mathrm{min}}$. This leads to a scale-up by a factor of $\sim
R_{\mathrm{min}}$ in the number of time steps.
Then,
the buffering between horizontal locations at which nontrivial gates take
place implies another factor
of $R_{\mathrm{min}}$ scale-up. Hence, the total scale-up factor is $s \sim
R_{\mathrm{min}}^2$.
On the other hand,
if the quantum circuit that we want to simulate is
not already local, then translating it into a local circuit will introduce
additional overhead (still scaling at worst polynomially in the number of
qubits in the quantum circuit).

\subsection{The 2D cluster model}
\label{sec_2d_cluster}
Now we will return to the model we are actually interested in: the 2D cluster
model on a square lattice. Investigations of the
effect on this model of perturbations 
\cite{doherty-bartlett-prl-2008,transitions_computational,klagges_constraints,cluster_fate} have demonstrated a variety
of results depending on the perturbation. Here, we will focus on perturbations
respecting an appropriate
symmetry. When this symmetry is enforced, the cluster model lies in a robust
SPT phase, within which the identity gate is protected and the effective noise
model construction of this paper can be applied.

In order to achieve our goal,
we will establish an equivalence between the 2D cluster model
and a `quasi-1D' model of the type considered in the previous section. The basic idea is
to define a (local) duality transformation $\mathcal{U}$ (not the same as the duality
transformation $\mathcal{D}$ which we have considered previously) which relates
the two models.
Specifically, we define
\begin{equation}
\mathcal{U} = \prod_{(i,j) \in L} (CZ)_{ij},
\end{equation}
where $(CZ)_{ij}$ is the controlled-Z gate acting on qubits $i$ and $j$, and the
product is over an appropriate set $L$ comprising nearest-neighbour pairs of
qubits. By an appropriate choice of $L$, we can ensure that 
 applying $\mathcal{U}$ to the 2D cluster Hamiltonian turns it into a model of the
 type we considered in the previous section.

Now, let $H$ be a perturbation to the 2-D cluster Hamiltonian. Then $\mathcal{U}
H \mathcal{U}^{\dagger}$ is a perturbation to the quasi-1D model, and the
arguments of the preceding section can be applied provided that the perturbation
respects the appropriate symmetry.
Furthermore, the result (in terms of statistics of measurement
outcomes) of performing the adaptive measurement
protocol described in Section \ref{sec_quasi_1d} on the ground state of $\cU H \cU^{\dagger}$, involving measuring the
observables $\hat{o}$, must be the same as the
effect of performing the same protocol on the ground state of $H$, but measuring
the observables $\mathcal{U}^{\dagger} \hat{o} \mathcal{U}$. We will now examine
in detail this corresponding measurement protocol for perturbations of the 2-D
cluster model.

In the quasi-1D resource state, there is a set of `redundant' qubits which
never need to be measured. It turns out to be convenient to assume, however,
that we do measure those qubits, in the $z$ basis, and that we do this before
any other measurements.
We observe that all the
measurements that are performed on the quasi-1D state (after applying
controlled-Z gates
to couple chains where we want to perform a two-qubit gate) are all single-qubit,
and are either in the $z$ basis, i.e.\ measuring $Z$, or in the $x$-$y$
plane, i.e.\ measuring $\sigma_\theta = (\cos \theta) X + (\sin \theta) Y$ for
some angle $\theta$. Hence,
the corresponding observables to measure in the 2D cluster state are either of the form $Z_j$ (for some qubit
$j$), or
\begin{equation}
\sigma_{\theta_j}^{(j)} \prod_{k \in
\mathcal{N}_j} Z_k
\end{equation}
for some qubit $j$ and angle $\theta_j$, and where $\mathcal{N}_j$ is some set of neighbouring
redundant qubits. But, since we measured the redundant qubits first, they are
all now in
eigenstates of $Z$. Therefore, labelling the measured values of $Z$ on the
redundant qubits by $\{ z_j \}$, we see that measuring $\cU \hat{o}_j \cU^{\dagger}$ is
equivalent to measuring $\sigma_{\theta_j}^{(j)} \left(\prod_{k \in
\mathcal{N}_j} z_k\right)$, which in turn is equivalent to measuring
$\sigma_{\theta_j}^{(j)}$ and reinterpreting the measurement outcomes based on
the value of $\prod_{k \in \mathcal{N}_j} z_k$.
Therefore, we have shown that the measurement protocol on the 2D cluster state
can be implemented using
only single-qubit measurements and adaptivity. It can be checked that the
measurement protocol so constructed is essentially the same as the usual one
for the 2D cluster state on a square lattice, which is described, e.g.\ in
\cite{raussendorf-prl-2001,raussendorf_et_al_2003}.

\begin{figure}
  \subfloat[Horizontal]{
  \includegraphics[scale=0.6]{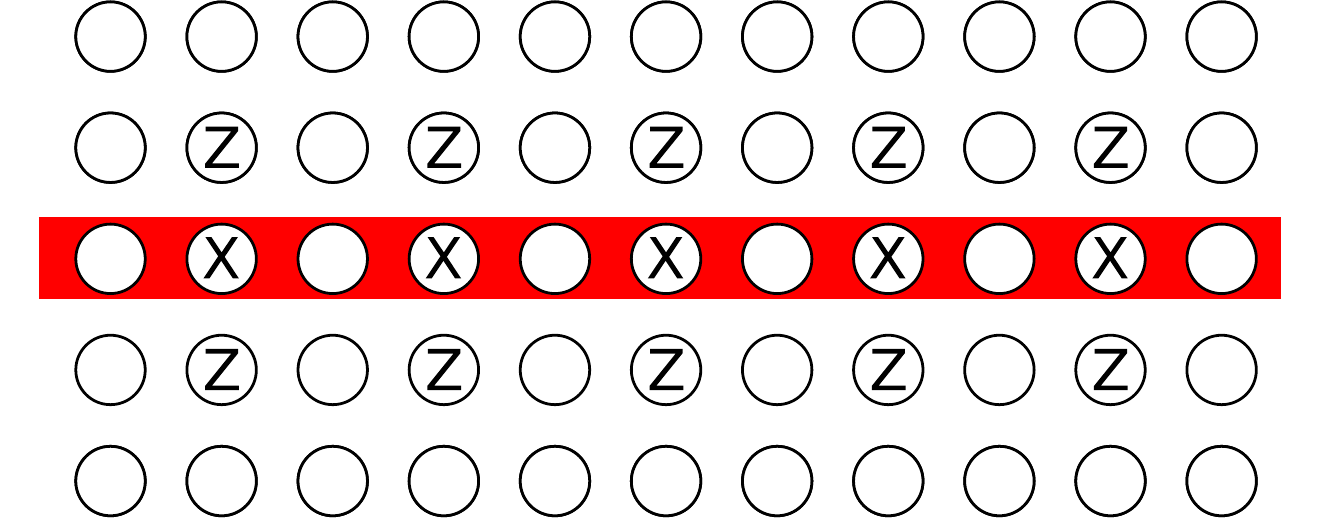}
  }
  \hspace{5cm}
  \subfloat[Diagonal]{
  \includegraphics[scale=0.6]{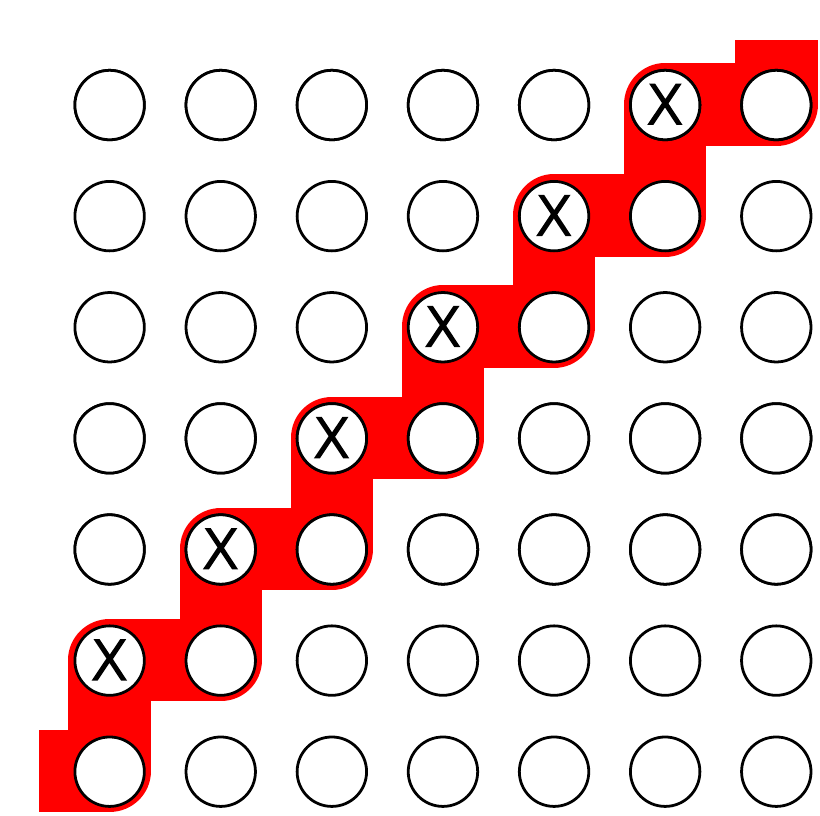}
  }
\caption{\label{clusterstrips} 
Two possible layouts for the 1D chains (red lines) on a 2D lattice.
Each layout
is associated with a measurement protocol for MBQC on the 2D cluster state, and
with a symmetry group (a representative generator of which is shown in each case).
For any layout, we can construct an effective noise
model corresponding to performing the associated measurement protocol on a perturbed
cluster state, provided that the perturbation respects the associated symmetry.
}
\end{figure}

Finally, let us discuss the required symmetry. The duality transformation
$\mathcal{U}$ can be used to relate the $(Z_2 \times Z_2)^{\times N}$ symmetry
which protects the quasi-1D model to a corresponding one in the 2D cluster
model.
The form of the generators of the latter symmetry is shown in
Figure \ref{clusterstrips}(a).
Let us remark that we can also make similar arguments in the case that the 1D
chains are arranged on the 2D square lattice in an unconventional way, for example
diagonally \cite{chung_string,doherty-bartlett-prl-2008} as shown in Figure \ref{clusterstrips}(b). The advantage of the diagonal layout is that
the symmetry [one of the generators of which is depicted in Figure
\ref{clusterstrips}(b)]
takes a particularly simple form, due to the fact that every non-chain qubit
neighbours an
even number of chain qubits, and so the $Z$'s that would normally appear on
non-chain qubits all cancel.
In particular, this symmetry
commutes with an especially simple and physically meaningful
perturbation, namely a uniform magnetic field in the $x$ direction, i.e.\ $V = B \sum_i
X_i$. (The effect of such a perturbation has been studied numerically in
\cite{aqocm,cluster_fate}; the SPT cluster phase persists up to a first-order phase transition at
$|B| = 1$.)

\subsection{Perturbed ground states are universal resources}
\label{sec_invoking_threshold}
Let us summarize the conclusions which we obtain from the
considerations in Sections \ref{sec_quasi_1d} and \ref{sec_2d_cluster} by
stating them as a theorem. In combination with the threshold theorem of
fault-tolerant quantum computation, this theorem will allow us to deduce that sufficiently
small symmetry-respecting perturbations to the 2-D cluster Hamiltonian retain
ground states which can serve as universal resources for MBQC.

We consider perturbations to the exact cluster Hamiltonian
$H_{\mathcal{C}}$ on the 2-D square lattice, which we can take to be a sum of local commuting terms,
with an energy gap to the first excited state of 2. Suppose now we consider a
perturbed Hamiltonian $H = H_{\mathcal{C}} + V$,
and $V$ is a perturbation of the form
\begin{equation}
V = \sum_{u \in \Lambda} V_u,
\end{equation}
where $\Lambda$ is the set of all lattice sites, and each $V_u$ is an
interaction term supported on the set $\mathcal{B}(u,r)$ of
sites within some fixed distance $r$ (more generally,
interactions decaying exponentially with distance would not present an obstacle to
our arguments).
We define the local strength of the perturbation by
\begin{equation}
J \equiv \max_u \|V_u\|.
\end{equation}
The cluster Hamiltonian $H_{\mathcal{C}}$ belongs to a class of Hamiltonians
for which it has been shown
\cite{topological_stability,*topological_stability_short} that the gap is
stable to local perturbations, i.e.\ there exists a
threshold $\eta > 0$ (depending only on $r$),
such that the gap of the perturbed Hamiltonian is at least $1$, provided
that $J \leq \eta$.

Let us assume that the perturbation $V$ respects an appropriate symmetry group,
constructed according to the procedure described in \secref{sec_2d_cluster}
(such as the one of the symmetry groups depicted in \figref{clusterstrips}).
Suppose then we want to use the perturbed ground state $\ket{\Psi}$ to simulate a local quantum
circuit containing $N$ qubits, $T$ time steps, and $m$ gates, with the gates drawn from the gate
set $\mathcal{S}$ comprising single-qubit rotations, a two-qubit entangling
gate (as constructed in \secref{sec_quasi_1d}), and the non-unitary RESTART gate (which
corresponds to the reinitialization of a qubit). We obtain the result

\begin{thm}
\label{noise_model_theorem}
Provided $J \leq \eta$,
we can find
an appropriate measurement protocol on the ground state
$\ket{\Psi}$ such that the final reduced state on the output qubits is
$\epsilon$-close in the trace norm to the outcome of the quantum circuit, with
added noise.
In each time step $t$ of the equivalent circuit process,
the appropriate gates are applied, followed by a noise process described by a superoperator
$\mathcal{E}_t$.
This superoperator can be written as a tensor product
$\mathcal{E}_t = \bigotimes_A \mathcal{E}_{t,A}$, where the product is over
`locations', i.e.\ sets of qubits coupled by a gate in the time step
$t$ (each qubit not coupled by a gate in the time step $t$ also counts as a
location, but $\mathcal{E}_{t,A} = \mathcal{I}$ in that case). Thus, the noise has no correlations in space (other than those due to gates
acting between qubits) or time.
Furthermore, the noise operator $\mathcal{E}_{t,A}$ at each location and time 
is close to the identity superoperator in the
diamond norm:
\begin{equation}
\label{Et_bound}
\| \mathcal{E}_{t,A} - \mathcal{I} \|_{\Diamond} \leq c J,
\end{equation}
for some constant $c$ (dependent only on $r$).
The number of qubits measured $n$ satisfies
\begin{equation}
n = NT \left\{\mathcal{O}(m) +
\mathcal{O}[\log(1/\epsilon)]\right\}.
\end{equation}
\end{thm}
\begin{proof}
The only aspect that we have not previously discussed is the bound
\eqnref{Et_bound}. Following the same argument as in the one-dimensional case
(\secref{subsec_effective_noise_model}), we find [using the analogue of
\eqnref{Ek_inequality}] that the deviation $\| \mathcal{E}_{t,A} - \mathcal{I}
\|_{\Diamond}$ is
bounded above by $\Delta_X \equiv \| \wtrho_X - \wtrho_{X,0} \|_1$, where $X$ is
the set of qubits in the 2-D lattice that affect the operation of the gate in
question, and $\wtrho_X$ and $\wtrho_{X,0}$ are the reduced
states on $X$ of the perturbed dual state $\ket{\widetilde{\Psi}}$ and the
unperturbed dual state $\ket{\widetilde{\Psi_0}}$ respectively. Physically, it
should be clear that $\Delta_X$ will be small for small perturbations;
in Appendix \ref{appendix_smallperturbations} we demonstrate that, so long as $J \leq \eta$, the inequality $\Delta_X \leq
cJ$ holds for some constant $c$ depending only on $r$.
\end{proof}

Now that we have shown that perturbations in the Hamiltonian correspond to
noisy quantum circuits, we can invoke the threshold theorem of
fault-tolerant quantum computation
\cite{aharonov_ben_or,aliferis_long_range,aharonov_long_range}. For our
purposes, the most suitable version is Theorem 13 of Ref.\
\cite{aharonov_ben_or}, which we can state as follows:
\begin{thm}
\label{aharonov_ben_or_theorem}
Let us assume a noise model as described in Theorem \ref{noise_model_theorem}.
Then there exists a threshold $\eta^{\prime} > 0$ and a constant $\alpha$ such that, so long as 
$\|\mathcal{E}_{t,A} - \mathcal{I}\|_{\Diamond} \leq \eta^{\prime}$
for all $A$,$t$, then the following propeties hold. For any $\epsilon > 0$, and
any local quantum circuit $\mathcal{C}$ made from gates drawn from
$\mathcal{S}$ (with $N$ qubits, $T$ time steps, and $m$ gates),
there exists another local circuit $\mathcal{C}^{\prime}$ with gates drawn from
$\mathcal{S}$, such
that $\mathcal{C}^{\prime}$ \emph{with} noise produces the same result (in terms
of the probability distribution for the final readout, and up to an
error $\epsilon$) as $\mathcal{C}$ \emph{without} noise. The scale-up
factors for the number of qubits, the number of time steps, and the number of
gates are all bounded  by
$\mathrm{(const.)} \times \log^{\alpha}(m/\epsilon)$.
\end{thm}

Combining Theorems \ref{noise_model_theorem} and \ref{aharonov_ben_or_theorem},
we obtain:
\begin{framed}
\begin{rmthm}
\label{mainthm}
Consider the perturbed model $H = H_{\mathcal{C}} + V$ as described above.
Then there exists some threshold $\eta^{\prime\prime}  =
\min\{\eta,\eta^{\prime}/c\} > 0$ (depending only on $r$) with the following property. Provided that $J < \eta^{\prime\prime}$, then for any local
quantum circuit $\mathcal{C}$ (with $N$ qubits, $T$ time steps, and $m$ gates),
with gates drawn from $\mathcal{S}$,
we can find 
an appropriate measurement protocol on the perturbed ground state $\ket{\Psi}$
such that the result is equivalent (in terms of the probability distribution for
the final readout, 
and up to an error $\epsilon$ which can be made arbitrarily small) to the outcome of
the original quantum circuit. As $m \to \infty$ with $\epsilon$ held fixed, 
the number of measured qubits $n$ satisfies
\begin{equation}
n \leq NT \times \mathcal{O}\left(m \log^{3\alpha} m\right).
\end{equation} 
\end{rmthm}
\end{framed}
This is sufficient to show that the perturbed ground states remain universal
resources, which is Theorem \ref{mainresult} as stated in
\secref{summary_of_results}; it is the main result of this paper.

\section{Conclusion}
In this paper, we have developed a framework to characterize
the effectiveness of measurement protocols for MBQC with SPT-ordered ground
states of quantum spin systems. This has allowed us to prove the universality
for MBQC of the ground states of perturbed versions of the 2-D cluster
Hamiltonian, provided that the perturbation is sufficiently small and respects
an appropriate symmetry.

The type of SPT order that we have presented here is that present in
one-dimensional systems, which is related to a nontrivial factor system (also
known as a 2-cocycle).
It is for
this reason that, in order to establish universality in two-dimensional systems,
we had to treat them as `quasi-one-dimensional' and assume an extensive
symmetry group $(Z_2 \times Z_2)^{\times N}$, which grows with the vertical
extent of the system. For standard, non-extensive symmetries in two dimensions,
SPT orders can be related to 3-cocycles \cite{spto_2d,spto_higher}, but it
remains to be seen whether
similar connections can be drawn between such two-dimensional SPT order and
MBQC.

Finally, we note that if
MBQC in ground states of quantum spin systems is to be a robust form of quantum
computation, then it must be possible in the presence of arbitrary (not
necessarily symmetry-respecting) local perturbations to the Hamiltonian,
as well as at nonzero temperature.
Non-symmetry-respecting perturbations break the symmetry that is essential to our argument; the
difficulty of extending our treatment to nonzero temperature was discussed in
\secref{sec_finite_temperature}. Nor have we considered the
effect of non-ideal measurements, or of decoherence of the resource state taking
place during the course of the measurement protocol.
Therefore, it remains an open question whether fault-tolerant MBQC is 
possible with such imperfections.

\section*{Acknowledgements}
We thank I.~Schwarz for bringing to our attention the connection between SPT
order and the nature of the adaptive measurement protocol to correct for ``wrong''
outcomes (see \secref{sec_dual}).
We
acknowledge support from the ARC via the Centre of Excellence in Engineered
Quantum Systems (EQuS), project number CE110001013.

\ifdraft
\begin{appendices}
\else
\appendix
\fi
\section{The dual finitely correlated state}
\label{dual_mps}
A pure finitely-correlated state (pFCS) \cite{fcs,pure_fcs} is the thermodynamic limit of the
translationally invariant MPS generated by a fixed MPS tensor $A$. The
nature of the correlations can be expressed through the \emph{transfer channel}
\begin{equation}
\cA(\sigma) = \sum_k A[k] \sigma A[k]^{\dagger}
\end{equation}
(here the sum is over some basis $\{ \ket{k} \}$ for the site Hilbert space; it
can be shown that this definition of $\cA$ is independent of the choice of
basis). In its canonical form, a pure FCS is further characterized by the
following properties:
\begin{enumerate}[(a)]
\item $\mathcal{A}$ is unital, i.e. $\cA(\mathbb{I}) =
\mathbb{I}$.
\item There exists a density operator $\Lambda$ such that $\mathcal{A}^{\dagger}(\Lambda)
= \Lambda$.
\item Defining $a$ to be the largest
magnitude
eigenvalue of $\mathcal{A}^{\dagger}$ other than the one corresponding to the eigenvector
$\Lambda$, we have that $|a|$ is strictly less than 1.
\end{enumerate}
The correlation length is then defined by $\xi \equiv -1/\log |a|$, and
the eigenvalues of $\Lambda$ correspond to the entanglement
spectrum obtained from a cut in an infinite chain.

Now we restrict ourselves to pFCS generated by tensors $A$ satisfying the
decomposition \eqnref{formula_for_A}.
We define the CPTP superoperator $\mathcal{V}_g$ according to $\mathcal{V}_g(\sigma) = [V(g)
\otimes \mathbb{I}] \sigma [V(g) \otimes \mathbb{I}]^{\dagger}$; it can be checked
that $\mathcal{V}_g$ is a linear representation (in the space of superoperators)
of the symmetry group $G$, and that it commutes with $\mathcal{A}^{\dagger}$ for all $g \in G$.
Therefore, since $\Lambda$ is the unique eigenvector of $\mathcal{A}^{\dagger}$ with
eigenvalue 1, it must satisfy $\mathcal{V}_g(\Lambda) = \chi(g) \Lambda$ for
some scalars $\chi(g)$. The fact that $\mathcal{V}_g$ is trace preserving
ensures that $\chi(g) = 1$. Hence we find that $\Lambda$ commutes with $V(g)
\otimes \mathbb{I}$; it follows by Schur's Lemma that $\Lambda$ factorizes as
$\Lambda = \Omega \otimes \widetilde{\Lambda}$ for some density operator
$\widetilde{\Lambda}$, and where $\Omega = \mathbb{I}/\sqrt{|G|}$ is the
maximally-mixed state on the protected subsystem (recall that the $\sqrt{|G|}$
is the dimension of the protected subsystem). It
follows that there is a $\sqrt{|G|}$-fold
degeneracy in the entanglement spectrum throughout the SPT phase, generalizing
the 2-fold degeneracy in the $(Z_2 \times Z_2)$-protected Haldane phase \cite{pollmann-prb-2010}.

We can define the transfer channel corresponding to the dual FCS (generated by
$\widetilde{A}$) according to
\begin{equation}
\widetilde{\cA}(\rho) = \sum_k \widetilde{A}[k] \rho \widetilde{A}[k]^{\dagger}.
\end{equation}
Observe that $\cA$ unital implies that $\widetilde{\cA}$ is also unital,
and that 
\begin{equation}
\cA^{\dagger}(\Omega \otimes \sigma) = \Omega \otimes \widetilde{\cA}^{\dagger}(\sigma)
\end{equation}
for any operator $\sigma$ acting on the junk subsystem. Hence, any eigenvalue of
$\widetilde{\cA}$ must also be an eigenvalue of $\cA$.
From this we can see that the dual FCS is also a pure FCS, and
$\widetilde{\Lambda}$ is the unique fixed point of $\widetilde{\cA}^{\dagger}$; thus, the
entanglement spectrum of the dual state is the same as that of the original
state, but with the $\sqrt{|G|}$-fold degeneracy removed\footnote{A similar
property was found numerically for the Kennedy-Tasaki transformation in 
\cite{topological_disentangler}. We will discuss the connection between that
transformation and our ``dual state'' in Appendix \ref{kt_connection}.}.
The respective correlation lengths
obey the inequality $\widetilde{\xi} \leq \xi$.

\begin{figure}
\subfloat[]{\input{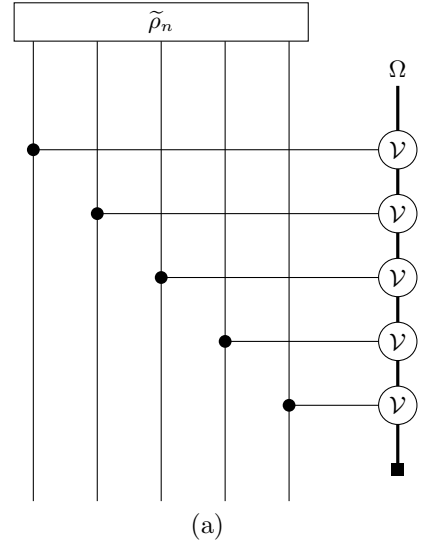}}
\hspace{2cm}
\subfloat[]{\input{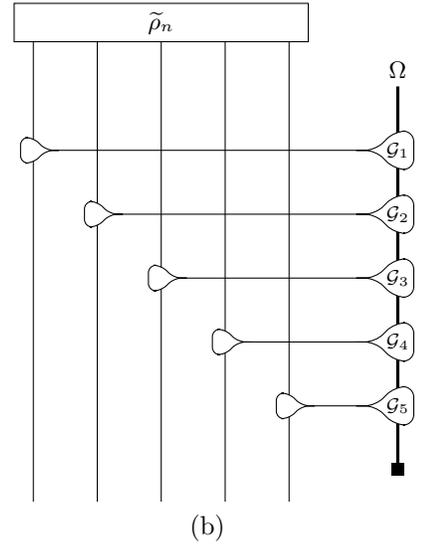}}
\caption{\label{mps_with_ends_10}(a) The quantum circuit generating the reduced
state on $n$ sites for the
original FCS. The ancilla particle is initially in the maximally-mixed state
$\Omega$,
and the symbol \protect\tikz{\protect \traceout{(0,0)}}
indicates that it should be traced out at the end. (b) The adaptive measurement
protocol acting on the original FCS is equivalent to a dual process acting on
the dual state, comprising a sequence of unitary interactions $\mathcal{G}_k$
coupling the dual state to an ancilla particle.}
\end{figure}

We now outline how the arguments of Section \ref{sec_dual} can be put 
on a rigorous footing within the pFCS formalism.
Given everything that we have established so far, it can be shown that the reduced state $\rho_n$ of the original FCS on a
block of $n$ adjacent sites can be obtained from the corresponding reduced state
$\widetilde{\rho}_n$ of the dual FCS according to the quantum circuit shown in
\figref{mps_with_ends_10}(a). Assuming that we choose $n$ large enough that all
measurements take place within this block of $n$ sites, we can then make
arguments similar to those of \secref{sec_dual}, and we find that the result of the
adaptive measurement protocol is equivalent to a sequence of interactions
between the dual state and an ancilla particle as shown in
\figref{mps_with_ends_10}(b), with the unitary interactions $\mathcal{G}_k$ defined as
they were previously.

\section{Connection with the Kennedy-Tasaki transformation}
\label{kt_connection}
The Kennedy-Tasaki (KT) transformation \cite{kt,kt2} is a non-local unitary transformation which
transforms a spin-1 chain in the $Z_2 \times Z_2$ symmetry-protected Haldane
phase into a system where the symmetry is spontaneously broken in the bulk. In this
section, adapting Ref.~\cite{topological_disentangler}, we will define a generalized version of the Kennedy-Tasaki
transformation, which can be applied to any system in the non-trivial SPT phase with
respect to an on-site representation $U(g) = [u(g)]^{\otimes N}$ of the symmetry
group $G = Z_2 \times Z_2$ (like the original KT transformation, our
generalization is defined for finite chains with open boundary conditions and no
terminating particles). We will show that, when the ground state of the original
system can be expressed as a pFCS, the ground state of the KT-transformed
system is essentially the same as the state which we have referred to throughout
this paper as the ``dual state''. We expect that for general ground states the
situation should be qualitatively similar.

Observe that the symmetry group is generated
by two commuting elements $x$ and $z$; hence for any $g \in G$, we can write $g
= x^{m(g)} z^{n(g)}$ for some $m(g)$ and $n(g)$ taking values of 0 or 1. We will
choose to write the unique non-trivial irreducible projective representation as
$V(g) = X^{m(g)} Z^{n(g)}$, where $X$ and $Z$ are the appropriate Pauli
operators. This is a rephasing of the Pauli representation $V_\cP$ defined in
\eqnref{pauli_projective}; thus the factor system is different to, but in the same cohomology
class as, that of $V_\cP$. (The construction of the dual state does depend
on the specific choice of representative factor system for a cohomology class, although in a
fairly trivial way; the present choice is the one that will ensure that the
Kennedy-Tasaki transformation reproduces the dual state exactly.)

The generalized Kennedy-Tasaki transformation $\cD_{KT}$ is then defined as
follows:
\begin{equation}
\label{DKT_product}
\cD_{KT} = \prod_{k < l} D_{kl},
\end{equation}
where $D_{kl}$ is a unitary coupling the two sites $k$ and $l$ according to
\begin{equation}
\label{Dkl}
D_{kl} \equiv \sum_i \ket{i} \bra{i} \otimes u\left(x^{m(g_i)}\right).
\end{equation}
Notice that all the operators $D_{kl}$ appearing in the product
\eqnref{DKT_product} commute.
In
the case that the particles are spin-1's, with the $Z_2 \times Z_2$ rotation
symmetry consisting of $\pi$ rotations about the $x$, $y$, and $z$ particles
(the simultaneous eigenbasis of the symmetry for a single site is then $\{ \ket{x},
\ket{y}, \ket{z} \}$,
where $\ket{\alpha}$, $\alpha = x,y,z$ is the $0$ eigenstate of the spin
component operator $S_\alpha$), the transformation $\mathcal{D}_{KT}$ reduces to
the standard Kennedy-Tasaki transformation.
Specifically, Eqs.\ (\ref{DKT_product}) and (\ref{Dkl}) correspond to Eqs.~(4)
and (5) in Ref.~\cite{topological_disentangler}.

For open boundary conditions, we expect there to be an approximate (becoming
exact in the limit as the length of the chain goes to infinity) four-fold
degeneracy, due to the two-fold degeneracy associated with each edge. An
appropriate analogue of the SPT pFCS for this choice of boundary conditions is as follows:
the low-energy subspace is spanned by states of the form
\begin{multline}
\label{mps_with_ends_23}
\ket{\Psi(L,R)} = \\
\input{mps_with_ends_23.fig},
\end{multline}
where the MPS tensor $A$ obeys the usual symmetry condition,
where $\ket{L_{*}}$ and $\bra{R_{*}}$ are \emph{fixed} end vectors, while
$\ket{L}$ and $\bra{R}$ are allowed to vary (within the two-dimensional irrep
space) in order to generate the four-dimensional low-energy subspace.

Now we want to examine what happens to a state of the form
\eqnref{mps_with_ends_23} under
$\mathcal{D}_{KT}$. Suppose we set $\bra{R} = \bra{+}$, $\ket{L} = \ket{0}$.
It is then straightforward to show [using the symmetry condition
\eqnref{A_symmetry_condition} and
the decomposition \eqnref{mps_tensor_decomposition}]
that applying all the pairwise interactions
$D_{1l}$ involving the first site gives the result
\begin{multline}
\left(\prod_{l > 1} D_{l1}\right) \ket{\Psi(0,+)} = \\
\input{mps_with_ends_24.fig}.
\end{multline}
Continuing this process, we find that
\begin{multline}
\label{Dkt_gives_dual}
\mathcal{D}_{KT} \ket{\Psi(0,+)} = \\
\input{mps_with_ends_25.fig},
\end{multline}
which is the dual state. To obtain the other states within the low-energy subspace,
it is sufficient to observe that, using the symmetry condition on the tensor
$A$, we have
\begin{align}
\ket{\Psi(1,+)} &= [u(x)]^{\otimes N} \ket{\Psi(0,+)}, \\
\ket{\Psi(0,-)} &= [u(z)]^{\otimes N} \ket{\Psi(0,+)}, \\
\ket{\Psi(1,-)} &= [u(xz)]^{\otimes N} \ket{\Psi(0,+)},
\end{align}
and that $\mathcal{D}_{KT}$ commutes with $u(g)$ for all $g \in Z_2 \times Z_2$.
Hence, the other basis states for the Kennedy-Tasaki transformed system can be
obtained from the dual state by application of a symmetry operation $[u(g)]^{\otimes N}$ for some $g
\in G$. The dual state thus represents one of the four distinct
symmetry-breaking states within the low-energy subspace of the transformed
system.

Let us also note that, although they relate to different boundary conditions, the unitary transformation $\mathcal{D}$ that we introduced in Sec.~\ref{exact_mps} is equivalent to the KT transformation $\mathcal{D}_{KT}$, in the following sense.
Indeed, an important property of $\mathcal{D}_{KT}$ is that for any local
symmetry-respecting observable $\hat{o}$, $\mathcal{D}_{KT} \hat{o}
\mathcal{D}_{KT}^{\dagger}$ remains local and symmetry-respecting. It turns out
that this also
holds for $\mathcal{D}$, in the case of observables $\hat{o}$ acting in the bulk
(the fact that
$\mathcal{D} \hat{o} \mathcal{D}^{\dagger}$ is local was established in
\secref{sec_dual_hamiltonian}; it
can be shown that $\mathcal{D} \hat{o} \mathcal{D}^{\dagger}$ still respects the
on-site symmetry as well). Thus, like $\mathcal{D}_{KT}$, the transformation
$\mathcal{D}$ can be applied to yield a local, symmetry-respecting Hamiltonian
$\widetilde{H}$ in the bulk. It can be shown $\widetilde{H}$ is precisely the KT transformed Hamiltonian~\cite{Else_forthcoming}.  Thus, in line with the results for pFCS described above, we expect $\widetilde{H}$ to have
four degenerate, locally distinguishable symmetry-breaking ground states in the
bulk. On the
other hand, when the transformation $\mathcal{D}$ is applied
to the edge interactions (those which couple the ends of the chain
to the terminating particles), the result need not respect the symmetry.
Therefore, including the edge interactions favours one of the symmetry-breaking
ground states over the others
and leads to the
non-degeneracy of the dual state $\ket{\widetilde{\Psi}}$.

\section{The factorization condition for ground states of a local Hamiltonian}
\label{appendix_factorization}
In this Appendix, we will show how to derive the approximate factorization result
\eqnref{weaker_general_bound} for a
non-degenerate gapped ground state $\ket{\Psi}$ of a local Hamiltonian. Our main tool is the
existing theorem on the exponential decay of correlation functions for such a
ground state \cite{hastings_lsm,nachtergaele_exponential,hastings_locality}. This theorem
states that there exists a correlation length $\xi$ and a function $f(x)$, with
$f(x) = \mathcal{O}[\exp(-x/\xi)]$ as $x \to \infty$, such that for any sets of
lattice sites $X$ and $Y$, and observables $A_X$ and $B_Y$ supported on $X$ and
$Y$ respectively, we have
\begin{multline}
\frac{|\langle A_X B_Y \rangle_\Psi - \langle A_X \rangle_\Psi \langle B_Y
\rangle_\Psi|}{\| A_X \| \| B_Y \|} \\
\leq f\bigl(\mathrm{dist}(X,Y)\bigr)
\min \{ |X|, |Y| \},
\end{multline}
where $\langle \cdot \rangle_\Psi$ denotes the expectation value of an
observable with respect to $\ket{\Psi}$, and $|X|$ denotes the number of
sites contained in the set $X$.

Suppose that instead of two sets of lattice sites, we have $m$ disjoint sets
$X_1, \ldots, X_m$. Let $N = \sum_{k=1}^m |X_k|$ be the total number of lattice
sites contained within all of the $X_k$'s. We can obtain the following
corollary:
\begin{lemma}
\label{exponential_corollary}
For any observables $A_{X_1}, \ldots, A_{X_m}$ supported on the respective
sets,
\begin{equation}
\label{cor}
\frac{\left|\langle A_{X_1} \cdots A_{X_m} \rangle_\Psi - \langle A_{X_1}
\rangle_\Psi \cdots \langle A_{X_m} \rangle_\Psi \right|}
{\| A_{X_1} \| \cdots \| A_{X_m} \|}
\leq  f(R) N,
\end{equation}
where $R$ is the smallest distance between any two of the $X_k$'s, i.e.\ $R
\equiv \min_{k \neq l} \dist(X_k,X_l)$.
\begin{proof}
Without loss of generality, we can assume that $\| A_{X_k} \| = 1$ for all
$k$.
Observing that $| \langle A_{X_1} \rangle | \leq \|A_{X_1}\| = 1$, we have that
\begin{multline}
\label{induction}
\left|\langle A_{X_1} \cdots A_{X_m} \rangle_\Psi - \langle A_{X_1}
\rangle_\Psi \cdots \langle A_{X_m} \rangle_\Psi \right|
\\ \leq |\langle A_{X_1} \cdots A_{X_m} \rangle_\Psi - 
\langle A_{X_1} \rangle_\Psi
\langle A_{X_2} \cdots A_{X_m} \rangle_\Psi|
\\+ |\langle A_{X_2} \cdots A_{X_m} \rangle_\Psi - \langle A_{X_2} \rangle_\Psi
\cdots \langle A_{X_m} \rangle_\Psi|.
\end{multline}
The two-body result implies that the first term in the right-hand side of \eqnref{induction} is bounded by $|X_1|
f(R)$. Continuing in this way, the lemma follows by induction.
\end{proof}
\end{lemma}

Now we want to show that the reduced state on the union of
the $X_k$'s is close to a product state. To do this, we make use of the
following lemma. We remind the reader that here we make use of both the trace norm $\| A \|_1 =
\Tr\sqrt{A^{\dagger} A}$ and the spectral norm $\| A \| =
\max_{\braket{\psi}{\psi} = 1} \| A \ket{\psi} \|$. They are both distinct from
the norm induced by the Hilbert-Schmidt
inner product.
\begin{lemma}
\label{basis_lemma}
Consider a vector space of dimension $D$.
Then there exists a
basis $\{ E_i | i = 1,...,D^2\}$ for the space of linear operators supported on the site,
orthonormal with respect to the Hilbert-Schmidt inner product $\langle A, B
\rangle = \Tr(A^{\dagger} B)$, and comprising Hermitian
operators $E_i$ such that such that $\|E_i\|_1 \| E_i \| = 1$.
\begin{proof}
Given a basis $\{\ket{m}\}$, $m = 1,\cdots,d$, such a set can be constructed
comprising the operators $\ket{m}\bra{m}$, and $(1/\sqrt{2})(\ket{m} \bra{n} +
\ket{n}\bra{m})$ and $(i/\sqrt{2})(\ket{m} \bra{n} - \ket{n} \bra{m})$ for
$m \neq n$.
\end{proof}
\end{lemma}
Now we are ready to prove the main result.
\begin{thm}
Let $\ket{\Psi}$ be the non-degenerate gapped ground state of a local
Hamiltonian. Let $\rho$ be the reduced state of $\ket{\Psi}$ on
$\bigcup_{k=1}^m X_k$, and let $\rho_k$ be the reduced state on $X_k$. Then
\begin{equation}
\left\| \rho - \rho_{\mathrm{prod}} \right\|_1 \leq N f(R) d^{2N},
\end{equation}
where $\rho_{\mathrm{prod}} = \bigotimes_k \rho_k$.
\begin{proof}
Recall that for a linear operator $P$, $\| P \|_1 = \max_{\|A\| =
1}|\Tr(A P)|$,
where the maximization is over all linear operators $A$ with unit spectral norm.
Now, we can expand
\begin{equation}
A = \sum_{i_1, \cdots, i_m} \alpha_{i_1, \cdots, i_m} E_{i_1} \otimes
\cdots \otimes E_{i_m},
\end{equation}
where the tensor product is over the sets $X_1, \cdots, X_m$; the
the $E_i$ are as constructed in Lemma \ref{basis_lemma}; and the scalars
$\alpha_{i_1, \cdots, i_m}$ are determined by
$\alpha_{i_1, \cdots, i_m} = \Tr[A^{\dagger}(E_{i_1} \otimes \cdots \otimes E_{i_m})]$,
which implies (since $\| A \| = 1$) that
\begin{equation}
|\alpha_{i_1, \cdots, i_m}| \leq \| E_{i_1} \|_1 \cdots \| E_{i_m} \|_1.
\end{equation}
Now, notice that
\begin{align}
\Tr(A\rho) &= \sum_{i_1, \cdots, i_m} \alpha_{i_1, \cdots, i_m} \langle E_{i_1} \otimes \cdots
\otimes E_{i_m} \rangle_{\rho}, \\
\Tr(A\rho_{\mathrm{prod}}) &= \sum_{i_1, \cdots, i_m}
\alpha_{i_1, \cdots, i_m} \langle E_{i_1} \rangle_{\rho_1} \cdots \langle E_{i_m}
\rangle_{\rho_m}.
\end{align}
Hence, by Lemma \ref{exponential_corollary}, we find that
\begin{align}
& | \Tr[A(\rho - \rho_{\mathrm{prod}})] | \\
&\leq N f(R) \sum_{i_1, \cdots, i_m} |\alpha_{i_1,\cdots,i_m}| \| E_{i_1} \| \cdots \|
E_{i_m} \| \\
&\leq N f(R) \sum_{i_1, \cdots, i_m} (\| E_{i_1} \|_1 \cdots \| E_{i_m} \|_1)
(\| E_{i_1} \| \cdots \| E_{i_m} \|) \\
&= N f(R) d^{2N},
\end{align}
since $\| E_k \|_1 \| E_k \| = 1$ for all $k$.
\end{proof}
\end{thm}

\section{Local perturbations perturb continuously}
\label{appendix_smallperturbations}
Physically, it should be clear that small perturbations in a gapped local Hamiltonian
lead to small variations in the reduced state obtained from the ground state on a finite
region of the lattice.
Here we will give a rigorous proof of this fact, as follows:
\begin{thm}
Let $H(s)$ be a differentiable path of Hamiltonians of the form
\begin{equation}
H(s) = \sum_{u \in \Lambda} H_u(s),
\end{equation}
where the sum is over all the lattice sites $u$ in a finite-dimensional lattice
$\Lambda$, and $H_u(s)$ is supported on the set $\mathcal{B}(u,r)$ of sites
within some fixed distance $r$ of $u$.
Suppose that for $0 \leq s \leq S$, the Hamiltonian $H(s)$ has a unique ground
state $\ket{\Psi(s)}$, and there is a uniform lower bound $\gamma > 0$
on the gap. 
Then there exists a constant $c$ (dependent only on the lattice
geometry, on $r$, and on $\gamma$) such that for any set $X$ of lattice sites,
we have
\begin{equation}
\|\rho_X(S) - \rho_X(0)\|_1 \leq  c |X| J^{\prime} S,
\end{equation}
where $\rho_X(s)$ is the reduced state on $X$, i.e.\ $\rho_X(s) = \Tr_{X^c}
\ket{\Psi(s)}\bra{\Psi(s)}$, and
$J^{\prime} \equiv \max_{u \in \Lambda, s \in [0,S]} \|\partial_s
H_u(s) \|$.
\begin{proof}
The proof relies on the following consequence of the theory of quasiadiabatic
continuation \cite{hastings_wen,osborne_quasiadiabatic}:
under the given assumptions, there exists a family of Hamiltonians
$\mathcal{H}(s)$ such that
\begin{equation}
i \frac{d}{ds}\ket{\Psi(s)} =  \mathcal{H}(s) \ket{\Psi(s)},
\end{equation}
where $\mathcal{H}(s)$ can be written as
\begin{equation}
\label{cH_sum}
\cH(s) = \sum_{u \in \Lambda} \cH_u(s),
\end{equation}
such that for any site $u$, $\cH_u(s)$ can be approximated by an observable
$\overline{\cH_u}(s)$ supported on $X^c$ (the complement of $X$), with error
\begin{equation}
\label{cH_commutator}
\| \overline{\cH_u}(s) - \cH_u(s) \| \leq J^{\prime} f\left(\frac{\mathrm{dist}(u,X)}{\gamma}\right)
\end{equation}
where $f$ is a rapidly decaying function (dependent only on $r$).

Now, for $s \in [0,S]$, we have (where $\rho(s) \equiv \ket{\Psi(s)}\bra{\Psi(s)}$)
\begin{equation}
\partial_s \rho_X(s) = i \Tr_{X^c} [\cH(s),\rho(s)].
\end{equation}
Hence
\begin{align}
\left\|\partial_s \rho_X(s)\right\|_1 &= \max_{\|A_X\|} \left|\Tr\bigl(
A_X \Tr_{X^c}
[\cH(s),\rho(s)] \bigr)\right| \\
&= \max_{\|A_X\| = 1} \left|\Tr\bigl( A_X [\cH(s),\rho(s)] \bigr)\right| \\
&= \max_{\|A_X\| = 1} \left|\Tr \bigl( [A_X, \cH(s)] \rho(s) \bigr)\right| \\
&\leq \max_{\|A_X\| = 1} \| [A_X, \cH(s)] \|.  \\
&\leq \max_{\|A_X\| = 1} \sum_{u \in \Lambda} \| [A_X, \cH_u(s)] \|
\label{the_sum}
\end{align}
Here the maximization is over all linear operators $A_X$ supported on $X$ with
unit spectral norm. We
have made use of the fact that for any linear operator $P$, $\| P\|_1 =
\max_{\|A\| = 1} |\Tr(A P)|$, with $\|A\|$ the spectral norm.

Now, for any operator $A_X$ supported on $X$, we have that
$[A_X, \overline{\cH_u}(s)] = 0$ since $A_X$ and $\overline{\cH_u}(s)$ are
supported on disjoint subsets. Hence, using \eqnref{cH_commutator} (and the fact
that $\| A_X \| = 1$), we find that
\begin{align}
\|[ A_X, \cH_u(s) ]\| &\leq 2 \| A_X \| \| \overline{\cH_u}(s) - \cH_u(s) \| \\ &\leq 2J^{\prime}
f\left(\frac{\mathrm{dist}(u,X)}{\gamma}\right),
\end{align}
so that
\begin{equation}
\| \partial_s \rho_X(s) \|_1 \leq 2J^{\prime} \sum_{u \in \Lambda} f\left( \frac{\min_{x \in
X}\mathrm{dist}(u,x)}{\gamma} \right).
\end{equation}
We can bound the sum according to
\begin{align}
& \sum_{u \in \Lambda}
 f\left(\frac{\dist(u,X)}{\gamma}\right) \\
&= \sum_{u \in \Lambda} f\left( \frac{\min_{x \in
X}\mathrm{dist}(u,x)}{\gamma} \right) \\
&\leq \sum_{u \in \Lambda} \sum_{x \in X} f\left( \frac{\mathrm{dist}(u,x)}{\gamma}
\right) \\
&= \sum_{x \in X} \sum_{u \in \Lambda} f\left(\frac{\mathrm{dist}(u,x)}{\gamma}\right) \\
&\leq |X| \max_{x \in X} \sum_{u \in \Lambda} f\left(\frac{\mathrm{dist}(u,x)}{\gamma}\right) \\
&\leq |X| c/2,
\end{align}
where in the last step the rapid decay of $f$ ensures that the sum is bounded by
a constant $c/2$ dependent on the lattice geometry, $\gamma$, and $r$.

Hence, by the triangle inequality for the trace norm, we have
\begin{align}
\| \rho_X(S) - \rho_X(0) \|_1 &\leq \int_0^S \left\| \partial_s \rho_X(s) \right\|_1 ds \\
&\leq c |X| J^{\prime} S,
\end{align}
as required.
\end{proof}
\end{thm}

\ifdraft
\end{appendices}
\fi

\ifdraft
   \bibliographystyle{apsrev4-1}
\fi
\bibliography{references}

\begin{thebibliography}{63}%
\makeatletter
\providecommand \@ifxundefined [1]{%
 \@ifx{#1\undefined}
}%
\providecommand \@ifnum [1]{%
 \ifnum #1\expandafter \@firstoftwo
 \else \expandafter \@secondoftwo
 \fi
}%
\providecommand \@ifx [1]{%
 \ifx #1\expandafter \@firstoftwo
 \else \expandafter \@secondoftwo
 \fi
}%
\providecommand \natexlab [1]{#1}%
\providecommand \enquote  [1]{``#1''}%
\providecommand \bibnamefont  [1]{#1}%
\providecommand \bibfnamefont [1]{#1}%
\providecommand \citenamefont [1]{#1}%
\providecommand \href@noop [0]{\@secondoftwo}%
\providecommand \href [0]{\begingroup \@sanitize@url \@href}%
\providecommand \@href[1]{\@@startlink{#1}\@@href}%
\providecommand \@@href[1]{\endgroup#1\@@endlink}%
\providecommand \@sanitize@url [0]{\catcode `\\12\catcode `\$12\catcode
  `\&12\catcode `\#12\catcode `\^12\catcode `\_12\catcode `\%12\relax}%
\providecommand \@@startlink[1]{}%
\providecommand \@@endlink[0]{}%
\providecommand \url  [0]{\begingroup\@sanitize@url \@url }%
\providecommand \@url [1]{\endgroup\@href {#1}{\urlprefix }}%
\providecommand \urlprefix  [0]{URL }%
\providecommand \Eprint [0]{\href }%
\@ifxundefined \urlstyle {%
  \providecommand \doi  [0]{\begingroup \@sanitize@url \@doi}%
  \providecommand \@doi [1]{\endgroup \@@startlink {\doibase
  #1}doi:\discretionary {}{}{}#1\@@endlink }%
}{%
  \providecommand \doi  [0]{doi:\discretionary{}{}{}\begingroup
  \urlstyle{rm}\Url }%
}%
\providecommand \doibase [0]{http://dx.doi.org/}%
\providecommand \Doi [0]{\begingroup \@sanitize@url \@Doi }%
\providecommand \@Doi  [1]{\endgroup\@@startlink{\doibase#1}\@@Doi}%
\providecommand \@@Doi [1]{#1\@@endlink}%
\providecommand \selectlanguage [0]{\@gobble}%
\providecommand \bibinfo  [0]{\@secondoftwo}%
\providecommand \bibfield  [0]{\@secondoftwo}%
\providecommand \translation [1]{[#1]}%
\providecommand \BibitemOpen [0]{}%
\providecommand \bibitemStop [0]{}%
\providecommand \bibitemNoStop [0]{.\EOS\space}%
\providecommand \EOS [0]{\spacefactor3000\relax}%
\providecommand \BibitemShut  [1]{\csname bibitem#1\endcsname}%
\bibitem [{\citenamefont {Raussendorf}\ and\ \citenamefont
  {Briegel}(2001)}]{raussendorf-prl-2001}%
  \BibitemOpen
  \bibfield  {author} {\bibinfo {author} {\bibfnamefont {Robert}\ \bibnamefont
  {Raussendorf}}\ and\ \bibinfo {author} {\bibfnamefont {Hans~J.}\ \bibnamefont
  {Briegel}},\ }\bibfield  {title} {\enquote {\bibinfo {title} {A one-way
  quantum computer},}\ }\Doi {10.1103/PhysRevLett.86.5188} {\bibfield
  {journal} {\bibinfo  {journal} {Phys. Rev. Lett.},\ }\textbf {\bibinfo
  {volume} {86}},\ \bibinfo {pages} {5188} (\bibinfo {year}
  {2001})}\BibitemShut {NoStop}%
\bibitem [{\citenamefont {Briegel}\ \emph {et~al.}(2009)\citenamefont
  {Briegel}, \citenamefont {Browne}, \citenamefont {Dur}, \citenamefont
  {Raussendorf},\ and\ \citenamefont {Van~den Nest}}]{briegel_mbqc_natphys}%
  \BibitemOpen
  \bibfield  {author} {\bibinfo {author} {\bibfnamefont {H.~J.}\ \bibnamefont
  {Briegel}}, \bibinfo {author} {\bibfnamefont {D.~E.}\ \bibnamefont {Browne}},
  \bibinfo {author} {\bibfnamefont {W.}~\bibnamefont {Dur}}, \bibinfo {author}
  {\bibfnamefont {R.}~\bibnamefont {Raussendorf}}, \ and\ \bibinfo {author}
  {\bibfnamefont {M.}~\bibnamefont {Van~den Nest}},\ }\bibfield  {title}
  {\enquote {\bibinfo {title} {Measurement-based quantum computation},}\ }\Doi
  {10.1038/nphys1157} {\bibfield  {journal} {\bibinfo  {journal} {Nat. Phys.},\
  }\textbf {\bibinfo {volume} {5}},\ \bibinfo {pages} {19} (\bibinfo {year}
  {2009})},\ \Eprint {http://arxiv.org/abs/0910.1116} {arXiv:0910.1116}
  \BibitemShut {NoStop}%
\bibitem [{\citenamefont {Gross}\ and\ \citenamefont
  {Eisert}(2007)}]{correlation_space_prl}%
  \BibitemOpen
  \bibfield  {author} {\bibinfo {author} {\bibfnamefont {D.}~\bibnamefont
  {Gross}}\ and\ \bibinfo {author} {\bibfnamefont {J.}~\bibnamefont {Eisert}},\
  }\bibfield  {title} {\enquote {\bibinfo {title} {Novel schemes for
  measurement-based quantum computation},}\ }\Doi
  {10.1103/PhysRevLett.98.220503} {\bibfield  {journal} {\bibinfo  {journal}
  {Phys. Rev. Lett.},\ }\textbf {\bibinfo {volume} {98}},\ \bibinfo {pages}
  {220503} (\bibinfo {year} {2007})},\ \Eprint
  {http://arxiv.org/abs/arXiv:quant-ph/0609149} {arXiv:quant-ph/0609149}
  \BibitemShut {NoStop}%
\bibitem [{\citenamefont {Gross}\ \emph {et~al.}(2007)\citenamefont {Gross},
  \citenamefont {Eisert}, \citenamefont {Schuch},\ and\ \citenamefont
  {P{\'e}rez-Garc{\'i}a}}]{correlation_space_pra}%
  \BibitemOpen
  \bibfield  {author} {\bibinfo {author} {\bibfnamefont {D.}~\bibnamefont
  {Gross}}, \bibinfo {author} {\bibfnamefont {J.}~\bibnamefont {Eisert}},
  \bibinfo {author} {\bibfnamefont {N.}~\bibnamefont {Schuch}}, \ and\ \bibinfo
  {author} {\bibfnamefont {D.}~\bibnamefont {P{\'e}rez-Garc{\'i}a}},\
  }\bibfield  {title} {\enquote {\bibinfo {title} {Measurement-based quantum
  computation beyond the one-way model},}\ }\Doi {10.1103/PhysRevA.76.052315}
  {\bibfield  {journal} {\bibinfo  {journal} {Phys. Rev. A},\ }\textbf
  {\bibinfo {volume} {76}},\ \bibinfo {pages} {052315} (\bibinfo {year}
  {2007})},\ \Eprint {http://arxiv.org/abs/arXiv:0706.3401} {arXiv:0706.3401}
  \BibitemShut {NoStop}%
\bibitem [{\citenamefont {Chen}\ \emph {et~al.}(2009)\citenamefont {Chen},
  \citenamefont {Zeng}, \citenamefont {Gu}, \citenamefont {Yoshida},\ and\
  \citenamefont {Chuang}}]{chen_tricluster}%
  \BibitemOpen
  \bibfield  {author} {\bibinfo {author} {\bibfnamefont {Xie}\ \bibnamefont
  {Chen}}, \bibinfo {author} {\bibfnamefont {Bei}\ \bibnamefont {Zeng}},
  \bibinfo {author} {\bibfnamefont {Zheng-Cheng}\ \bibnamefont {Gu}}, \bibinfo
  {author} {\bibfnamefont {Beni}\ \bibnamefont {Yoshida}}, \ and\ \bibinfo
  {author} {\bibfnamefont {Isaac~L.}\ \bibnamefont {Chuang}},\ }\bibfield
  {title} {\enquote {\bibinfo {title} {Gapped two-body hamiltonian whose unique
  ground state is universal for one-way quantum computation},}\ }\Doi
  {10.1103/PhysRevLett.102.220501} {\bibfield  {journal} {\bibinfo  {journal}
  {Phys. Rev. Lett.},\ }\textbf {\bibinfo {volume} {102}},\ \bibinfo {pages}
  {220501} (\bibinfo {year} {2009})},\ \Eprint {http://arxiv.org/abs/0812.4067}
  {arXiv:0812.4067} \BibitemShut {NoStop}%
\bibitem [{\citenamefont {Wei}\ \emph {et~al.}(2011)\citenamefont {Wei},
  \citenamefont {Affleck},\ and\ \citenamefont
  {Raussendorf}}]{raussendorf_aklt1}%
  \BibitemOpen
  \bibfield  {author} {\bibinfo {author} {\bibfnamefont {Tzu-Chieh}\
  \bibnamefont {Wei}}, \bibinfo {author} {\bibfnamefont {Ian}\ \bibnamefont
  {Affleck}}, \ and\ \bibinfo {author} {\bibfnamefont {Robert}\ \bibnamefont
  {Raussendorf}},\ }\bibfield  {title} {\enquote {\bibinfo {title}
  {{Affleck-Kennedy-Lieb-Tasaki} state on a honeycomb lattice is a universal
  quantum computational resource},}\ }\Doi {10.1103/PhysRevLett.106.070501}
  {\bibfield  {journal} {\bibinfo  {journal} {Phys. Rev. Lett.},\ }\textbf
  {\bibinfo {volume} {106}},\ \bibinfo {pages} {070501} (\bibinfo {year}
  {2011})},\ \Eprint {http://arxiv.org/abs/1102.5064} {arXiv:1102.5064}
  \BibitemShut {NoStop}%
\bibitem [{\citenamefont {Miyake}(2011)}]{miyake_aklt}%
  \BibitemOpen
  \bibfield  {author} {\bibinfo {author} {\bibfnamefont {Akimasa}\ \bibnamefont
  {Miyake}},\ }\bibfield  {title} {\enquote {\bibinfo {title} {Quantum
  computational capability of a {2D} valence bond solid phase},}\ }\Doi
  {10.1016/j.aop.2011.03.006} {\bibfield  {journal} {\bibinfo  {journal} {Ann.
  Phys.},\ }\textbf {\bibinfo {volume} {326}},\ \bibinfo {pages} {1656}
  (\bibinfo {year} {2011})},\ \Eprint {http://arxiv.org/abs/1009.3491}
  {arXiv:1009.3491} \BibitemShut {NoStop}%
\bibitem [{\citenamefont {Darmawan}\ \emph {et~al.}(2012)\citenamefont
  {Darmawan}, \citenamefont {Brennen},\ and\ \citenamefont
  {Bartlett}}]{darmawan_phase}%
  \BibitemOpen
  \bibfield  {author} {\bibinfo {author} {\bibfnamefont {Andrew~S}\
  \bibnamefont {Darmawan}}, \bibinfo {author} {\bibfnamefont {Gavin~K}\
  \bibnamefont {Brennen}}, \ and\ \bibinfo {author} {\bibfnamefont {Stephen~D}\
  \bibnamefont {Bartlett}},\ }\bibfield  {title} {\enquote {\bibinfo {title}
  {Measurement-based quantum computation in a two-dimensional phase of
  matter},}\ }\Doi {10.1088/1367-2630/14/1/013023} {\bibfield  {journal}
  {\bibinfo  {journal} {New J. Phys.},\ }\textbf {\bibinfo {volume} {14}},\
  \bibinfo {pages} {013023} (\bibinfo {year} {2012})},\ \Eprint
  {http://arxiv.org/abs/1108.4741} {arXiv:1108.4741} \BibitemShut {NoStop}%
\bibitem [{\citenamefont {Verstraete}\ and\ \citenamefont
  {Cirac}(2004){\natexlab{a}}}]{peps}%
  \BibitemOpen
  \bibfield  {author} {\bibinfo {author} {\bibfnamefont {F.}~\bibnamefont
  {Verstraete}}\ and\ \bibinfo {author} {\bibfnamefont {J.~I.}\ \bibnamefont
  {Cirac}},\ }\href@noop {} {\enquote {\bibinfo {title} {Renormalization
  algorithms for quantum-many body systems in two and higher dimensions},}\ }
  (\bibinfo {year} {2004}{\natexlab{a}}),\ \Eprint
  {http://arxiv.org/abs/arXiv:cond-mat/0407066} {arXiv:cond-mat/0407066}
  \BibitemShut {NoStop}%
\bibitem [{\citenamefont {Perez-Garcia}\ \emph {et~al.}(2008)\citenamefont
  {Perez-Garcia}, \citenamefont {Verstraete},\ and\ \citenamefont
  {Cirac}}]{peps_injectivity}%
  \BibitemOpen
  \bibfield  {author} {\bibinfo {author} {\bibfnamefont {David}\ \bibnamefont
  {Perez-Garcia}}, \bibinfo {author} {\bibfnamefont {Frank}\ \bibnamefont
  {Verstraete}}, \ and\ \bibinfo {author} {\bibfnamefont {J.~Ignacio}\
  \bibnamefont {Cirac}},\ }\bibfield  {title} {\enquote {\bibinfo {title}
  {{PEPS} as unique ground states of local hamiltonians},}\ }\href@noop {}
  {\bibfield  {journal} {\bibinfo  {journal} {Quant. Inf. Comput.},\ }\textbf
  {\bibinfo {volume} {8}},\ \bibinfo {pages} {650} (\bibinfo {year} {2008})},\
  \Eprint {http://arxiv.org/abs/0707.2260} {arXiv:0707.2260} \BibitemShut
  {NoStop}%
\bibitem [{\citenamefont {Verstraete}\ and\ \citenamefont
  {Cirac}(2004){\natexlab{b}}}]{verstraete_cirac_2004}%
  \BibitemOpen
  \bibfield  {author} {\bibinfo {author} {\bibfnamefont {F.}~\bibnamefont
  {Verstraete}}\ and\ \bibinfo {author} {\bibfnamefont {J.~I.}\ \bibnamefont
  {Cirac}},\ }\bibfield  {title} {\enquote {\bibinfo {title} {Valence-bond
  states for quantum computation},}\ }\Doi {10.1103/PhysRevA.70.060302}
  {\bibfield  {journal} {\bibinfo  {journal} {Phys. Rev. A},\ }\textbf
  {\bibinfo {volume} {70}},\ \bibinfo {pages} {060302} (\bibinfo {year}
  {2004}{\natexlab{b}})},\ \Eprint {http://arxiv.org/abs/quant-ph/0311130}
  {arXiv:quant-ph/0311130} \BibitemShut {NoStop}%
\bibitem [{\citenamefont {Jennings}\ \emph {et~al.}(2009)\citenamefont
  {Jennings}, \citenamefont {Dragan}, \citenamefont {Barrett}, \citenamefont
  {Bartlett},\ and\ \citenamefont {Rudolph}}]{jennings_etal_2009}%
  \BibitemOpen
  \bibfield  {author} {\bibinfo {author} {\bibfnamefont {David}\ \bibnamefont
  {Jennings}}, \bibinfo {author} {\bibfnamefont {Andrzej}\ \bibnamefont
  {Dragan}}, \bibinfo {author} {\bibfnamefont {Sean~D.}\ \bibnamefont
  {Barrett}}, \bibinfo {author} {\bibfnamefont {Stephen~D.}\ \bibnamefont
  {Bartlett}}, \ and\ \bibinfo {author} {\bibfnamefont {Terry}\ \bibnamefont
  {Rudolph}},\ }\bibfield  {title} {\enquote {\bibinfo {title} {Quantum
  computation via measurements on the low-temperature state of a many-body
  system},}\ }\Doi {10.1103/PhysRevA.80.032328} {\bibfield  {journal} {\bibinfo
   {journal} {Phys. Rev. A},\ }\textbf {\bibinfo {volume} {80}},\ \bibinfo
  {pages} {032328} (\bibinfo {year} {2009})},\ \Eprint
  {http://arxiv.org/abs/0906.3553} {arXiv:0906.3553} \BibitemShut {NoStop}%
\bibitem [{\citenamefont {Doherty}\ and\ \citenamefont
  {Bartlett}(2009)}]{doherty-bartlett-prl-2008}%
  \BibitemOpen
  \bibfield  {author} {\bibinfo {author} {\bibfnamefont {Andrew~C.}\
  \bibnamefont {Doherty}}\ and\ \bibinfo {author} {\bibfnamefont {Stephen~D.}\
  \bibnamefont {Bartlett}},\ }\bibfield  {title} {\enquote {\bibinfo {title}
  {Identifying phases of quantum many-body systems that are universal for
  quantum computation},}\ }\Doi {10.1103/PhysRevLett.103.020506} {\bibfield
  {journal} {\bibinfo  {journal} {Phys. Rev. Lett.},\ }\textbf {\bibinfo
  {volume} {103}},\ \bibinfo {pages} {020506} (\bibinfo {year} {2009})},\
  \Eprint {http://arxiv.org/abs/arXiv:0802.4314} {arXiv:0802.4314} \BibitemShut
  {NoStop}%
\bibitem [{\citenamefont {Bartlett}\ \emph {et~al.}(2010)\citenamefont
  {Bartlett}, \citenamefont {Brennen}, \citenamefont {Miyake},\ and\
  \citenamefont {Renes}}]{bartlett_renormalization}%
  \BibitemOpen
  \bibfield  {author} {\bibinfo {author} {\bibfnamefont {Stephen~D.}\
  \bibnamefont {Bartlett}}, \bibinfo {author} {\bibfnamefont {Gavin~K.}\
  \bibnamefont {Brennen}}, \bibinfo {author} {\bibfnamefont {Akimasa}\
  \bibnamefont {Miyake}}, \ and\ \bibinfo {author} {\bibfnamefont {Joseph~M.}\
  \bibnamefont {Renes}},\ }\bibfield  {title} {\enquote {\bibinfo {title}
  {Quantum computational renormalization in the {Haldane} phase},}\ }\Doi
  {10.1103/PhysRevLett.105.110502} {\bibfield  {journal} {\bibinfo  {journal}
  {Phys. Rev. Lett.},\ }\textbf {\bibinfo {volume} {105}},\ \bibinfo {pages}
  {110502} (\bibinfo {year} {2010})},\ \Eprint
  {http://arxiv.org/abs/arXiv:1004.4906} {arXiv:1004.4906} \BibitemShut
  {NoStop}%
\bibitem [{\citenamefont {Else}\ \emph {et~al.}(2012)\citenamefont {Else},
  \citenamefont {Schwarz}, \citenamefont {Bartlett},\ and\ \citenamefont
  {Doherty}}]{else_schwarz_bartlett_doherty_symmetry}%
  \BibitemOpen
  \bibfield  {author} {\bibinfo {author} {\bibfnamefont {Dominic~V.}\
  \bibnamefont {Else}}, \bibinfo {author} {\bibfnamefont {Ilai}\ \bibnamefont
  {Schwarz}}, \bibinfo {author} {\bibfnamefont {Stephen~D.}\ \bibnamefont
  {Bartlett}}, \ and\ \bibinfo {author} {\bibfnamefont {Andrew~C.}\
  \bibnamefont {Doherty}},\ }\bibfield  {title} {\enquote {\bibinfo {title}
  {Symmetry-protected phases for measurement-based quantum computation},}\
  }\Doi {10.1103/PhysRevLett.108.240505} {\bibfield  {journal} {\bibinfo
  {journal} {Phys. Rev. Lett.},\ }\textbf {\bibinfo {volume} {108}},\ \bibinfo
  {pages} {240505} (\bibinfo {year} {2012})},\ \Eprint
  {http://arxiv.org/abs/1201.4877} {arXiv:1201.4877} \BibitemShut {NoStop}%
\bibitem [{\citenamefont {Gu}\ and\ \citenamefont {Wen}(2009)}]{gu_wen_2009}%
  \BibitemOpen
  \bibfield  {author} {\bibinfo {author} {\bibfnamefont {Zheng-Cheng}\
  \bibnamefont {Gu}}\ and\ \bibinfo {author} {\bibfnamefont {Xiao-Gang}\
  \bibnamefont {Wen}},\ }\bibfield  {title} {\enquote {\bibinfo {title}
  {Tensor-entanglement-filtering renormalization approach and
  symmetry-protected topological order},}\ }\Doi {10.1103/PhysRevB.80.155131}
  {\bibfield  {journal} {\bibinfo  {journal} {Phys. Rev. B},\ }\textbf
  {\bibinfo {volume} {80}},\ \bibinfo {pages} {155131} (\bibinfo {year}
  {2009})},\ \Eprint {http://arxiv.org/abs/0903.1069} {arXiv:0903.1069}
  \BibitemShut {NoStop}%
\bibitem [{\citenamefont {Chen}\ \emph
  {et~al.}(2011){\natexlab{a}}\citenamefont {Chen}, \citenamefont {Gu},\ and\
  \citenamefont {Wen}}]{chen_gu_wen}%
  \BibitemOpen
  \bibfield  {author} {\bibinfo {author} {\bibfnamefont {Xie}\ \bibnamefont
  {Chen}}, \bibinfo {author} {\bibfnamefont {Zheng-Cheng}\ \bibnamefont {Gu}},
  \ and\ \bibinfo {author} {\bibfnamefont {Xiao-Gang}\ \bibnamefont {Wen}},\
  }\bibfield  {title} {\enquote {\bibinfo {title} {Classification of gapped
  symmetric phases in one-dimensional spin systems},}\ }\Doi
  {10.1103/PhysRevB.83.035107} {\bibfield  {journal} {\bibinfo  {journal}
  {Phys. Rev. B},\ }\textbf {\bibinfo {volume} {83}},\ \bibinfo {pages}
  {035107} (\bibinfo {year} {2011}{\natexlab{a}})},\ \Eprint
  {http://arxiv.org/abs/arXiv:1008.3745} {arXiv:1008.3745} \BibitemShut
  {NoStop}%
\bibitem [{\citenamefont {Schuch}\ \emph {et~al.}(2011)\citenamefont {Schuch},
  \citenamefont {P{\'e}rez-Garc{\'i}a},\ and\ \citenamefont {Cirac}}]{schuch}%
  \BibitemOpen
  \bibfield  {author} {\bibinfo {author} {\bibfnamefont {Norbert}\ \bibnamefont
  {Schuch}}, \bibinfo {author} {\bibfnamefont {David}\ \bibnamefont
  {P{\'e}rez-Garc{\'i}a}}, \ and\ \bibinfo {author} {\bibfnamefont {Ignacio}\
  \bibnamefont {Cirac}},\ }\bibfield  {title} {\enquote {\bibinfo {title}
  {Classifying quantum phases using matrix product states and projected
  entangled pair states},}\ }\Doi {10.1103/PhysRevB.84.165139} {\bibfield
  {journal} {\bibinfo  {journal} {Phys. Rev. B},\ }\textbf {\bibinfo {volume}
  {84}},\ \bibinfo {pages} {165139} (\bibinfo {year} {2011})},\ \Eprint
  {http://arxiv.org/abs/1010.3732} {arXiv:1010.3732} \BibitemShut {NoStop}%
\bibitem [{\citenamefont {Nielsen}\ and\ \citenamefont
  {Dawson}(2005)}]{nielsen_cluster_faulttolerance}%
  \BibitemOpen
  \bibfield  {author} {\bibinfo {author} {\bibfnamefont {Michael~A.}\
  \bibnamefont {Nielsen}}\ and\ \bibinfo {author} {\bibfnamefont
  {Christopher~M.}\ \bibnamefont {Dawson}},\ }\bibfield  {title} {\enquote
  {\bibinfo {title} {Fault-tolerant quantum computation with cluster states},}\
  }\Doi {10.1103/PhysRevA.71.042323} {\bibfield  {journal} {\bibinfo  {journal}
  {Phys. Rev. A},\ }\textbf {\bibinfo {volume} {71}},\ \bibinfo {pages}
  {042323} (\bibinfo {year} {2005})},\ \Eprint
  {http://arxiv.org/abs/quant-ph/0405134} {arXiv:quant-ph/0405134} \BibitemShut
  {NoStop}%
\bibitem [{\citenamefont {Aharonov}\ and\ \citenamefont
  {Ben-Or}(1999)}]{aharonov_ben_or}%
  \BibitemOpen
  \bibfield  {author} {\bibinfo {author} {\bibfnamefont {Dorit}\ \bibnamefont
  {Aharonov}}\ and\ \bibinfo {author} {\bibfnamefont {Michael}\ \bibnamefont
  {Ben-Or}},\ }\href@noop {} {\enquote {\bibinfo {title} {Fault-tolerant
  quantum computation with constant error rate},}\ } (\bibinfo {year} {1999}),\
  \Eprint {http://arxiv.org/abs/quant-ph/9906129} {arXiv:quant-ph/9906129}
  \BibitemShut {NoStop}%
\bibitem [{\citenamefont {Fannes}\ \emph {et~al.}(1992)\citenamefont {Fannes},
  \citenamefont {Nachtergaele},\ and\ \citenamefont {Werner}}]{fcs}%
  \BibitemOpen
  \bibfield  {author} {\bibinfo {author} {\bibfnamefont {M.}~\bibnamefont
  {Fannes}}, \bibinfo {author} {\bibfnamefont {B.}~\bibnamefont
  {Nachtergaele}}, \ and\ \bibinfo {author} {\bibfnamefont {R.}~\bibnamefont
  {Werner}},\ }\bibfield  {title} {\enquote {\bibinfo {title} {Finitely
  correlated states on quantum spin chains},}\ }\Doi {10.1007/BF02099178}
  {\bibfield  {journal} {\bibinfo  {journal} {Comm. Math. Phys.},\ }\textbf
  {\bibinfo {volume} {144}},\ \bibinfo {pages} {443} (\bibinfo {year}
  {1992})}\BibitemShut {NoStop}%
\bibitem [{\citenamefont {Fannes}\ \emph {et~al.}(1994)\citenamefont {Fannes},
  \citenamefont {Nachtergaele},\ and\ \citenamefont {Werner}}]{pure_fcs}%
  \BibitemOpen
  \bibfield  {author} {\bibinfo {author} {\bibfnamefont {M.}~\bibnamefont
  {Fannes}}, \bibinfo {author} {\bibfnamefont {B.}~\bibnamefont
  {Nachtergaele}}, \ and\ \bibinfo {author} {\bibfnamefont {R.F.}\ \bibnamefont
  {Werner}},\ }\bibfield  {title} {\enquote {\bibinfo {title} {Finitely
  correlated pure states},}\ }\Doi {10.1006/jfan.1994.1041} {\bibfield
  {journal} {\bibinfo  {journal} {Journal of Functional Analysis},\ }\textbf
  {\bibinfo {volume} {120}},\ \bibinfo {pages} {511} (\bibinfo {year}
  {1994})}\BibitemShut {NoStop}%
\bibitem [{\citenamefont {Shi}\ \emph {et~al.}(2006)\citenamefont {Shi},
  \citenamefont {Duan},\ and\ \citenamefont {Vidal}}]{vidal_tensor_network}%
  \BibitemOpen
  \bibfield  {author} {\bibinfo {author} {\bibfnamefont {Y.-Y.}\ \bibnamefont
  {Shi}}, \bibinfo {author} {\bibfnamefont {L.-M.}\ \bibnamefont {Duan}}, \
  and\ \bibinfo {author} {\bibfnamefont {G.}~\bibnamefont {Vidal}},\ }\bibfield
   {title} {\enquote {\bibinfo {title} {Classical simulation of quantum
  many-body systems with a tree tensor network},}\ }\Doi
  {10.1103/PhysRevA.74.022320} {\bibfield  {journal} {\bibinfo  {journal}
  {Phys. Rev. A},\ }\textbf {\bibinfo {volume} {74}},\ \bibinfo {pages}
  {022320} (\bibinfo {year} {2006})},\ \Eprint
  {http://arxiv.org/abs/quant-ph/0511070} {arXiv:quant-ph/0511070} \BibitemShut
  {NoStop}%
\bibitem [{\citenamefont {Schuch}\ \emph {et~al.}(2010)\citenamefont {Schuch},
  \citenamefont {Cirac},\ and\ \citenamefont
  {P{\'e}rez-Garc{\'i}a}}]{peps_degeneracy_and_topology}%
  \BibitemOpen
  \bibfield  {author} {\bibinfo {author} {\bibfnamefont {Norbert}\ \bibnamefont
  {Schuch}}, \bibinfo {author} {\bibfnamefont {Ignacio}\ \bibnamefont {Cirac}},
  \ and\ \bibinfo {author} {\bibfnamefont {David}\ \bibnamefont
  {P{\'e}rez-Garc{\'i}a}},\ }\bibfield  {title} {\enquote {\bibinfo {title}
  {{PEPS} as ground states: Degeneracy and topology},}\ }\Doi
  {10.1016/j.aop.2010.05.008} {\bibfield  {journal} {\bibinfo  {journal}
  {Annals of Physics},\ }\textbf {\bibinfo {volume} {325}},\ \bibinfo {pages}
  {2153} (\bibinfo {year} {2010})},\ \Eprint {http://arxiv.org/abs/1001.3807}
  {arXiv:1001.3807} \BibitemShut {NoStop}%
\bibitem [{\citenamefont {P{\'e}rez-Garc{\'i}a}\ \emph
  {et~al.}(2007)\citenamefont {P{\'e}rez-Garc{\'i}a}, \citenamefont
  {Verstraete}, \citenamefont {Wolf},\ and\ \citenamefont
  {Cirac}}]{perezgarcia_mps}%
  \BibitemOpen
  \bibfield  {author} {\bibinfo {author} {\bibfnamefont {D.}~\bibnamefont
  {P{\'e}rez-Garc{\'i}a}}, \bibinfo {author} {\bibfnamefont {F.}~\bibnamefont
  {Verstraete}}, \bibinfo {author} {\bibfnamefont {M.~M.}\ \bibnamefont
  {Wolf}}, \ and\ \bibinfo {author} {\bibfnamefont {J.~I.}\ \bibnamefont
  {Cirac}},\ }\bibfield  {title} {\enquote {\bibinfo {title} {Matrix product
  state representations},}\ }\href@noop {} {\bibfield  {journal} {\bibinfo
  {journal} {Quant.\ Inf.\ Comput.},\ }\textbf {\bibinfo {volume} {7}},\
  \bibinfo {pages} {401} (\bibinfo {year} {2007})},\ \Eprint
  {http://arxiv.org/abs/arXiv:quant-ph/0608197} {arXiv:quant-ph/0608197}
  \BibitemShut {NoStop}%
\bibitem [{\citenamefont {Verstraete}\ and\ \citenamefont
  {Cirac}(2006)}]{mps_faithfully}%
  \BibitemOpen
  \bibfield  {author} {\bibinfo {author} {\bibfnamefont {F.}~\bibnamefont
  {Verstraete}}\ and\ \bibinfo {author} {\bibfnamefont {J.~I.}\ \bibnamefont
  {Cirac}},\ }\bibfield  {title} {\enquote {\bibinfo {title} {Matrix product
  states represent ground states faithfully},}\ }\Doi
  {10.1103/PhysRevB.73.094423} {\bibfield  {journal} {\bibinfo  {journal}
  {Phys. Rev. B},\ }\textbf {\bibinfo {volume} {73}},\ \bibinfo {pages}
  {094423} (\bibinfo {year} {2006})},\ \Eprint
  {http://arxiv.org/abs/arXiv:cond-mat/0505140} {arXiv:cond-mat/0505140}
  \BibitemShut {NoStop}%
\bibitem [{\citenamefont {Hastings}(2007)}]{hastings_area_law}%
  \BibitemOpen
  \bibfield  {author} {\bibinfo {author} {\bibfnamefont {M.~B.}\ \bibnamefont
  {Hastings}},\ }\bibfield  {title} {\enquote {\bibinfo {title} {An area law
  for one-dimensional quantum systems},}\ }\Doi
  {10.1088/1742-5468/2007/08/P08024} {\bibfield  {journal} {\bibinfo  {journal}
  {J. Stat. Mech.: Theor. Exp.},\ \bibinfo {pages} {P08024}} (\bibinfo {year}
  {2007})},\ \Eprint {http://arxiv.org/abs/arXiv:0705.2024} {arXiv:0705.2024}
  \BibitemShut {NoStop}%
\bibitem [{\citenamefont {Son}\ \emph {et~al.}(2011)\citenamefont {Son},
  \citenamefont {Amico}, \citenamefont {Fazio}, \citenamefont {Hamma},
  \citenamefont {Pascazio},\ and\ \citenamefont {Vedral}}]{son-cluster-2011}%
  \BibitemOpen
  \bibfield  {author} {\bibinfo {author} {\bibfnamefont {W.}~\bibnamefont
  {Son}}, \bibinfo {author} {\bibfnamefont {L.}~\bibnamefont {Amico}}, \bibinfo
  {author} {\bibfnamefont {R.}~\bibnamefont {Fazio}}, \bibinfo {author}
  {\bibfnamefont {A.}~\bibnamefont {Hamma}}, \bibinfo {author} {\bibfnamefont
  {S.}~\bibnamefont {Pascazio}}, \ and\ \bibinfo {author} {\bibfnamefont
  {V.}~\bibnamefont {Vedral}},\ }\bibfield  {title} {\enquote {\bibinfo {title}
  {Quantum phase transition between cluster and antiferromagnetic states},}\
  }\Doi {10.1209/0295-5075/95/50001} {\bibfield  {journal} {\bibinfo  {journal}
  {Europhys. Lett.},\ }\textbf {\bibinfo {volume} {95}},\ \bibinfo {pages}
  {50001} (\bibinfo {year} {2011})},\ \Eprint {http://arxiv.org/abs/1103.0251}
  {arXiv:1103.0251} \BibitemShut {NoStop}%
\bibitem [{\citenamefont {Gross}\ and\ \citenamefont
  {Eisert}(2010)}]{gross-eisert-2010-webs}%
  \BibitemOpen
  \bibfield  {author} {\bibinfo {author} {\bibfnamefont {D.}~\bibnamefont
  {Gross}}\ and\ \bibinfo {author} {\bibfnamefont {J.}~\bibnamefont {Eisert}},\
  }\bibfield  {title} {\enquote {\bibinfo {title} {Quantum computational
  webs},}\ }\Doi {10.1103/PhysRevA.82.040303} {\bibfield  {journal} {\bibinfo
  {journal} {Phys. Rev. A},\ }\textbf {\bibinfo {volume} {82}},\ \bibinfo
  {pages} {040303} (\bibinfo {year} {2010})},\ \Eprint
  {http://arxiv.org/abs/arXiv:0810.2542} {arXiv:0810.2542} \BibitemShut
  {NoStop}%
\bibitem [{\citenamefont {P\'erez-Garc\'\i{}a}\ \emph
  {et~al.}(2008)\citenamefont {P\'erez-Garc\'\i{}a}, \citenamefont {Wolf},
  \citenamefont {Sanz}, \citenamefont {Verstraete},\ and\ \citenamefont
  {Cirac}}]{string_order_symmetries}%
  \BibitemOpen
  \bibfield  {author} {\bibinfo {author} {\bibfnamefont {D.}~\bibnamefont
  {P\'erez-Garc\'\i{}a}}, \bibinfo {author} {\bibfnamefont {M.~M.}\
  \bibnamefont {Wolf}}, \bibinfo {author} {\bibfnamefont {M.}~\bibnamefont
  {Sanz}}, \bibinfo {author} {\bibfnamefont {F.}~\bibnamefont {Verstraete}}, \
  and\ \bibinfo {author} {\bibfnamefont {J.~I.}\ \bibnamefont {Cirac}},\
  }\bibfield  {title} {\enquote {\bibinfo {title} {String order and symmetries
  in quantum spin lattices},}\ }\Doi {10.1103/PhysRevLett.100.167202}
  {\bibfield  {journal} {\bibinfo  {journal} {Phys. Rev. Lett.},\ }\textbf
  {\bibinfo {volume} {100}},\ \bibinfo {pages} {167202} (\bibinfo {year}
  {2008})},\ \Eprint {http://arxiv.org/abs/arXiv:0802.0447} {arXiv:0802.0447}
  \BibitemShut {NoStop}%
\bibitem [{\citenamefont {Singh}\ \emph {et~al.}(2010)\citenamefont {Singh},
  \citenamefont {Pfeifer},\ and\ \citenamefont
  {Vidal}}]{tensor_network_global_symmetry}%
  \BibitemOpen
  \bibfield  {author} {\bibinfo {author} {\bibfnamefont {Sukhwinder}\
  \bibnamefont {Singh}}, \bibinfo {author} {\bibfnamefont {Robert N.~C.}\
  \bibnamefont {Pfeifer}}, \ and\ \bibinfo {author} {\bibfnamefont {Guifr\'e}\
  \bibnamefont {Vidal}},\ }\bibfield  {title} {\enquote {\bibinfo {title}
  {Tensor network decompositions in the presence of a global symmetry},}\ }\Doi
  {10.1103/PhysRevA.82.050301} {\bibfield  {journal} {\bibinfo  {journal}
  {Phys. Rev. A},\ }\textbf {\bibinfo {volume} {82}},\ \bibinfo {pages}
  {050301} (\bibinfo {year} {2010})},\ \Eprint {http://arxiv.org/abs/0907.2994}
  {arXiv:0907.2994} \BibitemShut {NoStop}%
\bibitem [{\citenamefont {Kennedy}\ and\ \citenamefont
  {Tasaki}(1992){\natexlab{a}}}]{kt}%
  \BibitemOpen
  \bibfield  {author} {\bibinfo {author} {\bibfnamefont {Tom}\ \bibnamefont
  {Kennedy}}\ and\ \bibinfo {author} {\bibfnamefont {Hal}\ \bibnamefont
  {Tasaki}},\ }\bibfield  {title} {\enquote {\bibinfo {title} {Hidden {$Z_2
  \times Z_2$} symmetry breaking in {Haldane}-gap antiferromagnets},}\ }\Doi
  {10.1103/PhysRevB.45.304} {\bibfield  {journal} {\bibinfo  {journal} {Phys.
  Rev. B},\ }\textbf {\bibinfo {volume} {45}},\ \bibinfo {pages} {304--307}
  (\bibinfo {year} {1992}{\natexlab{a}})}\BibitemShut {NoStop}%
\bibitem [{\citenamefont {Kennedy}\ and\ \citenamefont
  {Tasaki}(1992){\natexlab{b}}}]{kt2}%
  \BibitemOpen
  \bibfield  {author} {\bibinfo {author} {\bibfnamefont {Tom}\ \bibnamefont
  {Kennedy}}\ and\ \bibinfo {author} {\bibfnamefont {Hal}\ \bibnamefont
  {Tasaki}},\ }\bibfield  {title} {\enquote {\bibinfo {title} {Hidden symmetry
  breaking and the {Haldane} phase in {$S = 1$} quantum spin chains},}\ }\Doi
  {10.1007/BF02097239} {\bibfield  {journal} {\bibinfo  {journal} {Comm. Math.
  Phys.},\ }\textbf {\bibinfo {volume} {147}},\ \bibinfo {pages} {431}
  (\bibinfo {year} {1992}{\natexlab{b}})}\BibitemShut {NoStop}%
\bibitem [{\citenamefont {Haldane}(1983){\natexlab{a}}}]{haldane1}%
  \BibitemOpen
  \bibfield  {author} {\bibinfo {author} {\bibfnamefont {F.~D.~M.}\
  \bibnamefont {Haldane}},\ }\bibfield  {title} {\enquote {\bibinfo {title}
  {Continuum dynamics of the {1-D} {Heisenberg} antiferromagnet: Identification
  with the {$\mathrm{O}(3)$} nonlinear sigma model},}\ }\Doi
  {10.1016/0375-9601(83)90631-X} {\bibfield  {journal} {\bibinfo  {journal}
  {Phys. Lett. A},\ }\textbf {\bibinfo {volume} {93}},\ \bibinfo {pages} {464}
  (\bibinfo {year} {1983}{\natexlab{a}})}\BibitemShut {NoStop}%
\bibitem [{\citenamefont {Haldane}(1983){\natexlab{b}}}]{haldane2}%
  \BibitemOpen
  \bibfield  {author} {\bibinfo {author} {\bibfnamefont {F.~D.~M.}\
  \bibnamefont {Haldane}},\ }\bibfield  {title} {\enquote {\bibinfo {title}
  {Nonlinear field theory of large-spin {Heisenberg} antiferromagnets:
  Semiclassically quantized solitons of the one-dimensional easy-axis {N\'eel}
  state},}\ }\Doi {10.1103/PhysRevLett.50.1153} {\bibfield  {journal} {\bibinfo
   {journal} {Phys. Rev. Lett.},\ }\textbf {\bibinfo {volume} {50}},\ \bibinfo
  {pages} {1153--1156} (\bibinfo {year} {1983}{\natexlab{b}})}\BibitemShut
  {NoStop}%
\bibitem [{\citenamefont {Pollmann}\ \emph {et~al.}(2012)\citenamefont
  {Pollmann}, \citenamefont {Berg}, \citenamefont {Turner},\ and\ \citenamefont
  {Oshikawa}}]{pollmann-arxiv-2009}%
  \BibitemOpen
  \bibfield  {author} {\bibinfo {author} {\bibfnamefont {Frank}\ \bibnamefont
  {Pollmann}}, \bibinfo {author} {\bibfnamefont {Erez}\ \bibnamefont {Berg}},
  \bibinfo {author} {\bibfnamefont {Ari~M.}\ \bibnamefont {Turner}}, \ and\
  \bibinfo {author} {\bibfnamefont {Masaki}\ \bibnamefont {Oshikawa}},\
  }\bibfield  {title} {\enquote {\bibinfo {title} {Symmetry protection of
  topological phases in one-dimensional quantum spin systems},}\ }\Doi
  {10.1103/PhysRevB.85.075125} {\bibfield  {journal} {\bibinfo  {journal}
  {Phys. Rev. B},\ }\textbf {\bibinfo {volume} {85}},\ \bibinfo {pages}
  {075125} (\bibinfo {year} {2012})},\ \Eprint {http://arxiv.org/abs/0909.4059}
  {arXiv:0909.4059} \BibitemShut {NoStop}%
\bibitem [{\citenamefont {Okunishi}(2011)}]{topological_disentangler}%
  \BibitemOpen
  \bibfield  {author} {\bibinfo {author} {\bibfnamefont {Kouichi}\ \bibnamefont
  {Okunishi}},\ }\bibfield  {title} {\enquote {\bibinfo {title} {Topological
  disentangler for the valence-bond-solid chain},}\ }\Doi
  {10.1103/PhysRevB.83.104411} {\bibfield  {journal} {\bibinfo  {journal}
  {Phys. Rev. B},\ }\textbf {\bibinfo {volume} {83}},\ \bibinfo {pages}
  {104411} (\bibinfo {year} {2011})},\ \Eprint {http://arxiv.org/abs/1011.3277}
  {arXiv:1011.3277} \BibitemShut {NoStop}%
\bibitem [{\citenamefont {Aharonov}\ \emph {et~al.}(1998)\citenamefont
  {Aharonov}, \citenamefont {Kitaev},\ and\ \citenamefont
  {Nisan}}]{qc_with_mixed_states}%
  \BibitemOpen
  \bibfield  {author} {\bibinfo {author} {\bibfnamefont {Dorit}\ \bibnamefont
  {Aharonov}}, \bibinfo {author} {\bibfnamefont {Alexei}\ \bibnamefont
  {Kitaev}}, \ and\ \bibinfo {author} {\bibfnamefont {Noam}\ \bibnamefont
  {Nisan}},\ }\href@noop {} {\enquote {\bibinfo {title} {Quantum circuits with
  mixed states},}\ } (\bibinfo {year} {1998}),\ \Eprint
  {http://arxiv.org/abs/quant-ph/9806029} {arXiv:quant-ph/9806029} \BibitemShut
  {NoStop}%
\bibitem [{\citenamefont {Pollmann}\ \emph {et~al.}(2010)\citenamefont
  {Pollmann}, \citenamefont {Turner}, \citenamefont {Berg},\ and\ \citenamefont
  {Oshikawa}}]{pollmann-prb-2010}%
  \BibitemOpen
  \bibfield  {author} {\bibinfo {author} {\bibfnamefont {Frank}\ \bibnamefont
  {Pollmann}}, \bibinfo {author} {\bibfnamefont {Ari~M.}\ \bibnamefont
  {Turner}}, \bibinfo {author} {\bibfnamefont {Erez}\ \bibnamefont {Berg}}, \
  and\ \bibinfo {author} {\bibfnamefont {Masaki}\ \bibnamefont {Oshikawa}},\
  }\bibfield  {title} {\enquote {\bibinfo {title} {Entanglement spectrum of a
  topological phase in one dimension},}\ }\Doi {10.1103/PhysRevB.81.064439}
  {\bibfield  {journal} {\bibinfo  {journal} {Phys. Rev. B},\ }\textbf
  {\bibinfo {volume} {81}},\ \bibinfo {pages} {064439} (\bibinfo {year}
  {2010})},\ \Eprint {http://arxiv.org/abs/arXiv:0910.1811} {arXiv:0910.1811}
  \BibitemShut {NoStop}%
\bibitem [{\citenamefont {Popp}\ \emph {et~al.}(2005)\citenamefont {Popp},
  \citenamefont {Verstraete}, \citenamefont {Martin-Delgado},\ and\
  \citenamefont {Cirac}}]{popp-le}%
  \BibitemOpen
  \bibfield  {author} {\bibinfo {author} {\bibfnamefont {M.}~\bibnamefont
  {Popp}}, \bibinfo {author} {\bibfnamefont {F.}~\bibnamefont {Verstraete}},
  \bibinfo {author} {\bibfnamefont {M.~A.}\ \bibnamefont {Martin-Delgado}}, \
  and\ \bibinfo {author} {\bibfnamefont {J.~I.}\ \bibnamefont {Cirac}},\
  }\bibfield  {title} {\enquote {\bibinfo {title} {Localizable entanglement},}\
  }\Doi {10.1103/PhysRevA.71.042306} {\bibfield  {journal} {\bibinfo  {journal}
  {Phys. Rev. A},\ }\textbf {\bibinfo {volume} {71}},\ \bibinfo {pages}
  {042306} (\bibinfo {year} {2005})},\ \Eprint
  {http://arxiv.org/abs/quant-ph/0411123} {arXiv:quant-ph/0411123} \BibitemShut
  {NoStop}%
\bibitem [{\citenamefont {Campos~Venuti}\ and\ \citenamefont
  {Roncaglia}(2005)}]{venuti-2005-le-string}%
  \BibitemOpen
  \bibfield  {author} {\bibinfo {author} {\bibfnamefont {L.}~\bibnamefont
  {Campos~Venuti}}\ and\ \bibinfo {author} {\bibfnamefont {M.}~\bibnamefont
  {Roncaglia}},\ }\bibfield  {title} {\enquote {\bibinfo {title} {Analytic
  relations between localizable entanglement and string correlations in spin
  systems},}\ }\Doi {10.1103/PhysRevLett.94.207207} {\bibfield  {journal}
  {\bibinfo  {journal} {Phys. Rev. Lett.},\ }\textbf {\bibinfo {volume} {94}},\
  \bibinfo {pages} {207207} (\bibinfo {year} {2005})},\ \Eprint
  {http://arxiv.org/abs/arXiv:cond-mat/0503021} {arXiv:cond-mat/0503021}
  \BibitemShut {NoStop}%
\bibitem [{\citenamefont {Hagiwara}\ \emph {et~al.}(1990)\citenamefont
  {Hagiwara}, \citenamefont {Katsumata}, \citenamefont {Affleck}, \citenamefont
  {Halperin},\ and\ \citenamefont {Renard}}]{haldane_emergent_edge2}%
  \BibitemOpen
  \bibfield  {author} {\bibinfo {author} {\bibfnamefont {M.}~\bibnamefont
  {Hagiwara}}, \bibinfo {author} {\bibfnamefont {K.}~\bibnamefont {Katsumata}},
  \bibinfo {author} {\bibfnamefont {Ian}\ \bibnamefont {Affleck}}, \bibinfo
  {author} {\bibfnamefont {B.~I.}\ \bibnamefont {Halperin}}, \ and\ \bibinfo
  {author} {\bibfnamefont {J.~P.}\ \bibnamefont {Renard}},\ }\bibfield  {title}
  {\enquote {\bibinfo {title} {Observation of {$S=1/2$} degrees of freedom in
  an {$S=1$} linear-chain {Heisenberg} antiferromagnet},}\ }\Doi
  {10.1103/PhysRevLett.65.3181} {\bibfield  {journal} {\bibinfo  {journal}
  {Phys. Rev. Lett.},\ }\textbf {\bibinfo {volume} {65}},\ \bibinfo {pages}
  {3181} (\bibinfo {year} {1990})}\BibitemShut {NoStop}%
\bibitem [{\citenamefont {Polizzi}\ \emph {et~al.}(1998)\citenamefont
  {Polizzi}, \citenamefont {Mila},\ and\ \citenamefont
  {S\o{}rensen}}]{haldane_emergent_edge1}%
  \BibitemOpen
  \bibfield  {author} {\bibinfo {author} {\bibfnamefont {E.}~\bibnamefont
  {Polizzi}}, \bibinfo {author} {\bibfnamefont {F.}~\bibnamefont {Mila}}, \
  and\ \bibinfo {author} {\bibfnamefont {E.~S.}\ \bibnamefont {S\o{}rensen}},\
  }\bibfield  {title} {\enquote {\bibinfo {title} {{$S=1/2$} chain-boundary
  excitations in the {Haldane} phase of one-dimensional {$S=1$} systems},}\
  }\Doi {10.1103/PhysRevB.58.2407} {\bibfield  {journal} {\bibinfo  {journal}
  {Phys. Rev. B},\ }\textbf {\bibinfo {volume} {58}},\ \bibinfo {pages} {2407}
  (\bibinfo {year} {1998})}\BibitemShut {NoStop}%
\bibitem [{\citenamefont {Michalakis}\ and\ \citenamefont
  {Pytel}(2011)}]{frustration_free_stability}%
  \BibitemOpen
  \bibfield  {author} {\bibinfo {author} {\bibfnamefont {Spyridon}\
  \bibnamefont {Michalakis}}\ and\ \bibinfo {author} {\bibfnamefont {Justyna}\
  \bibnamefont {Pytel}},\ }\href@noop {} {\enquote {\bibinfo {title} {Stability
  of frustration-free {Hamiltonians}},}\ } (\bibinfo {year} {2011}),\ \Eprint
  {http://arxiv.org/abs/1109.1588} {arXiv:1109.1588} \BibitemShut {NoStop}%
\bibitem [{\citenamefont {Raussendorf}\ \emph {et~al.}(2005)\citenamefont
  {Raussendorf}, \citenamefont {Bravyi},\ and\ \citenamefont
  {Harrington}}]{raussendorf_long_range}%
  \BibitemOpen
  \bibfield  {author} {\bibinfo {author} {\bibfnamefont {Robert}\ \bibnamefont
  {Raussendorf}}, \bibinfo {author} {\bibfnamefont {Sergey}\ \bibnamefont
  {Bravyi}}, \ and\ \bibinfo {author} {\bibfnamefont {Jim}\ \bibnamefont
  {Harrington}},\ }\bibfield  {title} {\enquote {\bibinfo {title} {Long-range
  quantum entanglement in noisy cluster states},}\ }\Doi
  {10.1103/PhysRevA.71.062313} {\bibfield  {journal} {\bibinfo  {journal}
  {Phys. Rev. A},\ }\textbf {\bibinfo {volume} {71}},\ \bibinfo {pages}
  {062313} (\bibinfo {year} {2005})},\ \Eprint
  {http://arxiv.org/abs/quant-ph/0407255} {arXiv:quant-ph/0407255} \BibitemShut
  {NoStop}%
\bibitem [{\citenamefont {Raussendorf}\ \emph {et~al.}(2003)\citenamefont
  {Raussendorf}, \citenamefont {Browne},\ and\ \citenamefont
  {Briegel}}]{raussendorf_et_al_2003}%
  \BibitemOpen
  \bibfield  {author} {\bibinfo {author} {\bibfnamefont {Robert}\ \bibnamefont
  {Raussendorf}}, \bibinfo {author} {\bibfnamefont {Daniel~E.}\ \bibnamefont
  {Browne}}, \ and\ \bibinfo {author} {\bibfnamefont {Hans~J.}\ \bibnamefont
  {Briegel}},\ }\bibfield  {title} {\enquote {\bibinfo {title}
  {Measurement-based quantum computation on cluster states},}\ }\Doi
  {10.1103/PhysRevA.68.022312} {\bibfield  {journal} {\bibinfo  {journal}
  {Phys. Rev. A},\ }\textbf {\bibinfo {volume} {68}},\ \bibinfo {pages}
  {022312} (\bibinfo {year} {2003})},\ \Eprint
  {http://arxiv.org/abs/arXiv:quant-ph/0301052} {arXiv:quant-ph/0301052}
  \BibitemShut {NoStop}%
\bibitem [{\citenamefont {Barrett}\ \emph {et~al.}(2009)\citenamefont
  {Barrett}, \citenamefont {Bartlett}, \citenamefont {Doherty}, \citenamefont
  {Jennings},\ and\ \citenamefont {Rudolph}}]{transitions_computational}%
  \BibitemOpen
  \bibfield  {author} {\bibinfo {author} {\bibfnamefont {Sean~D.}\ \bibnamefont
  {Barrett}}, \bibinfo {author} {\bibfnamefont {Stephen~D.}\ \bibnamefont
  {Bartlett}}, \bibinfo {author} {\bibfnamefont {Andrew~C.}\ \bibnamefont
  {Doherty}}, \bibinfo {author} {\bibfnamefont {David}\ \bibnamefont
  {Jennings}}, \ and\ \bibinfo {author} {\bibfnamefont {Terry}\ \bibnamefont
  {Rudolph}},\ }\bibfield  {title} {\enquote {\bibinfo {title} {Transitions in
  the computational power of thermal states for measurement-based quantum
  computation},}\ }\Doi {10.1103/PhysRevA.80.062328} {\bibfield  {journal}
  {\bibinfo  {journal} {Phys. Rev. A},\ }\textbf {\bibinfo {volume} {80}},\
  \bibinfo {pages} {062328} (\bibinfo {year} {2009})},\ \Eprint
  {http://arxiv.org/abs/0807.4797} {arXiv:0807.4797} \BibitemShut {NoStop}%
\bibitem [{\citenamefont {Klagges}\ and\ \citenamefont
  {Schmidt}(2012)}]{klagges_constraints}%
  \BibitemOpen
  \bibfield  {author} {\bibinfo {author} {\bibfnamefont {Daniel}\ \bibnamefont
  {Klagges}}\ and\ \bibinfo {author} {\bibfnamefont {Kai~Phillip}\ \bibnamefont
  {Schmidt}},\ }\bibfield  {title} {\enquote {\bibinfo {title} {Constraints on
  measurement-based quantum computation in effective cluster states},}\ }\Doi
  {10.1103/PhysRevLett.108.230508} {\bibfield  {journal} {\bibinfo  {journal}
  {Phys. Rev. Lett.},\ }\textbf {\bibinfo {volume} {108}},\ \bibinfo {pages}
  {230508} (\bibinfo {year} {2012})},\ \Eprint {http://arxiv.org/abs/1111.5945}
  {arXiv:1111.5945} \BibitemShut {NoStop}%
\bibitem [{\citenamefont {Kalis}\ \emph {et~al.}(2012)\citenamefont {Kalis},
  \citenamefont {Klagges}, \citenamefont {Or\'us},\ and\ \citenamefont
  {Schmidt}}]{cluster_fate}%
  \BibitemOpen
  \bibfield  {author} {\bibinfo {author} {\bibfnamefont {Henning}\ \bibnamefont
  {Kalis}}, \bibinfo {author} {\bibfnamefont {Daniel}\ \bibnamefont {Klagges}},
  \bibinfo {author} {\bibfnamefont {Rom\'an}\ \bibnamefont {Or\'us}}, \ and\
  \bibinfo {author} {\bibfnamefont {Kai~Phillip}\ \bibnamefont {Schmidt}},\
  }\bibfield  {title} {\enquote {\bibinfo {title} {Fate of the cluster state on
  the square lattice in a magnetic field},}\ }\Doi {10.1103/PhysRevA.86.022317}
  {\bibfield  {journal} {\bibinfo  {journal} {Phys. Rev. A},\ }\textbf
  {\bibinfo {volume} {86}},\ \bibinfo {pages} {022317} (\bibinfo {year}
  {2012})},\ \Eprint {http://arxiv.org/abs/1205.5185} {arXiv:1205.5185}
  \BibitemShut {NoStop}%
\bibitem [{\citenamefont {Chung}\ \emph {et~al.}(2009)\citenamefont {Chung},
  \citenamefont {Bartlett},\ and\ \citenamefont {Doherty}}]{chung_string}%
  \BibitemOpen
  \bibfield  {author} {\bibinfo {author} {\bibfnamefont {Thomas}\ \bibnamefont
  {Chung}}, \bibinfo {author} {\bibfnamefont {Stephen~D.}\ \bibnamefont
  {Bartlett}}, \ and\ \bibinfo {author} {\bibfnamefont {Andrew~C.}\
  \bibnamefont {Doherty}},\ }\bibfield  {title} {\enquote {\bibinfo {title}
  {Characterizing measurement-based quantum gates in quantum many-body systems
  using correlation functions},}\ }\Doi {10.1139/P08} {\bibfield  {journal}
  {\bibinfo  {journal} {Can.\ J.\ Phys.},\ }\textbf {\bibinfo {volume} {87}},\
  \bibinfo {pages} {219} (\bibinfo {year} {2009})},\ \Eprint
  {http://arxiv.org/abs/arXiv:0904.2609} {arXiv:0904.2609} \BibitemShut
  {NoStop}%
\bibitem [{\citenamefont {Or\'us}\ \emph {et~al.}(2009)\citenamefont {Or\'us},
  \citenamefont {Doherty},\ and\ \citenamefont {Vidal}}]{aqocm}%
  \BibitemOpen
  \bibfield  {author} {\bibinfo {author} {\bibfnamefont {Rom\'an}\ \bibnamefont
  {Or\'us}}, \bibinfo {author} {\bibfnamefont {Andrew~C.}\ \bibnamefont
  {Doherty}}, \ and\ \bibinfo {author} {\bibfnamefont {Guifr\'e}\ \bibnamefont
  {Vidal}},\ }\bibfield  {title} {\enquote {\bibinfo {title} {First order phase
  transition in the anisotropic quantum orbital compass model},}\ }\Doi
  {10.1103/PhysRevLett.102.077203} {\bibfield  {journal} {\bibinfo  {journal}
  {Phys. Rev. Lett.},\ }\textbf {\bibinfo {volume} {102}},\ \bibinfo {pages}
  {077203} (\bibinfo {year} {2009})},\ \Eprint {http://arxiv.org/abs/0809.4068}
  {arXiv:0809.4068} \BibitemShut {NoStop}%
\bibitem [{\citenamefont {Bravyi}\ \emph {et~al.}(2010)\citenamefont {Bravyi},
  \citenamefont {Hastings},\ and\ \citenamefont
  {Mikhalakis}}]{topological_stability}%
  \BibitemOpen
  \bibfield  {author} {\bibinfo {author} {\bibfnamefont {Sergey}\ \bibnamefont
  {Bravyi}}, \bibinfo {author} {\bibfnamefont {Matthew~B.}\ \bibnamefont
  {Hastings}}, \ and\ \bibinfo {author} {\bibfnamefont {Spyridon}\ \bibnamefont
  {Mikhalakis}},\ }\bibfield  {title} {\enquote {\bibinfo {title} {Topological
  quantum order: Stability under local perturbations},}\ }\Doi
  {10.1063/1.3490195} {\bibfield  {journal} {\bibinfo  {journal} {J. Math.
  Phys.},\ }\textbf {\bibinfo {volume} {51}},\ \bibinfo {pages} {093512}
  (\bibinfo {year} {2010})},\ \Eprint {http://arxiv.org/abs/1001.0344}
  {arXiv:1001.0344} \BibitemShut {NoStop}%
\bibitem [{\citenamefont {Bravyi}\ and\ \citenamefont
  {Hastings}(2010)}]{topological_stability_short}%
  \BibitemOpen
  \bibfield  {author} {\bibinfo {author} {\bibfnamefont {Sergey}\ \bibnamefont
  {Bravyi}}\ and\ \bibinfo {author} {\bibfnamefont {Matthew~B.}\ \bibnamefont
  {Hastings}},\ }\href@noop {} {\enquote {\bibinfo {title} {A short proof of
  stability of topological order under local perturbations},}\ } (\bibinfo
  {year} {2010}),\ \Eprint {http://arxiv.org/abs/1001.4363} {arXiv:1001.4363}
  \BibitemShut {NoStop}%
\bibitem [{\citenamefont {Aliferis}\ \emph {et~al.}(2006)\citenamefont
  {Aliferis}, \citenamefont {Gottesman},\ and\ \citenamefont
  {Preskill}}]{aliferis_long_range}%
  \BibitemOpen
  \bibfield  {author} {\bibinfo {author} {\bibfnamefont {P.}~\bibnamefont
  {Aliferis}}, \bibinfo {author} {\bibfnamefont {D.}~\bibnamefont {Gottesman}},
  \ and\ \bibinfo {author} {\bibfnamefont {J.}~\bibnamefont {Preskill}},\
  }\bibfield  {title} {\enquote {\bibinfo {title} {Quantum accuracy threshold
  for concatenated distance-3 codes},}\ }\href@noop {} {\bibfield  {journal}
  {\bibinfo  {journal} {Quant. Inf. Comput.},\ }\textbf {\bibinfo {volume}
  {6}},\ \bibinfo {pages} {97} (\bibinfo {year} {2006})},\ \Eprint
  {http://arxiv.org/abs/arXiv:quant-ph/0504218} {arXiv:quant-ph/0504218}
  \BibitemShut {NoStop}%
\bibitem [{\citenamefont {Aharonov}\ \emph {et~al.}(2006)\citenamefont
  {Aharonov}, \citenamefont {Kitaev},\ and\ \citenamefont
  {Preskill}}]{aharonov_long_range}%
  \BibitemOpen
  \bibfield  {author} {\bibinfo {author} {\bibfnamefont {Dorit}\ \bibnamefont
  {Aharonov}}, \bibinfo {author} {\bibfnamefont {Alexei}\ \bibnamefont
  {Kitaev}}, \ and\ \bibinfo {author} {\bibfnamefont {John}\ \bibnamefont
  {Preskill}},\ }\bibfield  {title} {\enquote {\bibinfo {title} {Fault-tolerant
  quantum computation with long-range correlated noise},}\ }\Doi
  {10.1103/PhysRevLett.96.050504} {\bibfield  {journal} {\bibinfo  {journal}
  {Phys. Rev. Lett.},\ }\textbf {\bibinfo {volume} {96}},\ \bibinfo {pages}
  {050504} (\bibinfo {year} {2006})},\ \Eprint
  {http://arxiv.org/abs/arXiv:quant-ph/0510231} {arXiv:quant-ph/0510231}
  \BibitemShut {NoStop}%
\bibitem [{\citenamefont {Chen}\ \emph
  {et~al.}(2011){\natexlab{b}}\citenamefont {Chen}, \citenamefont {Liu},\ and\
  \citenamefont {Wen}}]{spto_2d}%
  \BibitemOpen
  \bibfield  {author} {\bibinfo {author} {\bibfnamefont {Xie}\ \bibnamefont
  {Chen}}, \bibinfo {author} {\bibfnamefont {Zheng-Xin}\ \bibnamefont {Liu}}, \
  and\ \bibinfo {author} {\bibfnamefont {Xiao-Gang}\ \bibnamefont {Wen}},\
  }\bibfield  {title} {\enquote {\bibinfo {title} {Two-dimensional
  symmetry-protected topological orders and their protected gapless edge
  excitations},}\ }\Doi {10.1103/PhysRevB.84.235141} {\bibfield  {journal}
  {\bibinfo  {journal} {Phys. Rev. B},\ }\textbf {\bibinfo {volume} {84}},\
  \bibinfo {pages} {235141} (\bibinfo {year} {2011}{\natexlab{b}})},\ \Eprint
  {http://arxiv.org/abs/1106.4752} {arXiv:1106.4752} \BibitemShut {NoStop}%
\bibitem [{\citenamefont {Chen}\ \emph
  {et~al.}(2011){\natexlab{c}}\citenamefont {Chen}, \citenamefont {Gu},
  \citenamefont {Liu},\ and\ \citenamefont {Wen}}]{spto_higher}%
  \BibitemOpen
  \bibfield  {author} {\bibinfo {author} {\bibfnamefont {Xie}\ \bibnamefont
  {Chen}}, \bibinfo {author} {\bibfnamefont {Zheng-Cheng}\ \bibnamefont {Gu}},
  \bibinfo {author} {\bibfnamefont {Zheng-Xin}\ \bibnamefont {Liu}}, \ and\
  \bibinfo {author} {\bibfnamefont {Xiao-Gang}\ \bibnamefont {Wen}},\
  }\href@noop {} {\enquote {\bibinfo {title} {Symmetry protected topological
  orders and the cohomology class of their symmetry group},}\ } (\bibinfo
  {year} {2011}{\natexlab{c}}),\ \Eprint {http://arxiv.org/abs/1106.4772}
  {arXiv:1106.4772} \BibitemShut {NoStop}%
\bibitem [{\citenamefont {Else}\ \emph {et~al.}()\citenamefont {Else},
  \citenamefont {Bartlett},\ and\ \citenamefont {Doherty}}]{Else_forthcoming}%
  \BibitemOpen
  \bibfield  {author} {\bibinfo {author} {\bibfnamefont {Dominic~V.}\
  \bibnamefont {Else}}, \bibinfo {author} {\bibfnamefont {Stephen~D.}\
  \bibnamefont {Bartlett}}, \ and\ \bibinfo {author} {\bibfnamefont
  {Andrew~C.}\ \bibnamefont {Doherty}},\ }\href@noop {} {\bibinfo  {journal}
  {(unpublished)}}\BibitemShut {NoStop}%
\bibitem [{\citenamefont {Hastings}(2004)}]{hastings_lsm}%
  \BibitemOpen
\bibfield  {journal} {  }\bibfield  {author} {\bibinfo {author} {\bibfnamefont
  {M.~B.}\ \bibnamefont {Hastings}},\ }\bibfield  {title} {\enquote {\bibinfo
  {title} {{Lieb-Schultz-Mattis} in higher dimensions},}\ }\Doi
  {10.1103/PhysRevB.69.104431} {\bibfield  {journal} {\bibinfo  {journal}
  {Phys. Rev. B},\ }\textbf {\bibinfo {volume} {69}},\ \bibinfo {pages}
  {104431} (\bibinfo {year} {2004})}\BibitemShut {NoStop}%
\bibitem [{\citenamefont {Nachtergaele}\ and\ \citenamefont
  {Sims}(2006)}]{nachtergaele_exponential}%
  \BibitemOpen
  \bibfield  {author} {\bibinfo {author} {\bibfnamefont {Bruno}\ \bibnamefont
  {Nachtergaele}}\ and\ \bibinfo {author} {\bibfnamefont {Robert}\ \bibnamefont
  {Sims}},\ }\bibfield  {title} {\enquote {\bibinfo {title} {Lieb-{Robinson}
  bounds and the exponential clustering theorem},}\ }\Doi
  {10.1007/s00220-006-1556-1} {\bibfield  {journal} {\bibinfo  {journal} {Comm.
  Math. Phys.},\ }\textbf {\bibinfo {volume} {265}},\ \bibinfo {pages} {119}
  (\bibinfo {year} {2006})}\BibitemShut {NoStop}%
\bibitem [{\citenamefont {Hastings}(2010)}]{hastings_locality}%
  \BibitemOpen
  \bibfield  {author} {\bibinfo {author} {\bibfnamefont {M.~B.}\ \bibnamefont
  {Hastings}},\ }\href@noop {} {\enquote {\bibinfo {title} {Locality in quantum
  systems},}\ } (\bibinfo {year} {2010}),\ \Eprint
  {http://arxiv.org/abs/1008.5137} {arXiv:1008.5137} \BibitemShut {NoStop}%
\bibitem [{\citenamefont {Hastings}\ and\ \citenamefont
  {Wen}(2005)}]{hastings_wen}%
  \BibitemOpen
  \bibfield  {author} {\bibinfo {author} {\bibfnamefont {M.~B.}\ \bibnamefont
  {Hastings}}\ and\ \bibinfo {author} {\bibfnamefont {Xiao-Gang}\ \bibnamefont
  {Wen}},\ }\bibfield  {title} {\enquote {\bibinfo {title} {Quasiadiabatic
  continuation of quantum states: The stability of topological ground-state
  degeneracy and emergent gauge invariance},}\ }\Doi
  {10.1103/PhysRevB.72.045141} {\bibfield  {journal} {\bibinfo  {journal}
  {Phys. Rev. B},\ }\textbf {\bibinfo {volume} {72}},\ \bibinfo {pages}
  {045141} (\bibinfo {year} {2005})},\ \Eprint
  {http://arxiv.org/abs/cond-mat/0503554} {arXiv:cond-mat/0503554} \BibitemShut
  {NoStop}%
\bibitem [{\citenamefont {Osborne}(2007)}]{osborne_quasiadiabatic}%
  \BibitemOpen
  \bibfield  {author} {\bibinfo {author} {\bibfnamefont {Tobias~J.}\
  \bibnamefont {Osborne}},\ }\bibfield  {title} {\enquote {\bibinfo {title}
  {Simulating adiabatic evolution of gapped spin systems},}\ }\Doi
  {10.1103/PhysRevA.75.032321} {\bibfield  {journal} {\bibinfo  {journal}
  {Phys. Rev. A},\ }\textbf {\bibinfo {volume} {75}},\ \bibinfo {pages}
  {032321} (\bibinfo {year} {2007})},\ \Eprint
  {http://arxiv.org/abs/quant-ph/0601019} {arXiv:quant-ph/0601019} \BibitemShut
  {NoStop}%
\end{thebibliography}%

\end{document}